\pdfoutput=1

\RequirePackage[l2tabu, orthodox]{nag}
\documentclass[12pt]{article}
\usepackage[T1]{fontenc}
\usepackage{microtype}

\usepackage{amsmath, amssymb, dsfont, amsthm, graphicx}
\usepackage[nothm]{yk}
\usepackage{tikz}
\usepackage{float}
\usepackage{enumitem}
\usepackage{algorithm,setspace}
\usepackage[noend]{algpseudocode}
\algrenewcommand\algorithmicrequire{\textbf{Input:}}
\algrenewcommand\algorithmicensure{\textbf{Output:}}

\usepackage{amsmath}
\usepackage{amssymb}
\usepackage{amsthm}
\usepackage{mathtools}
\usepackage{thmtools}
\usepackage{mathrsfs}
\usepackage{tikz}

\usepackage{libertine}      

\usepackage{titlesec}

\usepackage{hyperref}

\titleformat*{\section}{\Large\bfseries}

\usepackage[english]{babel}
\usepackage[parfill]{parskip}
\usepackage{afterpage}
\usepackage{framed}
\usepackage{nicefrac}

{\endMakeFramed}

\newcounter{parcount}

\usepackage{graphicx}
\usepackage[labelfont=bf]{caption}
\usepackage[format=hang]{subcaption}
\usepackage{wrapfig}

\usepackage{booktabs}
\usepackage{multirow}
\usepackage{longtable}

\usepackage[sort]{natbib}
\usepackage{bibunits}

\usepackage{hyperref}
\usepackage[nameinlink]{cleveref}

\usepackage[acronym,nowarn]{glossaries}

\usepackage{listings}
\usepackage{lstbayes}
\lstdefinestyle{mystyle}{
    commentstyle=\color{OliveGreen},
    numberstyle=\tiny\color{black!60},
    stringstyle=\color{BrickRed},
    basicstyle=\ttfamily\scriptsize,
    breakatwhitespace=false,
    breaklines=true,
    captionpos=b,
    keepspaces=true,
    numbers=none,
    numbersep=5pt,
    showspaces=false,
    showstringspaces=false,
    showtabs=false,
    tabsize=2
}
\usepackage{enumitem}

\crefname{lemma}{lemma}{lemmas}
\Crefname{lemma}{Lemma}{Lemmas}
\crefname{theorem}{theorem}{theorems}
\Crefname{theorem}{Theorem}{Theorems}
\crefname{prop}{proposition}{propositions}
\Crefname{prop}{Proposition}{Propositions}

\newtheorem{theorem}{Theorem}
\newtheorem{lemma}{Lemma}
\newtheorem{assumption}{Assumption}

\newtheorem{proposition}{Proposition}
\newtheorem{remark}{Remark}

\usepackage[margin=1in]{geometry}
\usepackage{tcolorbox}
\usepackage{adjustbox}

\usepackage[defaultlines=3,all]{nowidow}

\usepackage{libertine}  
\usepackage{dsfont}

\usepackage{caption}
\usepackage{subcaption}
\usepackage{wrapfig}

\DeclareCaptionFormat{empty}{}

\everypar{\looseness=-1}

\title{Estimation and Inference for the Average \\Treatment Effect in a Score-Explained \\Heterogeneous Treatment Effect Model}

\author{
  Kevin Christian Wibisono\\
        University of Michigan, Statistics\\
  kwib@umich.edu\\
  \and
  Debarghya Mukherjee\\
          Boston University, Statistics\\
  mdeb@bu.edu\\
  \and
  Moulinath Banerjee\\
          University of Michigan, Statistics\\
  moulib@umich.edu\\
  \and
  Ya'acov Ritov\\
          University of Michigan, Statistics\\
  yritov@umich.edu\\
  }

\date{\today}

\begin{document}
\maketitle

%\begin{bibunit}[abbrvnat]
% !TEX root = attention.tex

\begin{abstract}
\noindent
In many practical situations, randomly assigning treatments to subjects is uncommon due to feasibility constraints. For example, economic aid programs and merit-based scholarships are often restricted to those meeting specific income or exam score thresholds. In these scenarios, traditional approaches to estimating treatment effects typically focus solely on observations near the cutoff point, thereby excluding a significant portion of the sample and potentially leading to information loss. Moreover, these methods generally achieve a non-parametric convergence rate. While some approaches, e.g., \citet{mukherjee2021estimation}, attempt to tackle these issues, they commonly assume that treatment effects are constant across individuals, an assumption that is often unrealistic in practice. 
In this study, we propose a differencing and matching-based estimator of the average treatment effect on the treated (ATT) in the presence of heterogeneous treatment effects, utilizing all available observations. We establish the asymptotic normality of our estimator and illustrate its effectiveness through various synthetic and real data analyses. Additionally, we demonstrate that our method yields non-parametric estimates of the conditional average treatment effect (CATE) and individual treatment effect (ITE) as a byproduct.
\end{abstract}

\textit{Keywords:} Average treatment effect on the treated, first-order differencing, heterogeneous \\treatment effect, non-random treatment allocation, residual matching.

\section{Introduction}
\label{sec:intro}
In numerous practical scenarios, the random assignment of treatments to subjects is impractical. For instance, economic initiatives targeting impoverished individuals may only be extended to those with incomes below a specified threshold. Within the medical realm, treatments are often administered to patients facing urgent needs. Academic scholarships, too, are commonly awarded based on merit exams, where applicants surpassing a predetermined cutoff score receive the scholarship. These instances, among others, emphasize the significance of investigating non-random treatment allocations. 

In many of these examples, a variable, referred to as the ``score variable" and henceforth denoted as $Q$, along with a predetermined cutoff $\tau_0$, determines the treatment allocation. In the economic initiative case, an individual's income serves as the score variable $Q$, while in the scholarship case, $Q$ corresponds to a student's merit test score. Our focus is on estimating the impact of the treatment on some response variable $Y$. In the economic initiative case, we may seek to determine whether a specific initiative benefits the have-nots, as measured by an individual's happiness level. In the scholarship case, we may be interested in assessing the effect of the scholarship on a student's future prospects, as measured by their future income. In both these cases, happiness level and future income serve as response variables, respectively.

Additionally, we often have background information $X$ on the individuals (e.g., socio-economic background, education, race, gender, and age) that may affect both the response variable $Y$ and the score variable $Q$. One way to capture the effect of $(X, Q)$ on $Y$ is via the following model:  
$$
Y_i = \alpha(X_i, Q_i) \mathds{1}_{Q_i \geq \tau_0} + X_i^\top \beta_0 + \nu_i \,,
$$
where $\nu_i$ is the unobserved error. 
Here, $\alpha(\cdot)$ denotes the individual treatment effect (ITE) which can potentially depend on the background information and score variable.

However, as is true for most real world applications, the score variable $Q$ and the unobserved error $\nu$ can be correlated through some \textit{unobserved confounders}. In the exam-based scholarship example, $\nu$ may encode students' innate abilities or intelligence that affects both the score of the merit test and their future income. In other words, $(Q, X)$ may fail to capture all factors that are essential for explaining $Y$. This is exactly what differentiates our setting with the standard treatment effect models, where either (i) the treatment is allocated randomly (also known as the \textit{randomized controlled trial} or RCT); or (ii) the observed covariates $(Q, X)$ are assumed to explain all the effects between the treatment and response variables, also known as the \emph{ignorability} or \emph{unconfoundedness} assumption (\citet{rosenbaum1983central, robins1994estimation,imbens2004nonparametric}). 
%Indeed, in the initial paper of RDD, the authors analyzed data from a national merit competition to understands its effect on ....... \KW{To do} 
%To illustrate, the allocation of treatment is determined by the score variable and a pre-determined cutoff that rules out RCT and a potential correlation between $\nu$ and $Q$ negates ignorability. 

Traditionally, researchers address this endogeneity issue, i.e., a non-zero correlation between the score variable and the unobserved error, using \emph{regression discontinuity design} (RDD) \citep{thistlethwaite1960regression}. The key idea of RDD is to localize the problem: students whose scores belong to a small vicinity around the cutoff---say within the neighborhood $[\tau_0 - h, \tau_0 + h]$ for some small $h$---are similar in terms of their abilities \citep{calonico2019regression,cattaneo2019practical}. Consequently, it is enough to compare the future income of students who barely missed the scholarship (i.e., $Q \in [\tau_0 -h, \tau_0)$ and students who barely cleared the merit exam (i.e., $Q \in [\tau_0, \tau_0 + h])$. RDD has a rich literature and has found various applications in numerous fields, such as education (\citet{jacob2004remedial,moss2006shaping,banks2012effect}), health (\citet{venkataramani2016regression,chen2018effect,christelis2020impact}), and epidemiology (\citet{mody2018estimating,basta2019evaluating,anderson2020effect}).

While being general and often fully non-parametric, the local approach described above suffers from several drawbacks. First, it \textit{only takes into account the observations within a certain neighborhood (bandwidth) around the cutoff, consequently losing information by rejecting other observations}. In situations where the treatment effect depends on the distance from the cutoff, such methods only assess the impact on individuals near the threshold, who may not be the primary focus of our analysis. A common example is the effect of a medicine on patients who are in dire need of it---these patients might be far from the eligibility threshold, but are the most relevant for analysis. Second, it yields an estimator of the treatment effect with a \textit{slow (non-parametric) rate of convergence}: when the bandwidth $h$ decreases to zero with the number of observations $n$---essential for consistent estimation of local effects---the number of effective samples, i.e., those whose scores are within the neighborhood $[\tau_0 - h, \tau_0 + h]$, is of order $nh$, which is smaller than $n$. As a partial solution, \citet{angrist2015wanna} introduced a covariate-based method that uses all available information and constructed a $\sqrt{n}$-consistent estimator of the treatment effect. 
However, their method relies on the conditional independence assumption, i.e., $\mathbb{E}(Y \mid Q, X) = \mathbb{E}(Y \mid X)$, a variant of the exogeneity assumption which typically does not hold as argued earlier. 

% In situations where the treatment effect is influenced by the distance from the cutoff, these methods only assess the impact on individuals of lesser importance to our concerns. In other words, the situation at hand may demand that the effect of treatment on those away from the cutoff is more important for analysis. A common example is the effect of a medicine on the patients who are in dire need of it. Also, from a statistical point of view, the local analysis obtains an estimator of the treatment effect with a slow (non-parametric) rate of convergence as when the bandwidth $h$ decreases to zero with the number of observations $n$, the number of effective samples (i.e. those whose scores are within $[\tau_0 - h, \tau_0 + h]$) is of order $nh$, which is smaller than $n$.

% Furthermore, this analysis introduces a lack of robustness for specific deviations from the model. 
% % \DM{This line seems to a bit vague.}

Motivated by these observations, \cite{mukherjee2021estimation} proposed an efficient $\sqrt{n}$-rate estimator of the treatment effect in the presence of endogeneity that uses all observations. Their approach assumes a homogeneous treatment effect model, where the response variable $Y$ is modeled as 
$$
Y = \alpha_0 \mathds{1}_{Q \ge \tau_0} + X_i^\top \beta_0 + \nu_i,
$$
with $\alpha_0$ being the parameter of interest. 
The key step in their method is to model the score variable $Q$ using the background information $Z$---which may or may not be the same as $X$---via the equation $Q_i = Z_i^\top \gamma_0 + \eta_i$. In this model, $(\eta, \nu)$ encodes all unobserved factors (i.e., innate abilities of students taking the merit test) and can be arbitrarily correlated. These insights lead to the following model: 
\begin{align}
\label{eq:mukherjee}
    Y_i & = \alpha_0 \mathds{1}_{Q_i \ge \tau_0} + X_i^\top \beta_0 + \ell(\eta_i) + \eps_i \\
\label{eq:mukherjee2}
    Q_i & = Z_i^\top \gamma_0 + \eta_i\,,
\end{align}
where $\ell(\eta) = \bbE[\nu \mid \eta]$ and $\eps$ is an error orthogonal to $(X, Z, \eta)$. Under the model given by Equations \eqref{eq:mukherjee} and \eqref{eq:mukherjee2}, \citet{mukherjee2021estimation} constructed an estimator of $\alpha_0$ that is $\sqrt{n}$-consistent, asymptotically normal, and semi-parametrically efficient.% \DM{The previous paragraph can be better.}

The main shortcoming of the above model is that the treatment effect, i.e., $\alpha_0$, is assumed to be constant. This oversimplification restricts the applicability of this model to real-world problems. 
For example, the effect of scholarship on a student's future income may be larger for older students or students with higher innate abilities. This crucial observation motivates a natural extension of this model that incorporates the effect of both the background information and unobserved confounders. 

\subsection{Our contributions }
In this paper, we analyze a generalized version of the endogenous treatment effect model of \citet{mukherjee2021estimation} (Equations \eqref{eq:mukherjee} and \eqref{eq:mukherjee2})
that incorporates a heterogeneous treatment effect. Specifically, we analyze the following model:  
\begin{align}
\label{eq:model_1}
    Y_i & = \alpha_0(X_i, \eta_i) \mathds{1}_{Q_i \ge \tau_0} + X_i^\top \beta_0 + \ell(\eta_i) + \eps_i \\
\label{eq:model_2} Q_i & = Z_i^\top \gamma_0 + \eta_i \,.
\end{align}
Here, $\alpha_0(\cdot)$, which we call the individual treatment effect (ITE), is a non-parametric function of both the observed background information $X$ and unobserved covariates $\eta$ (e.g., innate abilities).
%The primary difference between our model and that of \cite{mukherjee2021estimation} is that we allow $\alpha_0$ to be a non-parametric function of $(X, \eta)$, i.e., both the observed background information and unobserved covariates (e.g., students' merit). In other words, $\alpha_0(\cdot)$ encodes the heterogeneous ITE. 
It is important to note that $\alpha_0(\cdot)$ relies on the unobserved variable $\eta$, in contrast to a common assumption in the \emph{conditional average treatment effect} (CATE) literature that the dependence is solely on the observed $X$. In the scholarship example, $\alpha_0$ not only depends on the observed background information $X$, but also on the unobserved innate merit $\eta$. This allows for a situation where, for instance, a brighter student can benefit from an advanced curriculum and become more successful in the future.

One may argue that $\eta$ is not exactly the innate ability, but rather its noisy version. While this is true, it is not possible to encode innate abilities exactly as they are not observed. Nevertheless, if the score is obtained by \textit{aggregating multiple scores efficiently} (e.g., the average of multiple test scores of a student), then $\eta$ is expected to be less noisy and consequently a good proxy for the unobserved innate abilities. Furthermore, when $X = Z$ (which may be true for many practical scenarios), $\alpha_0(\cdot)$ can then be viewed as a function of $(X, Q)$ and the map $(X, Q) \leftrightarrow (X, \eta)$ is bijective. We also note that all results to be established in later sections continue to hold when $\alpha_0(\cdot)$ depends solely on $X$, which corresponds to the standard CATE.
The relationship among the variables in our model (with $X = Z$) is pictorially presented in Figure \ref{fig:relationship}. 
% The new variable $\eta$ here quantifies the unobserved confounders and $\eps$ is exogenous error and we henceforth assume $\bbE[\eps \mid X, \eta, Q] = 0$. This unobserved confounder can be correlated with both $(X, Q)$. A graphical representation is the above model is the following: 
\begin{figure}[h]
\centering
\begin{tikzpicture}
    % Nodes
    \node[circle, fill = blue, draw, label={$X$}] (X) at (-1,1) {};
    \node[circle, fill = blue, draw, label=right:{$Y$}] (Y) at (0,0) {};
    \node[circle, fill = blue, draw, label=left:{$Q$}] (Q) at (-2,0) {};
    \node[circle, fill = blue, draw, label={$\eta$}] (eta) at (-3,2) {};
    \node[circle, fill = blue, draw, label={$\epsilon$}] (eps) at (1,2) {};

    % Arrows
    \draw[->, thick] (X) -- (Q);
    \draw[->, thick] (X) -- (Y);;
    \draw[->, thick] (eps) -- (Y);
    \draw[->, thick] (Q) -- (Y);
    \draw[->, thick] (eta) to[bend left=60] (Y);
    %\draw[->, thick] (eta) -- (X);
    %\draw[->, thick] (eta) -- (X);
    \draw[->, thick] (eta) -- (Q);

\end{tikzpicture}
\caption{A graphical representation of the variables of Equations \eqref{eq:model_1} and \eqref{eq:model_2}, with $X=Z$.}
\label{fig:relationship}
\end{figure}
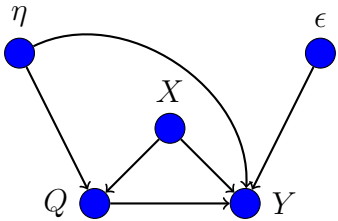

In this paper, our primary goal is to estimate the \emph{average treatment effect on the treated} (ATT), defined as
\begin{equation}
    \label{eq:param_interest}
    \theta_0 = \bbE[\alpha_0(X, \eta) \mid Q \ge \tau_0] \,.
\end{equation}

As elaborated previously, the key difficulty in estimating $\theta_0$ is precisely the unobserved covariates $\eta$. If $\eta$ were known, one could construct a consistent estimator of $\theta_0$ using the following steps: 
\begin{enumerate}
    \item On the control observations (i.e., observations for which $Q_i < \tau_0$), we have
    $$
    Y_i = X_i^\top \beta_0 + \ell(\eta_i) + \eps_i,
    $$
    which is a standard \emph{partial linear model} (PLM). Therefore, we can use any standard technique available in the literature of PLM to estimate $(\beta_0, \ell)$. 

    \item Let $\hat \beta$ and $\hat f$ respectively be the estimators of $\beta_0$ and $f$ obtained in the previous step. Consider the residuals $R_i = Y_i - X_i^\top \hat \beta - \hat \ell(\eta_i)$ for all the treatment observations (i.e., those with $Q_i \ge \tau_0$) and set
    $$
    \hat \theta = \frac{\sum_{i: Q_i \ge \tau_0} R_i}{\sum_i \mathds{1}(Q_i \ge \tau_0)} \,.
    $$
\end{enumerate}
However, we do not observe $\eta$ in practice and need to approximate it via regression residuals from Equation \eqref{eq:model_2}. This fairly complicates the analysis as the approximation error now depends on $(Z, Q)$, and hence on $X$. Therefore, an appropriate modification of the above procedure is necessary, and this is outlined below.
% Furthermore, \textit{our goal here is to construct not only a consistent estimator of $\theta_0$, but a $\sqrt{n}$-consistent and asymptotically normal one}. Observe that the aforementioned procedure does not generally produce a $\sqrt{n}$-consistent estimator of $\theta_0$, even when $\eta$ is observed. This lack of $\sqrt{n}$-consistency arises from the inclusion of $\hat{f}$ in the process. Since $f$ is a non-parametric function, the residuals incorporate this estimation error, whose order exceeds $n^{-1/2}$. As a consequence, $\hat \theta$ generally estimates $\theta_0$ at a non-parametric rate that depends on the smoothness of $f$. Hence, in order to obtain a $\sqrt{n}$-consistent estimator, it is necessary to appropriately modify Step 2.

We first obtain an approximation of $\eta$ from Equation \eqref{eq:model_2} by taking the residuals upon regressing $Q$ on $Z$. From this approximation, say $\hat \eta_i$, we broadly follow Steps 1 and 2 as mentioned above, albeit with suitable modifications. First, we use the control observations (along with these estimated $\hat \eta_i$) to estimate $\beta_0$ using a \textit{difference-based technique} (see, e.g., \citet{yatchew1997an} and \citet{wang2011difference}) for estimating the linear parameter in a PLM, which precludes the need to estimate $\ell$ explicitly. The fundamental idea of this technique is to estimate $\beta_0$ in a PLM of the form $Y = X^\top \beta_0 + \ell(\eta) + \epsilon$ by ordering the $\eta_i$ values and calculating the first-order differences of the corresponding $Y_i$ and $X_i$ values. In our scenario, a careful adaptation of this method is necessary as the $\eta_i$'s are unobserved and the ordering is made on their estimates $\hat{\eta}_i$'s.

To obtain a $\sqrt{n}$-consistent estimator of $\theta_0$, we employ a \emph{residual matching} technique inspired by the matching-based estimator studied in \citet{abadie2016matching}. Specifically, for each treatment observation $i$, we identify its \emph{nearest} control observation $c(i)$ whose $\hat \eta$ value is closest to that of $i$ (this is unlike \citet{abadie2016matching}, which matches on estimated propensity scores). We then take the difference between the responses corresponding to both indices, i.e., $Y_i - Y_{c(i)}$. As the two $\hat \eta$ values are close, the effect of the non-parametric function $f$ basically vanishes and only the linear function $(X_i - X_{c(i)})^\top \beta_0$, which is estimable at a $\sqrt{n}$-rate, and the ITE $\alpha_0(X_i, \eta_i)$ remain. Finally, we consider the differences between $Y_i - Y_{c(i)}$ and $(X_i - X_{c(i)})^\top \hat{\beta}$ and compute the average of these differences over all treatment observations $i$. The details of our method are elaborated in Section \ref{sec:meth}. 

At this point, we note that our model shares a degree of similarity with the simultaneous triangular equation framework studied in \citet{newey1999nonparametric} and \citet{pinkse2000nonparametric}. However, the existing literature on such models typically assumes a smooth link function between the outcome variable $Y$ and the covariates $(X, Z, Q)$, whereas in our setting, a natural discontinuity arises due to the deterministic assignment of treatment based on a thresholded score function. Furthermore, the inclusion of Equation \eqref{eq:model_2} may suggest that our method directly falls under the purview of the instrumental variable (IV) framework \citep{angrist1996identification}, where the main idea is to identify one or more variables (called instruments) that are correlated with the covariates and affect the outcome solely through their association with the covariates (known as exclusion restriction; see, e.g., \citet{lousdal2018introduction}). However, our model relaxes the standard exclusion restriction condition by allowing $X = Z$ or, more generally, permitting $X$ and $Z$ to share common predictors. For example, in the scholarship example, a student's high school grade can directly affect both their merit test score and future income.

We summarize our key contributions as follows: 
\begin{itemize}
    \item[(a)] We propose a $\sqrt{n}$-consistent and asymptotically normal estimator of $\theta_0$ in the presence of heterogeneous treatment effect (i.e., Equations \eqref{eq:model_1} and $\eqref{eq:model_2}$), where the effect of the treatment depends on both the observed $X$ and unobserved $\eta$. 
    \item[(b)] We demonstrate the effectiveness of our estimator through various synthetic and real data analyses. 
    \item[(c)] As a byproduct, our analysis also yields a non-parametric estimate of the ITE and CATE.
    
    % \item[(d)] \DM{Add a point for the nonparametric extension.}
\end{itemize}
\textbf{Organization of the paper:} Section \ref{sec:meth} details our methodology for estimating the ATT, followed by an extension to estimate the ITE and CATE. Section \ref{sec:theores} presents our theoretical results on the $\sqrt{n}$-consistency and asymptotic normality of the estimator. Section \ref{sec:sim} conducts simulation studies to verify our theoretical results and assess how well our method estimates the ITE and CATE. 
Section \ref{sec:rda} applies our method to real data sets, examining the impact of Islamic political rule on women's empowerment and the effect of grade-based academic probation on students' future GPA. Finally, Section \ref{sec:conc} concludes the paper and Section \ref{sec:sketch} includes proof sketches for the theoretical results.

\section{Methodology}
\label{sec:meth}
In this section, we present our methodology for constructing a $\sqrt{n}$-consistent and asymptotically normal estimator of $\theta_0$. Recall that we observe $\{(Y_i, X_i, Z_i, Q_i)\}_{1 \le i \le n}$ from Equations \eqref{eq:model_1} and \eqref{eq:model_2}. 
% We outline the methods and results in this section, deferring technical explanations for later sections. Section \ref{sec:mot} provides the rationale behind our approach. Section \ref{sec:alg} outlines our algorithm for estimating the average treatment effect on the treated (ATT). Section \ref{sec:theo} delves into the theoretical findings, while Section \ref{sec:ite-av-est} discusses the estimation of the individual treatment effect (ITE) and the asymptotic variance of our ATT estimate.
% \subsection{Motivation}
% \label{sec:mot}
% As discussed in Section \ref{sec:intro}, we are interested in estimating the average treatment effect on the treated (ATT), i.e., $z_0 = \mathbb{E}(\alpha(\eta,X)|Q \geq \tau)$, given the model
% \begin{equation}
% \label{eq:main}
% \begin{split}
%     &Y_i = \alpha(\eta_i,X_i) \mathds{1}_{Q_i \geq \tau} + X_i^\top \beta_0 + f(\eta_i) + \epsilon_i, \\
%     &Q_i = Z_i^\top \gamma_0 + \eta_i.
% \end{split}
% \end{equation}
% We assume $\mathbb{E}(\eta|Z) = \mathbb{E}(\epsilon|X,\eta) =0$, $\textrm{var}(\eta|Z) = \sigma_\eta^2$, $\textrm{var}(\epsilon|X,\eta) = \sigma_\epsilon^2$, $\gamma_0 \neq 0$, and that we have access to i.i.d. data $\{ (Y_i, X_i, Z_i, Q_i)\}_{i=1}^{N}$.
The first central idea of our method involves taking first-order differences and performing matching on the estimated residuals, i.e., $\hat{\eta}_i$'s. We draw inspiration from the literature on partially linear models (PLMs), particularly techniques for estimating a model's parametric component through first-order differences of its non-parametric counterpart (see, e.g., \cite{yatchew1997an} and \citet{wang2011difference}). 

To elaborate on this approach, consider a generic PLM of the form 
$$
Y = X^\top \beta_0 + \ell(\eta) + \epsilon \,,
$$
where we observe $(Y, X, \eta)$\footnote{We first describe our approach in the simple case where the $\eta_i$'s are known. We then discuss how this method can be adapted for the practical scenario where the $\eta_i$'s are unknown.} and assume that $\mathbb{E}(\epsilon \mid X, \eta) = 0$ and $\textrm{var}(\epsilon \mid X, \eta) = \sigma_\epsilon^2$. The goal here is to construct a $\sqrt{n}$-consistent and asymptotically normal estimator of $\beta_0$ based on an i.i.d. sample $\left\{(Y_i, X_i, \eta_i)\right\}_{i=1}^n$. The method involves two key steps. First, we sort the observations based on the $\eta_i$'s and compute the first-order differences of the sorted $Y_i$'s and $X_i$'s, denoted by $\Delta Y_i$ and $\Delta X_i$. If $(X_{(i)}, Y_{(i)})$ denotes the observation corresponding to $\eta_{(i)}$, where $\eta_{(i)}$ is the $i^{th}$ order statistic among $\{\eta_1, \dots, \eta_n\}$, then first-order differencing yields
$$
Y_{(i+1)} - Y_{(i)} = (X_{(i+1)} - X_{(i)})^\top \beta_0 + \ell(\eta_{(i+1)}) - \ell(\eta_{(i)}) + \eps_{(i+1)} - \eps_{(i)}
$$

As $\eta_{(i+1)} - \eta_{(i)}$ is small (typically of the order of $n^{-1}$), we expect $\ell(\eta_{(i+1)}) - \ell(\eta_{(i)})$ to be negligible as long as $f$ has minimal smoothness (e.g., $f$ is $\alpha$-H\"{o}lder with $\alpha > 1/2$). 
Therefore, we have
$$
Y_{(i+1)} - Y_{(i)} \approx (X_{(i+1)} - X_{(i)})^\top \beta_0 + \eps_{(i+1)} - \eps_{(i)}.
$$
Now, regressing the differences of the sorted $Y_i$'s on the sorted $X_i$'s yields a $\sqrt{n}$-consistent and asymptotically normal estimator of $\beta_0$. Refer to \cite{yatchew1997an} or \cite{wang2011difference} for more details on this approach.  
% We then have $\Delta Y = \Delta X^\top \beta_0 + \Delta f(\eta) + \Delta \epsilon$. Note that $\Delta \epsilon \approx 0$, and $\Delta f(\eta) \approx 0$ since the first-order differencing entails taking the differences between $f(\eta_k)$ and $f(\eta_l)$ with $\eta_k \approx \eta_l$. Consequently, $\Delta Y \approx \Delta X^\top \beta_0$, allowing for the estimation of $\beta_0$ via ordinary least squares (OLS) regression of $\Delta Y$ on $\Delta X$. \DM{Use a better notation for vector and matrix and elaborate.}

Let us now consider Equations \eqref{eq:model_1} and \eqref{eq:model_2}, again assuming that the $\eta_i$'s are known. Note that observations in the control group satisfy $Q_i < \tau_0$, which implies 
$$
Y_i = X_i^\top \beta_0 + \ell(\eta_i) + \epsilon_i \,.
$$ 
Therefore, we can use the first-order difference-based method described in the previous paragraphs to obtain an estimate $\hat{\beta}$ of $\beta_0$. 

We now introduce the second central idea of our method, which is residual matching. For the $i^{th}$ observation in the treatment group, we identify an observation in the control group whose $\eta$ value is closest to it. To be more specific, we define $c(i)$ as 
$$
c(i) = \argmin_{j \in \text{ control group}} |\eta_j - \eta_i| \,.
$$
Recall that the response of an observation in the treatment group satisfies 
$$
Y_i = \alpha_0(X_i, \eta_i) + X_i^\top \beta_0 + \ell(\eta_i) + \eps_i \,.
$$
Therefore, if the $i^{th}$ observation in the treatment group is matched to the $c(i)^{th}$ observation in the control group, we have
$$
Y_i - Y_{c(i)} = \alpha_0(X_i, \eta_i) + (X_i - X_{c(i)})^\top \beta_0 + \ell(\eta_i) - \ell(\eta_{c(i)}) + \eps_i - \eps_{c(i)} \,.
$$
Note that $\eta_{c(i)}$, by definition, is the closest control value to $\eta_i$. Thus, we expect $\eta_i - \eta_{c(i)}$ to be small and $\ell(\eta_i) - \ell(\eta_{c(i)})$ to be negligible under minimum smoothness assumption on $f$. Therefore, we have
\begin{equation}
    \label{eq:matching}
    Y_i - Y_{c(i)} \approx \alpha_0(X_i, \eta_i) + (X_i - X_{c(i)})^\top \beta_0 + \eps_i - \eps_{c(i)} \,. 
\end{equation}

Recall that we have obtained an estimate $\hat \beta$ of $\beta_0$ which is $\sqrt{n}$-consistent and asymptotically normal as described before. Thus, we can replace $\beta_0$ by $\hat \beta$ in Equation \eqref{eq:matching}, which yields
$$
Y_i - Y_{c(i)} \approx \alpha_0(X_i, \eta_i) + (X_i - X_{c(i)})^\top \hat \beta + \eps_i - \eps_{c(i)} \,. 
$$
Finally, taking the average of $(Y_i - Y_{c(i)}) - (X_i - X_{c(i)})^\top \hat \beta$ over all the treatment observations yields our final estimate: 
$$
\hat \theta = \frac{1}{n_T}\sum_{i: Q_i \ge \tau_0}\left\{(Y_i - Y_{c(i)}) - (X_i - X_{c(i)})^\top \hat \beta\right\} \,,
$$
% \begin{equation}
% \label{eq:mot}
%  \alpha(\eta_T, X_T) \approx Y_T - Y_C - (X_T - X_C)^\top \beta_0 \approx Y_T - Y_C - (X_T - X_C)^\top \hat{\beta}.   
% \end{equation}
% Taking the average of the right-most term in Equation \ref{eq:mot} over all observations in the treatment arm yields an estimate for the ATT. 

% We now outline our method in three key steps, again assuming the $\eta_i$'s are known.

% \textbf{Step 1: First-order differencing} 
% \begin{itemize}
%     \item Consider only observations in the control arm.
%     \item Perform first-order differencing based on $\eta$ to approximately eliminate the non-parametric component.
%     \item Run OLS to obtain $\hat{\beta}$, an estimate for $\beta_0$.
% \end{itemize}

% \textbf{Step 2: Matching}
% \begin{itemize}
%     \item Match each observation in the treatment arm to an observation in the control arm based on $\eta$.
%     \item For every matched pair $(T,C)$, calculate $Y_T - Y_C - (X_T - X_C)^\top \hat{\beta}$.
% \end{itemize}

% \textbf{Step 3: ATT estimation}
% \begin{itemize}
%     \item Take the average of $Y_T - Y_C - (X_T - X_C)^\top \hat{\beta}$ across all pairs.
%     \item The resulting average is our ATT estimate.
% \end{itemize}
where $n_T$ is the number of observations in the treatment group. 

\textit{However, in practice, the $\eta_i$'s are unknown}. In this case, they can be estimated by regressing $Q$ on $Z$ and taking the regression residuals following Equation \eqref{eq:model_2}. We now implement the method described above, replacing the $\eta_i$'s with their estimates $\hat \eta_i$'s (i.e., we perform the differencing and residual matching with respect to $\hat \eta_i$). For technical convenience, we also perform data splitting. The entire method is summarized in Algorithm \ref{alg:est}. 

Let us now briefly elaborate on the steps of Algorithm \ref{alg:est}. In Step 1, we start by partitioning the data into three roughly equal parts (i.e., when $n$ is divisible by $3$, each part consists of $n/3$ data points; if not, each of the first two parts consists of $\lfloor n/3 \rfloor$ data points and the last part consists of $n - 2\lfloor n/3 \rfloor$ data points). 
The three sets of data are denoted by $I_1, I_2, I_3$. 
Data splitting allows certain technical advantages in our quantitative analysis, though in practice, our algorithm can also be used without data splitting. 

In Step 2, we approximate the unobserved confounders $\eta$ by regressing $Q$ on $Z$ using the observations in $I_1$. We let $\hat \gamma$ be the estimated coefficient, and set $\hat \eta_i = Q_i - Z_i^\top \hat \gamma$ for all observations.
% We detail the full algorithm in Section \ref{sec:alg}.
% \subsection{Algorithm}
% \label{sec:alg}
% Our method for estimating the average treatment effect on the treated (ATT) is summarized in Algorithm \ref{alg:est}. Note that the estimation of $\gamma_0$ in Step 0 is crucial as the true $\eta_i$'s are unknown.
\begin{algorithm}[h]
\setstretch{1.05}
\caption{Estimation of the ATT $\theta_0$}
\label{alg:est}
\begin{algorithmic}[1]
\Require i.i.d. data $\{ (X_i, Y_i, Z_i, Q_i)\}_{i=1}^n$, threshold $\tau_0$
\Ensure A $\sqrt{n}$-consistent and asymptotically normal estimator of $\theta_0$ 
\State Partition $\{1,\cdots,n\}$ into $I_1$, $I_2$, $I_3$ of roughly equal sizes.
\vspace{0.05in}
\State Perform OLS of $Q_i$ against $Z_i$ for $i \in I_1$; obtain $\hat{\gamma}$ and set $\hat \eta_i = Q_i - Z_i^\top \hat \gamma$ for all $i \in I_2 \cup I_3 $\,.
\vspace{0.05in}
%\Comment{\textit{Step 0}}
% \State Let $I_2^C = \{ i \in I_2 \mid Q_i < \tau\}$ \Comment{\textit{Step 1}}
% \State \textbf{For} $i \in I_2^C$, compute $\hat{\eta}_i = Q_i- Z_i^\top \hat{\gamma}$
\State Order $\{\hat \eta_i\}_{i \in I_2^C}$ and denote by $\{\hat \eta_{(i)}\}$ be the corresponding order statistics. Denote by $\{(X_{(i)}, Y_{(i)})\}$ the induced order on $\{(X_i ,Y_i)\}_{i \in I_2^C}$.  
\vspace{0.05in}
% \State \textbf{For} $t \in \{2,\cdots,|I_2^C|\}$, compute $\Delta Y_{(t)} = Y_{(t)} - Y_{(t-1)}$ and $\Delta X_{(t)} = X_{(t)} - X_{(t-1)}$
\State Regress the first-order difference $\Delta Y_{(i)}$ on $\Delta X_{(i)}$ (where $i \in I_2^C$) and set $\hat \beta$ to be the coefficients corresponding to this regression. 
\vspace{0.05in}
% \State \textbf{For} $i \in I_3$, compute $\tilde{Y_i} = Y_i - X_i^\top \hat{\beta}$ and $\hat{\eta}_i = Q_i- Z_i^\top \hat{\gamma}$ \Comment{\textit{Step 2}}
% \State Let $I_3^C = \{ i \in I_3 \mid Q_i < \tau\}$ and $I_3^T = \{ i \in I_3 \mid Q_i \geq \tau\}$
\State Perform residual matching: for each $i \in I_3^T$, find $c(i)$ defined as: 
$$
c(i) = \argmin_{j \in I_3^C} |\hat{\eta}_i - \hat{\eta}_j| \,.
$$
% \State The estimated ATT is given by \Comment{\textit{Step 3}}
%     \begin{equation*}
%         \frac{1}{|I_3^T|} \sum_{t \in I_3^T} \left( \tilde{aY}_t - \tilde{Y}_{c_t} \right)
%     \end{equation*}
\State \Return the estimated ATT defined as follows: 
$$
\hat \theta =   \frac{1}{|I_3^T|} \sum_{i \in I_3^T} \left( (Y_i - X_i^\top \hat \beta)  - (Y_{c(i)} - X_{c(i)}^\top \hat \beta) \right) \,.
$$
\end{algorithmic}
\end{algorithm}
In Step 3, we consider $I_2^C$, the observations in $I_2$ that belong to the control group. We sort $\hat \eta_i$ for all $i \in I_2^C$ and let $\hat \eta_{(i)}$ be the $i^{th}$ order statistic among the $\hat \eta$'s in $I_2^C$. In Step 4, we take first-order differences of the $(X_{(i)}, Y_{(i)})$ observations corresponding to $\hat \eta_{(i)}$ and apply the method described previously to obtain $\hat \beta$.

Finally, we focus on $I_3$ and define $I_3^T$ and $I_3^C$ similarly (i.e., treatment observations in $I_3$ and control observations in $I_3$). In Step 5, for each treatment observation $i \in I_3^T$, we select a control observation $c(i) \in I_3^C$ such that
$$
c(i) = \argmin_{j \in I_3^C} \left|\hat \eta_i - \hat \eta_j\right| \,.
$$
In other words, we match each treatment observation $i \in I_3^T$ with a control observation $c(i) \in I_3^C$ whose $\hat \eta$ value is closest to that of $i$. Lastly, in Step 6, we calculate the differences $(Y_i - X_i^\top \hat \beta)  - (Y_{c(i)} - X_{c(i)}^\top \hat \beta)$ for all $i \in I_3^T$ and average them out to obtain $\hat \theta$. 
% Following the argument in Section \ref{sec:mot}, it is easy to see how Algorithm \ref{alg:est} intuitively yields a reasonable estimate of the ATT. Note that for each $t \in \{2,\cdots,|I_2^C|\}$, we have
% \begin{equation*}
% \begin{split}
%     Y_{(t)} - Y_{(t-1)} &= (X_{(t)} - X_{(t-1)})^\top \beta_0 + f(\eta_{(t)}) - f(\eta_{(t-1)}) + \epsilon_{(t)} - \epsilon_{(t-1)}  \\
%     &\approx  (X_{(t)} - X_{(t-1)})^\top \beta_0
% \end{split}
% \end{equation*}
% since $ \eta_i \approx \hat{\eta}_i$ and the ordering is done on the $\hat{\eta}_i$'s. Therefore, performing OLS results in a reasonable estimate $\hat{\beta}$ for $\beta_0$. Moreover, for each $t \in I_3^T$, we have
% \begin{equation*}
% \begin{split}
%     \tilde{Y}_t - \tilde{Y}_{c_t} &= \alpha(\eta_t, X_t) + (X_t - X_{c_t})^\top (\beta_0 - \hat{\beta}) + f(\eta_t) - f(\eta_{c_t}) + \epsilon_t - \epsilon_{c_t} \\
%     &\approx \alpha(\eta_t, X_t)
% \end{split}
% \end{equation*}
% since $\beta_0 \approx \hat{\beta}$, $\eta_i \approx \hat{\eta}_i$, and the matching is done on the $\hat{\eta}_i$'s. Averaging $\tilde{Y}_t - \tilde{Y}_{c_t}$ over all $t \in I_3^T$ leads to our estimate of the ATT.
\\
\begin{remark}[Cross-fitting]
\label{remark:cf}
In order to reduce the asymptotic variance of our estimator, one may use cross-fitting, i.e., implementing Algorithm \ref{alg:est} by interchanging the role of $I_1, I_2, I_3$. In particular, note that the final estimator $\hat \theta$ in Algorithm \ref{alg:est} is calculated from $I_3$. For cross-fitting, we can repeat the algorithm twice to obtain two additional versions of $\hat \theta$, one from each of $I_1$ and $I_2$. Finally, we can take the average of the three $\hat \theta$'s as our final estimate $\hat \theta$.
% we can repeat Algorithm \ref{alg:est} for all six permutations of $(I_1, I_2, I_3)$. The final estimate of the ATT is then obtained by averaging the resulting estimates from these permutations.
\end{remark}
\vspace{2mm}
\begin{remark}[Applications to CATE and fixed treatment effects]
It is easy to see that our method can also be applied when the ITE only depends on $X$ (i.e., $\alpha_0(X)$), known in the literature as the CATE, and more specifically for a fixed treatment effect model, i.e., $\alpha_0$ is a constant.
\end{remark}
\vspace{2mm}
\begin{remark}[Comparison between our method and \cite{mukherjee2021estimation}]
One key distinction between our approach and that of \cite{mukherjee2021estimation} is that our approach avoids the need to estimate $f$ when estimating the ATT and ITE. Estimating non-parametric functions is typically more computationally intensive than dealing with finite-dimensional parameters, and requires careful selection of tuning parameters (such as bandwidth or the number of basis functions). Consequently, our method is more straightforward to implement in practice as compared to the method in \cite{mukherjee2021estimation}.
\end{remark}

\subsection{Individual treatment effect estimation}
\label{sec:ite-av-est}
So far, we have presented Algorithm \ref{alg:est} for estimating the ATT $\theta_0 = \mathbb{E}(\alpha_0(X, \eta) \mid Q \geq \tau)$ and established its theoretical properties. However, in certain scenarios, there might also be interest in the ITE, denoted by $\alpha_0(X, \eta)$. For instance, in the context of the scholarship example, the award committee might wish to understand how the effect of the scholarship varies based on students' background characteristics $X$ (such as age and race) and innate abilities $\eta$. We now present Algorithm \ref{alg:ite}, a slightly modified version of Algorithm \ref{alg:est} to estimate the ITE $\alpha_0(X, \eta)$.

% However, there are at least two problems with this interpretation: (1) $\eta$ is unknown and needs to be estimated by $\hat{\eta}$; and (2) $\eta$ is a proxy for both innate abilities and noise terms. 
% From this perspective, it makes more sense to assume an ITE form that only depends on $X$, i.e., $\alpha_0(X)$, which is equivalent to a CATE. In the scholarship example, the ITE can now be interpreted as the effect of the scholarship for a student having a specific background characteristic $X$ (e.g., a 21-year-old white student who went to a private high school).

\begin{algorithm}[h]
\setstretch{1.05}
\caption{Estimation of the ITE $\alpha_0(X, \eta)$}
\label{alg:ite}
\begin{algorithmic}[1]
\Require i.i.d. data $\{ (X_i, Y_i, Z_i, Q_i)\}_{i=1}^n$, threshold $\tau_0$
\Ensure An estimate of $\alpha_0(X)$ 
\State Follow Steps 1 to 5 of Algorithm \ref{alg:ite}. 
\vspace{0.05in}
\State Estimate $\alpha(\cdot)$ by regressing $(Y_i - X_i^\top \hat \beta)  - (Y_{c(i)} - X_{c(i)}^\top \hat \beta)$ on $X_i$ and $\hat{\eta}_i$ using any non-parametric regression algorithm (e.g., basis expansion, kernel smoothing, neural networks, etc.), and call the estimator $\hat \alpha(\cdot)$. 
\State \Return $\hat \alpha(\cdot)$. 
\end{algorithmic}
\end{algorithm}

Algorithm \ref{alg:ite} can be summarized as follows. We first follow Steps 1 to 5 of Algorithm \ref{alg:est}. We then utilize any non-parametric regression algorithm (e.g., B-splines or local parametric regression) to estimate $\alpha_0(X, \eta)$. In the case where the ITE depends only on $X$ (equivalent to the standard CATE parameter), we can adjust the algorithm so that the non-parametric regression is done on $X_i$ (instead of $X_i$ and $\hat{\eta}_i$) to yield an estimator $\hat{\alpha}(X)$ of $\alpha_0(X)$.

\section{Theoretical results}
\label{sec:theores}
In this section, we establish the theoretical properties of our estimator $\hat{\theta}$, demonstrating that it is $\sqrt{n}$-consistent and asymptotically normal (CAN).
%Recall that we model the response variable $Y$ and the score variable $Q$ as follows:  
%\begin{align*}
%    Y_i & = \alpha_0(X_i, \eta_i) \mathds{1}_{Q_i \ge \tau_0} + X_i^\top \beta_0 + f(\eta_i) + \eps_i \\
%\label{eq:model_2} Q_i & = Z_i^\top \gamma_0 + \eta_i.
%\end{align*}
We denote by $d_X$ (resp. $d_Z$) the dimension of $X$ (resp. $Z$). 
Recall that our parameter of interest is the ATT $\theta_0 = \mathbb{E}(\alpha_0(X, \eta) \mid Q \geq \tau_0)$, which we estimate using $n$ i.i.d. realizations of $(X, Y, Z, Q)$.  As elaborated in Algorithm \ref{alg:est}, our method relies on splitting the data into three (almost) equal parts. 
Henceforth, we assume $n = 3\tilde{n}$ for some positive integer $\tilde{n}$ for ease of presentation and represent the entire data $\cD_n := \cI_1 \cup \cI_2 \cup \cI_3$, where each $\cI_j$ contains $\tilde{n}$ observations. 
% \DM{This $N$ notation is not very good as typically people use $N$ to denote either something larger than $n$ or sometimes the dimension, especially in the literature of factor models. But we can keep it if it is difficult to change.}
Briefly speaking, Algorithm \ref{alg:est} first uses $\cI_1$ to estimate $\gamma_0$. It then uses $\cI_2$ to estimate $\beta_0$ using the estimator of $\gamma_0$ obtained from $\cI_1$, and finally uses $\cI_3$ to estimate $\theta_0$ using the estimators of $(\gamma_0, \beta_0)$. We now state the assumptions required to show that $\hat \theta$ is $\sqrt{n}$-CAN:
\\
% \\\\
% \underline{Assumptions on errors}

% \begin{assumption}
% \label{asm:etadensity}
%     Conditional on $Q < \tau_0$, $\eta$ has a bounded and continuous density function.
% \end{assumption}

%\underline{Assumptions on data distribution}

% \begin{assumption}
% \label{asm:momentofx}
%     For any $\delta \in \mathbb{R}^{d_Z}$, the second and third absolute moments of $X$ are uniformly bounded conditional on $\eta - Z^\top \delta$ and $Q < \tau_0$. 
% \end{assumption}

% \begin{assumption}[Covariates]
% \label{asm:momentofxz}
% Assume that $(X, Z)$ are compactly supported and have finite six moments. Furthermore, conditional on $Q < \tau_0$, $(X,Z)$ has a bounded continuous density function. 
% \end{assumption}
\begin{assumption}[Covariates]
\label{asm:momentofxz}
$(X, Z)$ are compactly supported, and $0 \in \mathrm{supp}(Z)$. 
%Furthermore, conditional on $Q < \tau_0$, $(X,Z)$ has a continuous density function. 
\end{assumption}
\vspace{2mm}
% \begin{assumption}[Errors]
% \label{asm:errormeanvar}
% The error $\eps$ satisfies $\mathbb{E}(\epsilon \mid X, \eta) = 0$, $\mathrm{var}(\epsilon \mid X, \eta) = \sigma_\epsilon^2$, and $\mathbb{E}(|\epsilon|^3 \mid X, \eta)$ is finite. The error $\eta$ satisfies $\bbE(\eta \mid Z) = 0$, $\mathrm{var}(\eta \mid Z) = \sigma_\eta^2$ and $\eta$ has a bounded and continuous density function conditional on $Q < \tau_0$. Furthermore, $\eta$ is assumed to be compactly supported (in other words, the score $Q$ is compactly supported). 

\begin{assumption}[Errors]
\label{asm:errormeanvar}
The error $\eps$ satisfies $\mathbb{E}(\epsilon \mid X, \eta) = 0$, $\mathrm{var}(\epsilon \mid X, \eta) = \sigma_\epsilon^2$, and $\mathbb{E}(|\epsilon|^3 \mid X, \eta)$ is finite. The error $\eta$ satisfies $\bbE(\eta \mid Z) = 0$ and $\mathrm{var}(\eta \mid Z) = \sigma_\eta^2$. Furthermore, $\eta$ (equivalently $Q$) is compactly supported.
%and has a continuous density function conditional on $Q < \tau_0$.
\end{assumption}

For notational simplicity, we define the following conditional mean, variance and covariance functions of the covariates given the error and the control indicator: 
\begin{align*}
g_\delta(b) & = \mathbb{E}(X \mid \eta - Z^\top \delta= b, Q < \tau_0) \,, \\
q_\delta(b) & = \mathbb{E}(Z \mid \eta - Z^\top\delta= b  , Q < \tau_0) \,, \\
v_\delta(b) & = \mathrm{var}(X \mid \eta - Z^\top\delta= b , Q < \tau_0) \,, \\
w_\delta(b) & = \mathrm{var}(Z\mid \eta - Z^\top\delta= b , Q < \tau_0) \,, \\
k_\delta(b) & = \mathrm{cov}(X,Z\mid \eta - Z^\top\delta= b , Q < \tau_0) \,.
\end{align*}

% \begin{assumption}[Smoothness of the conditional functions]
% \label{asm:smoothness_conditional}
% For any $\delta \in \mathbb{R}^{d_Z}$, assume that $g_\delta, q_\delta, v_\delta$ and $w_\delta$ are uniformly bounded over $\delta$ and Lipschitz continuous over $\delta$. 
% \end{assumption}

\begin{assumption}[Smoothness of the conditional functions]
\label{asm:smoothness_conditional}
For any $\delta \in \mathbb{R}^{d_Z}$, we have $g_\delta(b), q_\delta(b), v_\delta(b)$ and $w_\delta(b)$, and $k_\delta(b)$ are Lipschitz continuous w.r.t. $b \in \mathrm{supp}(\eta - Z^\top \delta)$. Furthermore, for any sequence $\{\delta_n\}_{n \geq 1}$ that converges to $0$ and $b \in \mathrm{supp}(\eta)$, we have $g_{\delta_n}(b) \to g_{0}(b)$, $q_{\delta_n}(b) \to q_{0}(b)$, $v_{\delta_n}(b) \to v_{0}(b)$, $w_{\delta_n}(b) \to w_{0}(b)$, and $k_{\delta_n}(b) \to k_{0}(b)$.
% for any $b \in \mathbb{R}$, $g_\delta(b)$, $q_\delta(b)$, $v_\delta(b)$ and $w_\delta(b)$ are continuous w.r.t. $\delta$ such that $b \in \textrm{supp}(\eta - Z^\top \delta)$. \DM{If we think about a sequence of $\delta_n$ converges to $\delta$, then basically we need $b \in \cap_{n \in \bbN} \supp(\eta - Z^\top \delta_n)$. So for any $b$ in the intersection, we have $w_{\delta_n}(b) \to w_{\delta}(b)$.}
\end{assumption}
\vspace{2mm}
%Define $g_\delta(b) = \mathbb{E}(X \mid \eta = b + Z^\top \delta , Q < \tau_0)$ and $v_\delta(b) = \mathrm{var}(X \mid \eta = b + Z^\top \delta, Q < \tau_0)$. For any $\delta \in \mathbb{R}^{d_Z}$, assume that $g_\delta$ and $v_\delta$ are uniformly bounded over $\delta$ and Lipschitz continuous over $\delta$. \DM{Uniformly on $b$ right?}
%Define $q_\delta(b) = \mathbb{E}(Z \mid \eta = b + Z^\top \delta , Q < \tau_0)$ and $w_\delta(b) = \mathrm{var}(Z\mid \eta = b + Z^\top \delta, Q < \tau_0)$. For any $\delta \in \mathbb{R}^{d_Z}$, assume that $q_\delta$ and $w_\delta$ are also uniformly bounded over $\delta$ and Lipschitz continuous over $\delta$. 
% \leavevmode
%     \begin{itemize}
%         \item $\mathbb{E}(\epsilon \mid X, \eta) = 0$, $\mathrm{var}(\epsilon \mid X, \eta) = \sigma_\epsilon^2$, and $\mathbb{E}(|\epsilon|^3 \mid X, \eta)$ is finite.
%         \item $\eta$ is independent of $Z$, $\mathbb{E}(\eta) = 0$, $\mathrm{var}(\eta) = \sigma_\eta^2$, and $\mathbb{E}(|\eta|^3)$ is finite.
%     \end{itemize}
% \end{assumption}

% \begin{assumption}
% \label{asm:gdelta}
%     Define $g_\delta(b) = \mathbb{E}(X \mid \eta - Z^\top \delta= b , Q < \tau_0)$. For any $\delta \in \mathbb{R}^{d_Z}$, $g_\delta$ is uniformly bounded and Lipschitz continuous.
% \end{assumption}

\begin{assumption}
\label{asm:eigenvalue}
    For a large enough $\tilde{n}$, we have $\mathbb{E}\left( 1/\left(\lambda_{\min}\left( \frac{\tilde{Z}^\top \tilde{Z}}{\tilde{n}} \right)\right)^6 \right)$ is finite, where $\lambda_{\min}(\Gamma)$ is the smallest eigenvalue of $\Gamma$ and $\tilde{Z} = (Z_1; Z_2; \cdots; Z_{\tilde{n}})^\top \in \mathbb{R}^{\tilde{n} \times d_Z}$.
\end{assumption}
\vspace{2mm}
% \DM{MB comment -- seems to be related to concentration for the minimum eigenvalue of a random covariance matrix, so could potentially be simplified.}
\begin{assumption}[Model parameters]
\label{asm:fbounded}
    The non-linear function $\ell(\cdot)$ in Equation \eqref{eq:model_1} is Lipschitz continuous and has a bounded second derivative. Also, $\mathbb{E}(\left| \alpha_0(X, \eta)\right|^3 \mid Q \geq \tau_0)$ is finite.
\end{assumption} 
%\underline{Assumptions for CAN of $\theta_0$}
% \begin{assumption}[Density ratio]
% \label{asm:abadie}
% For any fixed $\delta \in \mathbb{R}^{d_Z}$ and a real number $b$, let $f_{0,\delta}(b)$ and $f_{1,\delta}(b)$ denote the density of $\eta - Z^\top \delta$ at $b$ conditional on $Q < \tau_0$ and $Q \ge \tau_0$, respectively. Then, $f_{0,\delta}(b)$ and $f_{1,\delta}(b)$ are continuous w.r.t. $b$. Also, $f_{1,\delta}(b)/f_{0,\delta}(b)$ is uniformly bounded and uniformly bounded away from zero (in $\delta$). Moreover, for any fixed $b$, $f_{0,\delta}(b)$ and $f_{1,\delta}(b)$ are continuous w.r.t. $\delta$.
% %\KW{Can we simplify the last part of the assumption?}
% \end{assumption} 
\vspace{2mm}
\begin{assumption}[Density ratio]
\label{asm:abadie}
For $\delta \in \mathbb{R}^{d_Z}$ and $b \in \mathbb{R}$, let $f_{0,\delta}(b)$ and $f_{1,\delta}(b)$ denote the density of $\eta - Z^\top \delta$ at $b$ conditional on $Q < \tau_0$ and $Q \ge \tau_0$, respectively. Then, for any $\delta \in \mathbb{R}^{d_Z}$ and $b \in \mathrm{supp}(\eta - Z^\top \delta)$, $f_{1,\delta}(b)$, $f_{0,\delta}(b)$, and $f_{1,\delta}(b)/f_{0,\delta}(b)$ are uniformly bounded and uniformly bounded away from zero. Furthermore, for any sequence $\{\delta_n\}_{n \geq 1}$ that converges to $0$ and $b \in \mathrm{supp}(\eta)$, we have $f_{0,{\delta_{n}}}(b) \to f_{0,0}(b)$ and $f_{1,{\delta_{n}}}(b) \to f_{1,0}(b)$.
% $f_{0,\delta}(b)$ and $f_{1,\delta}(b)$ are Lipschitz continuous w.r.t. $b \in \textrm{supp}(\eta - Z^\top \delta)$; moreover, $f_{1,\delta}(b)/f_{0,\delta}(b)$ is uniformly bounded and uniformly bounded away from zero. Furthermore, for any $b \in \mathbb{R}$, $f_{0,\delta}(b)$ and $f_{1,\delta}(b)$ are continuous 
% w.r.t. $\delta$ such that $b \in \textrm{supp}(\eta - Z^\top \delta)$. \DM{Change it analogously.0}
\end{assumption}

The compactness of $X$ and $Z$ in Assumption \ref{asm:momentofxz} is made for technical convenience and can be relaxed with careful truncation arguments. However, in practical scenarios, $X$ and $Z$ are naturally bounded or can be made so through proper scaling. Moreover, the assumption that $0 \in \mathrm{supp}(Z)$ ensures that $\textrm{supp}(\eta)$ is contained in $\textrm{supp}(\eta - Z^\top \delta)$, which is crucial for our proofs. Assumption \ref{asm:errormeanvar} is standard in the non-parametric regression literature, where we posit homoskedasticity for the errors $\eps$ and $\eta$. This can be relaxed by assuming $\textrm{var}(\epsilon \mid X, \eta)$ and $\textrm{var}(\eta \mid Z)$ to be uniformly bounded and bounded away from zero. Furthermore, the compactness of $\eta$ is made for technical convenience and can be relaxed through careful truncation arguments.

% The requirement that the error density is bounded and continuous on the control group is mild, allowing for a variety of distributions (e.g., Normal and Student's \textit{t}). Furthermore, the assumption of compactness for $\eta$ is made for technical convenience and can be relaxed through careful truncation arguments.

The continuity of $g_\delta$, $q_\delta$, $v_\delta$, $w_\delta$, and $k_\delta$ w.r.t. $\delta$ in Assumption \ref{asm:smoothness_conditional} are satisfied as long as the density function of $(X, \eta - Z^\top \delta)$ (resp. $(Z, \eta - Z^\top \delta)$) are continuous w.r.t. $\delta$ for any fixed $X$ (resp. $Z$). Assumption \ref{asm:eigenvalue} is a technical assumption to ensure the Lindeberg's condition for the martingale central limit theorem \citep{billingsley1995probability} is satisfied. Furthermore, Assumption \ref{asm:fbounded} sets a minimal requirement on the smoothness of the function $\ell$ and the boundedness of the third moment of the heterogeneous treatment effect $\alpha_0(X, \eta)$.

% distributions of $(X, \eta - Z^\top \delta)$ and $(Z, \eta - Z^\top \delta)$ are continuous with respect to $\delta$, conditionally on $Q < \tau_0$. Note that continuity here refers to the continuity of the density function of $(X, \eta - Z^\top \delta)$ (resp. $(Z, \eta - Z^\top \delta)$) for any fixed $X$ (resp. $Z$). 
% Assumption \ref{asm:fbounded} sets a minimal requirement on the smoothness of the function $f$ and assumes that the heterogeneous treatment effect $\alpha_0(X, \eta)$ has a finite third moment. 

Finally, Assumption \ref{asm:abadie} is standard for the density ratio of control and treatment observations (see, e.g., Assumption 2 in \cite{abadie2016matching}).
To grasp this, consider the case where $\delta = 0$. Here, $f_{1,0}(b)/f_{0,0}(b)$ corresponds to the density ratio of $\eta$ for treatment and control observations. Our method's core idea is to match treatment and control observations with respect to $\eta$. If $\mathbb{P}(\eta \in S \mid Q \geq \tau_0)/\mathbb{P}(\eta \in S \mid Q < \tau_0) = \infty$ for some set $S \subseteq \mathbb{R}$, matching becomes infeasible since in $S$, $\eta$'s in the treatment group cannot be matched with $\eta$'s in the control group. Here, the need for the density ratio of $\eta - Z^\top \delta$ to be bounded for any fixed $\delta$ arises because we do not directly observe $\eta$, but instead approximate it using $\hat \eta = \eta - Z^\top (\hat \gamma - \gamma_0)$.

It is easy to see that $\hat{\gamma}$ is $\sqrt{{n}}$-CAN. Our first proposition establishes that the estimator of $\beta_0$ obtained in Step 4 of Algorithm \ref{alg:est} is $\sqrt{{n}}$-CAN.
\\
\begin{proposition}[$\hat{\beta}$ is $\sqrt{{n}}$-CAN]
\label{prop:beta}
    Consider the estimator $\hat{\beta}$ of $\beta_0$ following the first four steps of Algorithm \ref{alg:est}. Under Assumptions \ref{asm:momentofxz} to \ref{asm:fbounded}, the estimator is $\sqrt{{n}}$-CAN:
    $$\sqrt{\tilde{n}}(\hat{\beta} - \beta_0) \overset{d}{\longrightarrow} \mathcal{N}\left(0, 
    \Sigma_\beta \right),$$
    where the explicit value of $\Sigma_\beta$ can be found in Equation \eqref{eq:betahat_var} of Appendix \ref{app:proof-prop}.
% \begin{itemize}[leftmargin=*]
%     \item $\zeta = 4 (\bbP(Q < \tau_0))^2 \mathbb{E}\left(f'(\eta) u_0 w_0^\top \mid Q < \tau_0 \right) \Sigma_\gamma \mathbb{E}\left(f'(\eta) w_0 u_0^\top \mid Q < \tau_0 \right) + 6 \bbP(Q < \tau_0) \sigma_\epsilon^2 \Sigma_u$,
%     \item $\Sigma_u = \mathbb{E}_\eta(\textrm{var}(X \mid {\eta}, Q < \tau_0) \mid Q < \tau_0)$,
%     \item $u_0 = X - \mathbb{E}(X \mid \eta, Q < \tau_0)$, and
%     \item $w_0 = Z - \mathbb{E}(Z \mid \eta, Q < \tau_0)$.
% \end{itemize}    
\end{proposition}
% \DM{I think it is better to just write $\sqrt{N}(\hat{\beta} - \beta_0) \overset{d}{\longrightarrow} \mathcal{N}\left(0, \Sigma_\beta \right)$, and mention that the explicit value of $\Sigma_\beta$ can be found in the proof.}

We elaborate the key steps of the proof in Section \ref{sec:sketch} and defer the complete proof to Appendix \ref{app:proof-prop}.
\\
\begin{remark}[$\sqrt{\tilde{n}}$ versus $\sqrt{n}$]
    In Proposition \ref{prop:beta}, we present the asymptotic normality of $\sqrt{\tilde{n}}(\hat \beta - \beta_0)$, where, $\tilde{n} = n/3$ by definition. Therefore, it is immediate that $$\sqrt{n}(\hat \beta - \beta_0) \overset{d}{\longrightarrow} \mathcal{N}\left(0, 
    3\Sigma_\beta \right).$$
    Note that the factor of $3$ can be removed by cross-fitting. 
    % \KW{If we don't split the data equally, how do we define $N$? If $N \approx n/3$, then $N/n$ necessarily approaches $1/3$? Are we saying that if we split the data into 3 parts of size $N_1$, $N_2$ and $N_3$, then asymptotic normality still holds as long as $N_i/n \rightarrow c_i$ where $c_i > 0$ for each $i$?}
    % Although in the above proposition, we present the asymptotic normality of $\sqrt{N}(\hat \beta - \beta_0)$, it is also immediate that $\sqrt{n}(\hat \beta - \beta_0)$ is also asymptotically normal as long as $N/n \to c$ for some $c \in (0, \infty)$. Typically, $N/n \to 1/3$ in our algorithm. \DM{Mention cross fitting to improve the variance.} \KW{Actually, it is already mentioned in Remark 1.}
\end{remark}
% \subsection{Consistenty and asymptotic normality of $\hat{\theta}$}
% \label{sec:thetahatest}

Upon obtaining a $\sqrt{{n}}$-CAN estimator of $\beta_0$, we are now ready to present our main result 
pertaining to the $\sqrt{{n}}$-consistency and asymptotic normality of the ATT estimator $\hat \theta$, which is obtained via matching each treatment observation with its nearest control observation in terms of $\hat{\eta}$.
\\
%We require the following additional assumptions.

% \begin{assumption}
% \label{asm:abadie}
% For any fixed $\delta \in \mathbb{R}^{d_Z}$ and a real number $b$, let $f_{0,\delta}(b)$ and $f_{1,\delta}(b)$ denote the density of $\eta - Z^\top \delta$ at $b$ conditional on $Q < \tau_0$ and $Q \ge \tau_0$, respectively. Then, $f_{0,\delta}(b)$ and $f_{1,\delta}(b)$ are continuous w.r.t. $b$ and share a common support. Also, $f_{1,\delta}(b)/f_{0,\delta}(b)$ is uniformly bounded and uniformly bounded away from zero. Moreover, for any fixed $b$, $f_{0,\delta}(b)$ and $f_{1,\delta}(b)$ are continuous w.r.t. $\delta$. 
% %\KW{Can we simplify the last part of the assumption?}
% \end{assumption}

% \begin{assumption}
% \label{asm:abadie2}
% $X$ is a bounded random variable. \KW{Is this realistic? Can we relax this?} \DM{I think this is fine; we will just write that for technical simplicity, we assume $X$ is bounded, otherwise we use a truncation argument.}
% \end{assumption}

% \begin{assumption}
% \label{asm:abadie3}
% $\mathbb{E}(\left| \alpha_0(X, \eta)\right|^3 \mid Q \geq \tau_0)$ is finite.
% \end{assumption}

% \begin{remark}
%     Assumptions \ref{asm:abadie} and \ref{asm:abadie2} are imposed so that we can apply results in \citet{abadie2016matching}. Assumption \ref{asm:abadie3} is one of the sufficient conditions for the Lyapunov's condition in the application of the martingale central limit theorem \citep{billingsley1995probability} to be satisfied.
% \end{remark}

% We have the following theorem:
\begin{theorem}[$\hat{\theta}$ is $\sqrt{{n}}$-CAN]
\label{thm:main}
Consider the estimator $\hat{\theta}$ of $\theta_0$ summarized in Algorithm \ref{alg:est}. Under Assumptions \ref{asm:errormeanvar} to \ref{asm:abadie}, the estimator is $\sqrt{{n}}$-CAN:
    $$\sqrt{\tilde{n}}(\hat{\theta} - \theta_0) \xrightarrow{d} \mathcal{N}(0, \sigma_\theta^2),$$
where the explicit form of $\sigma_\theta^2$ can be found in Equation \eqref{eq:thetavar} of Appendix \ref{app:proof-thm}.
% \begin{itemize}[leftmargin=*]
%     \item $\sigma^2_\theta = A_5^\top \Sigma_\gamma A_5 + 6\mathbb{P}(Q < \tau_0) \sigma_\epsilon^2 A_3^\top \Sigma_u A_3 + B + C$,
%     \item $A_3 = \frac{1}{2\mathbb{P}(Q < \tau_0)} \Sigma_u^{-1} \left( \mathbb{E}\left( \frac{f_1({\eta})}{f_0({\eta})} X \mid Q < \tau_0 \right) - \mathbb{E}\left(X \mid Q \geq \tau_0\right) \right)$,
%     \item $A_5 =  2\mathbb{P}(Q < \tau_0) \mathbb{E}(f'(\eta)w_0 u_0^\top \mid Q < \tau_0) A_3 + \mathbb{E}\left(Zf'({\eta}) \mid Q \geq \tau_0\right) - \mathbb{E}\left( \frac{f_1({\eta})}{f_0({\eta})} Z f'({\eta}) \mid Q < \tau_0 \right)$,
%     \item $B = \frac{\textrm{var}(\alpha_0(X,\eta) \mid Q \geq \tau_0)}{\bbP(Q \geq \tau_0)}$,
%     \item $C =   \sigma_\epsilon^2 \left( \frac{2}{\bbP(Q \geq \tau_0)} + \frac{3}{2 \bbP(Q < \tau_0)}  \mathbb{E}\left( \left(\frac{f_1(\eta)}{f_0(\eta)} \right)^2 \mid Q < \tau_0 \right) \right)$,
%     \item $\Sigma_u$, $u_0$ and $w_0$ are defined similarly as in Proposition \ref{prop:beta}, and
%     \item $f_0(\eta)$ and $f_1(\eta)$ denote the densities of $\eta$ conditional on $Q < \tau_0$ and $Q \geq \tau_0$, respectively.
% \end{itemize}
\end{theorem}
% \DM{I think it is better to write "The explicit form of $\sigma_\theta^2$ can be found in the proof, see Equation \eqref{} for details.}

Theorem \ref{thm:main} is the main result of this paper, which proves that the ATT can be estimated at a parametric rate despite a non-parametric correlation between the unobserved errors in Equations \eqref{eq:model_1} and \eqref{eq:model_2}. The key steps of the proof are presented in Section \ref{sec:sketch}, while the detailed proof can be found in Appendix \ref{app:proof-thm}.
\\
\begin{remark}[Cross-fitting]
\label{remark:cf2}
    We can gain efficiency (in terms of asymptotic variance) by performing cross-fitting as described in Remark \ref{remark:cf}. In particular, if $\hat{\theta}_{\textrm{cf}}$ denotes our cross-fitting estimator, we have $$\sqrt{n}(\hat{\theta}_{\textrm{cf}} - \theta_0) \xrightarrow{d} \mathcal{N}(0, \sigma_\theta^2).$$
\end{remark}

\subsection{Estimation of asymptotic variance of ATT}
\label{sec:bootstr}
In our main theorem (Theorem \ref{thm:main}), we have established that the estimator obtained in Algorithm \ref{alg:est} is a $\sqrt{n}$-CAN estimator of the ATT $\theta_0$, with an asymptotic variance denoted by $\sigma_\theta^2$. Therefore, if we have a consistent estimator of $\sigma_\theta$ (say $\hat{\sigma}_\theta$), then by Slutsky's theorem, 
$$
\frac{1}{{\hat \sigma}_\theta}\sqrt{\tilde{n}}(\hat \theta - \theta_0) \overset{d}{\longrightarrow} \cN(0, 1) \,.
$$
This allows us to perform statistical inference. For example, one might be interested in testing $H_0: \theta_0 = 0$ vs. $H_1: \theta_0 \neq 0$, as it quantifies whether there is any treatment effect on the treated individuals.

% As a consequence, if we are interested in testing $H_0: \theta_0 = c$ vs. $H_1: \theta_0 \neq c$ for any $c \in \reals$, we can have an asymptotically optimal confidence interval\footnote{$c = 0$ may be of particular interest, as it quantifies whether there is any treatment effect on the treated.}: 
% $$
% \widehat{\rm CI}_{1- \alpha} = \left[\hat \theta - \frac{\hat{\sigma}_\theta}{\sqrt{\tilde{n}}}z_{\alpha/2}, \ \ \hat \theta + \frac{\hat{\sigma}_\theta}{\sqrt{\tilde{n}}}z_{\alpha/2}\right],
% $$
% and reject (resp. accept) $H_0$ at level $\alpha$ depending on whether $c$ lies within the above confidence interval. 

% Nevertheless, we show in Appendix \ref{app:est-sep} that it is not impossible to obtain a consistent estimator \KW{We only have consistency for most parts of the estimation methods, not all.} \DM{Algorithm \ref{alg:estex} is not solving any purpose here so maybe we put the entire estimation method in the appendix.} 

However, $\hat{\sigma}_\theta$ is notoriously difficult to estimate as it involves numerous nuisance parameters. A more practical approach is to use techniques like bootstrapping. 
% As a practical approach to estimating $\sigma_\theta^2$, one may employ approximation techniques such as bootstrapping. 
In a standard $n$-out-of-$n$ bootstrapping procedure, $n$ observations are drawn from the full data set of $n$ observations with replacement, and Algorithm \ref{alg:est} is applied to these sampled observations. This process is repeated $B$ times (i.e., sample $n$ observations with replacement and estimate $\theta_0$ based on these sampled observations) to yield $\{\hat \theta_b\}_{b=1}^B$. 
Based on these bootstrapped estimators, we estimate $\sigma_\theta^2$ as the variance of these estimators (scaled by $\tilde{n}$): 
% for each sample $b \in \{1,2,\cdots,B\}$, yielding $\hat{\theta}_b$. 
% An estimate of $\sigma_\theta^2$ is then calculated as follows:
\begin{equation*}
    \hat{\sigma}_\theta^2 = \tilde{n} \left(\frac{\sum_{b=1}^B \left(\hat{\theta}_b - \bar \theta_B\right)^2}{B-1}\right),
\end{equation*}
where $\bar \theta_B = \frac{1}{B}\left(\hat{\theta}_1 + \cdots + \hat{\theta}_B\right)$ and $\tilde{n} = n/3$. We empirically show in Section \ref{sec:sim2} that the bootstrap estimator is consistent, and leave its proof for future work.

\section{Simulation studies}
\label{sec:sim}
In this section, we conduct three simulation studies. The first simulation seeks to numerically illustrate the consistency and asymptotic normality of the estimator $\hat{\theta}$ for the ATT $\theta_0$, as outlined in Algorithm \ref{alg:est}. Additionally, it checks whether the asymptotic variance of the estimator matches the formula from Equation \eqref{eq:thetavar} and whether cross-fitting results in an efficiency gain as explained in Remarks \ref{remark:cf} and \ref{remark:cf2}. The second simulation illustrates that bootstrapping, as described in Section \ref{sec:bootstr}, provides a good approximation of the true asymptotic variance. Lastly, the third simulation study evaluates the capability of Algorithm \ref{alg:ite} to estimate the ITE.

\subsection{Data generating process}
For the first and second simulations, we consider a simple data-generating process:
\begin{align}
\label{eq:dgp-sim-1}
    Y &= (X_1^2 + X_2 X_3 + \eta^2) \mathds{1}_{Q \geq 0} + X_1 + X_3 + \eta/2 + \epsilon \\
\label{eq:dgp-sim-2} Q &= X_4 + \eta \,.
\end{align}

where $X_1, \cdots, X_4 \sim \mathcal{N}(0,1)$ i.i.d., $\epsilon \sim \mathcal{N}(0, 0.5)$, $\eta \sim \mathcal{U}(-1,1)$. Note that Equations \eqref{eq:dgp-sim-1} and \eqref{eq:dgp-sim-2} are a special instance of Equations \eqref{eq:model_1} and \eqref{eq:model_2} with $Z = (X_1, X_2, X_3, X_4)$, $X = (X_1, X_2, X_3)$, $\alpha_0(X,\eta) = X_1^2 + X_2 X_3 + \eta^2$, $\ell(\eta) = \eta/2$, $\tau_0 = 0$, $\beta_0 = (1,0,1)^\top$, and $\gamma_0 = (0,0,0,1)^\top$. For this particular data-generating process, $\sigma_\theta^2 \approx 11.455$. Also, the true ATT $=\theta_0 = \mathbb{E}(X_1^2 + X_2 X_3 + \eta^2 \mid X_4 + \eta \geq 0) \approx 1.333.$

\subsection{Simulation 1: Estimation of the ATT}
\label{sec:sim1}
By Theorem \ref{thm:main}, $\sqrt{\tilde{n}}(\hat{\theta} - \theta_0) \xrightarrow{d} \mathcal{N}(0, 11.455)$. To numerically verify this, we generate $1,000$ Monte-Carlo iterations following the data-generating process in Equations \eqref{eq:dgp-sim-1} and \eqref{eq:dgp-sim-2}, each of size $n = 12,000$.
This means that each split is of size $\tilde{n} = 4,000$. For each iteration $1 \leq k \leq 1,000$, we compute $\zeta_k = \sqrt{\tilde{n}}(\hat{\theta}_k - \theta_0)$, where $\hat{\theta}_k$ is the estimate of the ATT $\theta_0$ using iteration $k$ following Algorithm \ref{alg:est}. The sample mean and variance of the $\zeta_k$'s are around $0.05$ and $12.5$, respectively, close to $0$ and $11.455$. Moreover, the histogram of the $\zeta_k$'s (not shown) resembles a Gaussian distribution. It is worth noting that the result in Theorem \ref{thm:main} is still roughly valid despite the violation of Assumption \ref{asm:momentofxz} on the compactness of the support of $X$ and $Z$, since $X,Z$ are effectively compactly supported.

We now repeat what we did before, but instead compute $\zeta_k^{\textrm{cf}} = \sqrt{n}(\hat{\theta}_k^{\textrm{cf}} - \theta_0)$, where $\hat{\theta}_k^{\textrm{cf}}$ is the cross-fitted estimate of the ATT $\theta_0$ for iteration $k$ (see Remark \ref{remark:cf}). The sample mean and variance of the $\zeta_k^{\textrm{cf}}$'s are around $0.09$ and $11.2$ respectively, close to $0$ and $11.455$. Also, the histogram of the $\zeta_k^{\textrm{cf}}$'s, as shown in Figure \ref{fig:histo}, resembles that of a normal distribution, thus corroborating Remark \ref{remark:cf2}.

\begin{figure}[h]
    \centering
    \includegraphics[scale=0.55]{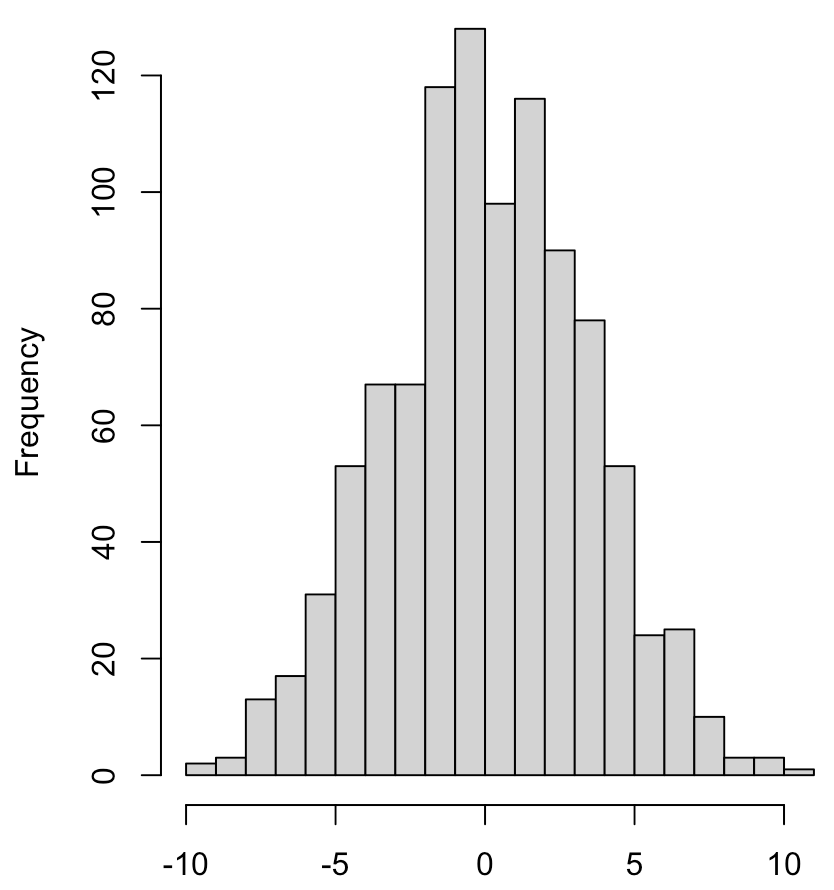}
    \caption{The histogram of the $\zeta_k^{\textrm{cf}}$'s looks fairly normal and centered around 0.}
    \label{fig:histo}
\end{figure}

\subsection{Simulation 2: Estimation of the asymptotic variance via bootstrapping}
\label{sec:sim2}
Our next simulation pertains to the validity of our asymptotic variance estimator obtained via bootstrapping. Recall that in the bootstrap, we start with a data set of size $n$ and follow the method described in Section \ref{sec:bootstr}, which involves sampling $n$ observations (with replacement) from the data set and performing Algorithm \ref{alg:est} to obtain an estimate of the ATT $\theta_0$. We then repeat this running procedure $B = 200$ times, with each iteration $1 \leq b \leq 200$ resulting in an estimate $\hat{\theta}_b$. The bootstrap estimator of the asymptotic variance is given by
\begin{equation*}
    \hat{\sigma}_\theta^2 = \tilde{n} \left(\frac{\sum_{b=1}^B \left(\hat{\theta}_b - \overline{{\theta}}_B\right)^2}{B-1}\right),
\end{equation*}
where $\overline{{\theta}}_B = \frac{1}{B}\left(\hat{\theta}_1 + \cdots + \hat{\theta}_B\right)$ and $\tilde{n} = n/3$. 

To evaluate the performance of the bootstrap, for each $n$, we generate $100$ different data sets of size $n$ and calculate $\hat{\sigma}_\theta^2$ for each data set. Table \ref{tab:boot} summarizes the mean and 90\% Monte-Carlo confidence region (CR) of the $\hat{\sigma}_\theta^2$'s for each $n \in \{2,000, 5,000, 12,000\}$. 

\begin{table}[h]
\centering
\begin{tabular}{cccc}
\hline
$n$      & $2,000$ & $5,000$ & $12,000$  \\ \hline
Mean of the $\hat{\sigma}_\theta^2$'s & $12.0$ & $11.7$ & 11.9 \\ \hline
90\% Monte-Carlo CR of the $\hat{\sigma}_\theta^2$'s & $(10.2, 13.7)$ & $(10.2, 13.5)$ & (10.3, 13.4) \\ \hline
\end{tabular}
\caption{The bootstrap provides a good estimate of the true asymptotic variance $\sigma_\theta^2 = 11.455$.} 
%\KW{MB comment -- support of the central $1-\alpha$ proportion of the distribution of $\hat{\sigma}_\theta^2$}}
\label{tab:boot}
\end{table} 

% Lastly, we utilize the same data set but separately estimate each component of $\sigma_\theta^2$ following the method in Section \ref{sec:ite-av-est}. The estimated asymptotic variance obtained through this method is $15.5$, which is also reasonably close to the true asymptotic variance.

\subsection{Simulation 3: Estimation of the ITE}
This section presents our simulation results for estimating the ITE function $\alpha_0$ via the method proposed in Section \ref{sec:ite-av-est}. To illustrate our method, we consider two simulation scenarios: (i) $\alpha_0$ only depends on $X$, which is equivalent to the standard CATE parameter, i.e., $Y_i(1) - Y_i(0) = \alpha_0(X_i)$; and (ii) $\alpha_0$ depends on both $X_i$ and $\eta_i$, i.e., $Y_i(1) - Y_i(0) = \alpha_0(X_i, \eta_i)$. 

\subsubsection{Case I: $\alpha_0$ only depends on $X$}
This subsection assumes that $\alpha_0$ depends only on the background information $X$; in particular, we take $\alpha_0(X) = X_1^2 + X_2X_3$. Our data generating process is as follows: 
\begin{equation*}
\begin{split}
   &Y = \left(X_1^2 + X_2X_3 \right) \mathds{1}_{Q \geq 0} + X_1 + X_3 + \eta/2 + \epsilon \\
   &Q = X_4 + \eta,
\end{split}
\end{equation*}
where, same as before, $X_1, \cdots, X_4 \sim \mathcal{N}(0,1)$ i.i.d., $\epsilon \sim \mathcal{N}(0, 0.5)$, $\eta \sim \mathcal{U}(-1,1)$, $Z = (X_1, X_2, X_3, X_4)$, and $X = (X_1, X_2, X_3)$. 

To obtain an estimate $\hat{\alpha}(X)$ of the ITE, we first follow Steps 1 to 5 of Algorithm \ref{alg:est} and then regress $(Y_i - X_i^\top \hat \beta)  - (Y_{c(i)} - X_{c(i)}^\top \hat \beta)$ on $X_i$ using cubic B-splines with degrees of freedom chosen via 4-fold cross-validation, and quadratic interaction terms. Table \ref{tab:mse1} shows that the mean squared error (MSE) of $\hat{\alpha}(X)$ among all treated individuals tends to decrease as the sample size $n$ increases. Moreover, Figure \ref{fig:itex} provides a comparison between $\hat{\alpha}(X)$ and $\alpha_0(X)$ for $(X_2, X_3) = (\pm 0.7, \pm 0.2)$ when $n = 50,000$, demonstrating that our approach can predict the ITE well.

\begin{figure}[t]
    \centering
    \includegraphics[scale=0.3]{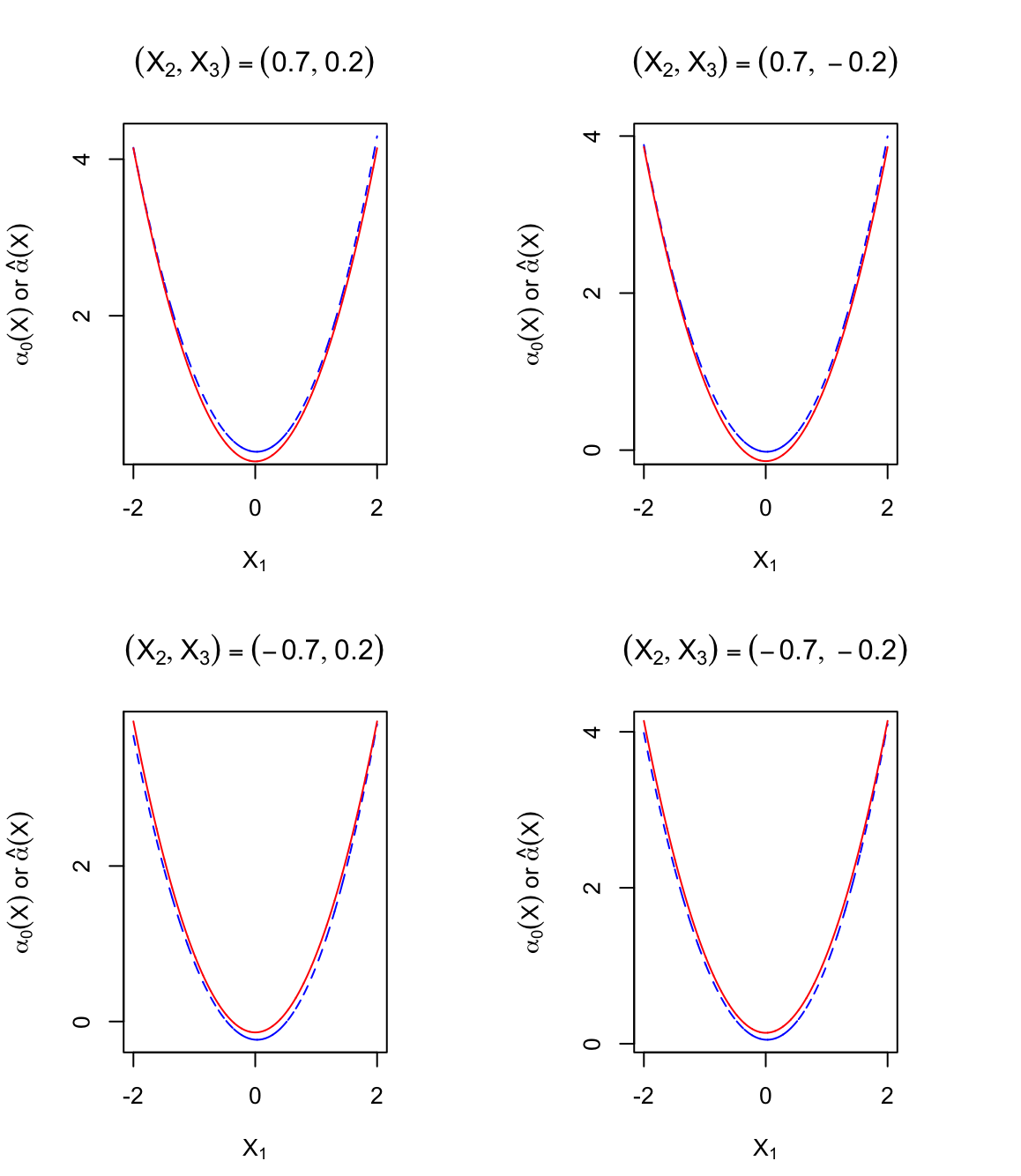}
    \caption{\textcolor{red}{Actual} vs. \textcolor{blue}{predicted} ITE for $(X_2, X_3) = (\pm 0.7, \pm 0.2)$.}
    \label{fig:itex}
\end{figure}

\begin{table}[h]
\centering
\begin{tabular}{ccccc}
\hline
$n$      & 4-fold CV d.f. & MSE of $\hat{\alpha}(X)$ & 4-fold CV d.f. & MSE of $\hat{\alpha}(X, \hat{\eta})$   \\ \hline
$5,000$ & 3 & 0.16 & 3 & 0.22 \\ \hline
$10,000$ & 3 & 0.04 & 3 & 0.08 \\ \hline
$20,000$ & 3 & 0.04 & 3 & 0.05 \\ \hline
$50,000$ & 3 & 0.01 & 3 & 0.03 \\ \hline
\end{tabular}
\caption{The MSEs of $\hat{\alpha}(X)$ and $\hat{\alpha}(X, \hat{\eta})$ tend to decrease as the sample size $n$ increases.}
\label{tab:mse1}
\end{table}

\subsubsection{Case II: $\alpha_0$ depends on $(X, \eta)$}
We now consider the same data generating process as used in Sections \ref{sec:sim1} and \ref{sec:sim2}, repeated below for clarity:
\begin{equation*}
\begin{split}
   &Y = \left(X_1^2 + X_2 X_3 + \eta^2 \right) \mathds{1}_{Q \geq 0} + X_1 + X_3 + \eta/2 + \epsilon \\
   &Q = X_4 + \eta.
\end{split}
\end{equation*}
Here, the ITE $\alpha_0(X,\eta) = X_1^2 + X_2 X_3 + \eta^2$ depends on both $X$ and $\eta$.

\begin{figure}[t]
    \centering
    \includegraphics[scale=0.3]{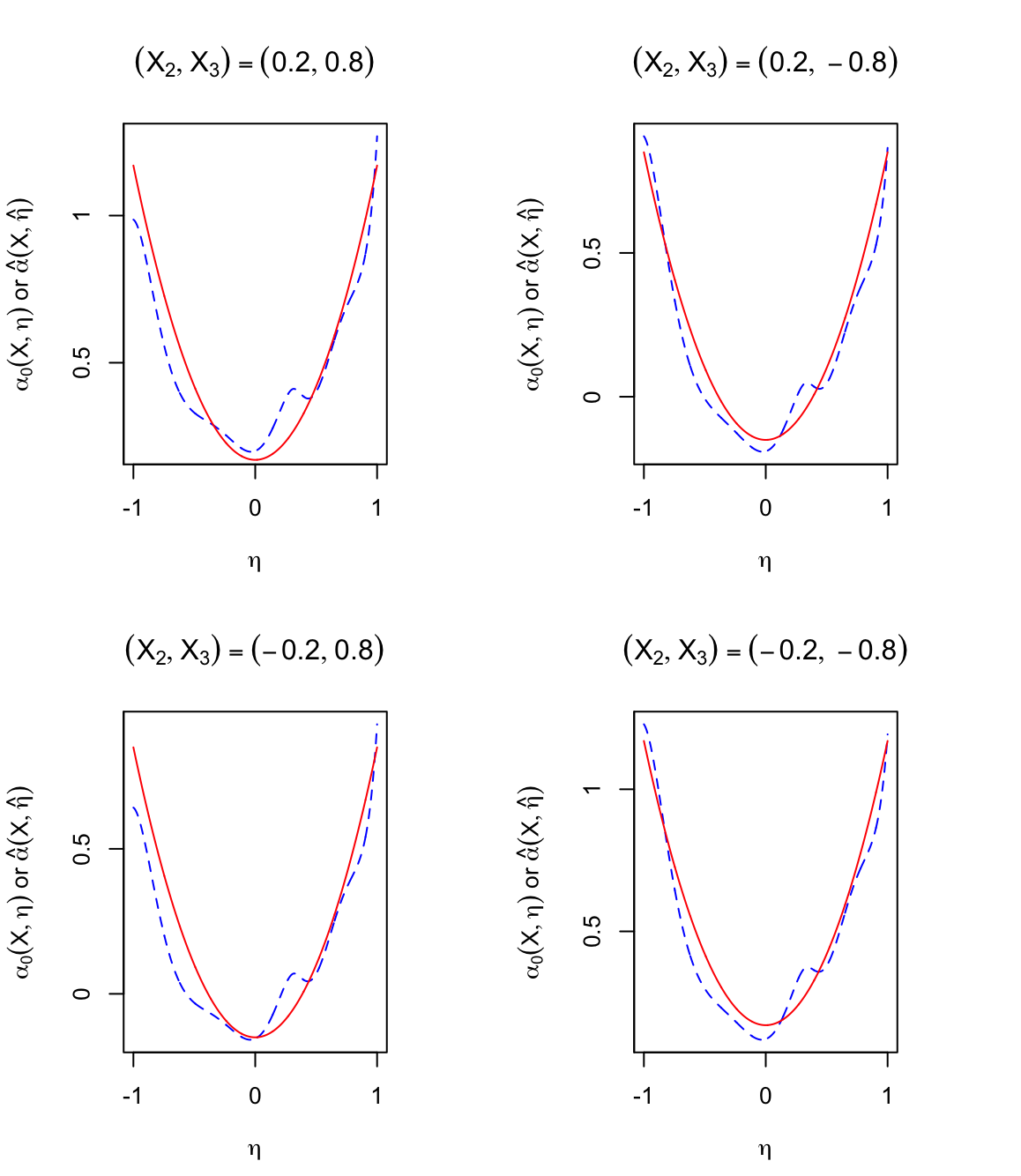}
    \caption{\textcolor{red}{Actual} vs. \textcolor{blue}{predicted} ITE for $(X_1, X_2, X_3, X_4) = (0.1, \pm 0.2, \pm 0.8, 1.5)$.}
    \label{fig:itex2}
\end{figure}

To obtain an estimate $\hat{\alpha}(X,\hat{\eta})$ of the ITE, we first follow Steps 1 to 5 of Algorithm \ref{alg:est} and then regress $(Y_i - X_i^\top \hat \beta)  - (Y_{c(i)} - X_{c(i)}^\top \hat \beta)$ on $X_i$ and $\hat{\eta}_i$ using cubic B-splines with degrees of freedom (d.f.) chosen via 4-fold cross-validation (CV), and quadratic interaction terms. Recall that $\eta$ is unknown in this scenario and thus needs to be estimated. Table \ref{tab:mse1} shows that the mean squared error (MSE) of $\hat{\alpha}(X, \hat{\eta})$ among all treated individuals tends to decrease as the sample size $n$ increases. Moreover, Figure \ref{fig:itex2} provides a comparison between $\hat{\alpha}(X, \hat{\eta})$ and $\alpha_0(X, \eta)$ for $(X_1, X_2, X_3, X_4) = (0.1, \pm 0.2, \pm 0.8, 1.5)$ and $\eta \in [-1, 1]$ when $n = 50,000$, demonstrating that our approach can predict the ITE well.

\section{Real data analysis}
\label{sec:rda}
In this section, we apply our method to two real data sets.  The first data set examines how Islamic political rule impacts women's empowerment, while the second one focuses on how academic probation based on grades affects a student's future GPA.

\subsection{Effect of Islamic party on women's education}

In this subsection, we utilize a data set originally introduced in \citet{meyersson2014islamic} regarding the effect of Islamic political rule on women's empowerment. Specifically, we aim to investigate whether the winning of Islamic parties affects women's educational outcomes. The data set consists of $2,629$ rows, each representing a municipality. The response variable $Y$ is calculated as $Y_w - Y_m$, where $Y_w$ ($Y_m$) denotes the percentage of women (men) aged 15 to 20 who completed high school by 2020. The treatment, $Q$, is determined by the difference in vote share between the largest Islamic and largest secular party in the 1994 election. An Islamic party is considered elected if and only if $Q > 0$. 

The covariates $X$ used in the analysis are as follows: (1) the Islamic vote share in 1994; (2) the number of parties with at least one vote in 1994; (3) the log of population in 1994; (4) whether a municipality is a district center; (5) whether a municipality is a province center; (6) whether a municipality is a sub-metro center; (7) whether a municipality is a metro center; (8) the share of population below the age of 19 in 2000; (9) the share of population above the age of 60 in 2000; (10) the ratio of males to females in 2000; and (11) the average household size in 2000.  

Using our method, we find an estimated ATT of $\hat{\theta} = 0.65$. To test the null hypothesis $H_0: \theta_0 = 0$ versus the alternative hypothesis $H_1: \theta_0 \neq 0$, we employ bootstrapping with $B = 500$ bootstrap samples. We obtain a bootstrap mean of $0.68$ and a 95\% bootstrap confidence interval of $(-0.62, 2.13)$. This result is similar to the result obtained by \citet{mukherjee2021estimation}, who assumed a homogeneous treatment effect, where they found a 95\% bootstrap confidence interval of $(-0.42, 1.42)$.

\subsection{Effect of academic probation on subsequent GPA}
We now utilize \citeauthor{lindo2010ability}'s (\citeyear{lindo2010ability}) data to analyze whether grade-based academic probation affects a student's subsequent GPA. This data set comprises $44,362$ rows, with each row representing a student from one of three large Canadian universities denoted as $A$, $B$, and $C$. For each student, the response variable $Y$ is their GPA in the first term of the second year, and the treatment $Q$ is the difference between their first-year GPA and the academic probation threshold. A student is placed on probation if and only if $Q > 0$. 

The covariates $X$ used are as follows: (1) the student's high school grade; (2) the total credits taken by the student in the first year; (3) whether the student is from university $A$; (4) whether the student is from university $B$; (5) whether the student is a male; (6) whether the student was born in North America; (7) the student's age when entering college; and (8) whether the student is a native English speaker.

Using our method, we obtain an estimated ATT of $\hat{\theta} = 0.28$. To test the null hypothesis $H_0: \theta_0 = 0$ against the alternative hypothesis $H_1: \theta_0 \neq 0$, we again employ bootstrapping with $B = 500$ bootstrap samples. The bootstrap mean and 95\% confidence interval are $0.27$ and $(0.22, 0.32)$, respectively. Since the confidence interval is entirely positive, we conclude that students who are placed on academic probation in their first year tend to see an improvement in their GPA in the first term of the second year. This conclusion is consistent with findings by Lindo et al. (2010) and Mukherjee et al. (2021), with the latter reporting a 95\% bootstrap confidence interval of $(0.25, 0.32)$.
\section{Conclusion}
\label{sec:conc}
In this paper, we developed an algorithm to estimate the ATT in a non-randomized treatment setting, where the treatment assignment depends on whether a variable exceeds a pre-specified threshold. Our method assumes that the treatment effect for each individual depends on their observed and unobserved covariates, and incorporates all individuals rather than only those close to the threshold. We proved that the resulting ATT estimator is both $\sqrt{n}$-consistent and asymptotically normal under standard regularity conditions, with empirical evidence showing that its asymptotic variance can be consistently estimated via the bootstrap. Moreover, a slight adjustment to our algorithm  allows us to estimate both the ITE and CATE, though we do not explore the theoretical properties of the corresponding estimators given the manuscript's complexity. Finally, we validated the effectiveness of our method through synthetic and real data analyses. Future work may focus on establishing the consistency of the bootstrap variance estimator and conducting inference on the ITE and CATE estimators.
\section{Roadmap of theoretical proofs}
\label{sec:sketch}

In this section, we outline proof sketches for Proposition \ref{prop:beta} and Theorem \ref{thm:main}, with full proofs provided in Appendices \ref{app:proof-prop} and \ref{app:proof-thm}, respectively.

\subsection{Proof sketch of Proposition \ref{prop:beta}}
We divide the proof into 4 key steps:

\underline{\textbf{Step 1: Asymptotic normality of $\hat{\gamma}_{\tilde{n}}$}} 

In Algorithm \ref{alg:est}, we first perform OLS of $Q_i$ against $Z_i$ for observations in $I_1$. The goal is to estimate each $\eta$ with $\hat{\eta} = Q - Z^\top \hat{\gamma}_{\tilde{n}}$, the main ingredient for differencing and matching in Steps 4 and 6, respectively. Observe that $$\sqrt{\tilde{n}}(\hat{\gamma}_{\tilde{n}} - \gamma_0) = \sum_{i=1}^{\tilde{n}} \frac{1}{\sqrt{\tilde{n}}} \left( \frac{\tilde{Z}^\top \tilde{Z}}{\tilde{n}}\right)^{-1} Z_i \eta_i,$$
where $\tilde{Z} = (Z_1; Z_2; \cdots; Z_{\tilde{n}})^\top \in \mathbb{R}^{{\tilde{n}} \times d_Z}$, converges in distribution to $\mathcal{N}(0, \Sigma_\gamma)$, where $\Sigma_\gamma = \sigma_\eta^2 \Sigma_Z^{-1}$. To estimate $\beta_0$, we regress the first-order differences of $Y$ on $X$ (based on the $\hat{\eta}$ values) for observations in the second partition that belong to the control group, denoted by $I_2^C$. 

Observe that \begin{equation}
\label{eq:betahatnormeqp}
    \sqrt{{\tilde{n}}}({\hat{\beta}_{\tilde{n}}} - \beta_0) = \left( \frac{1}{{\tilde{n}}} (\Delta X)^\top \Delta X \right)^{-1} \left( \frac{1}{\sqrt{{\tilde{n}}}} (\Delta X)^\top \Delta w \right),
\end{equation}
where $\Delta X$ consists of $X_{(i+1)} - X_{(i)}$ terms, $\Delta Y$ consists of $Y_{(i+1)} - Y_{(i)}$ terms, and $\Delta w$ consists of $\ell(\eta_{(i+1)}) - \ell(\eta_{(i)}) + \epsilon_{(i+1)} - \epsilon_{(i)}$ terms. To establish the asymptotic normality of ${\hat{\beta}_{\tilde{n}}}$, we examine each term in the product on the RHS of Equation \eqref{eq:betahatnormeqp}. We show that the first term converges in probability while the second terms converges in distribution to a normal distribution, whence the conclusion follows via Slutsky's theorem.

\underline{\textbf{Step 2: First term of RHS of Equation \eqref{eq:betahatnormeqp}}} 

We initially fix the first partition of the data; in other words, we first assume that ${\hat{\delta}_{\tilde{n}}}:= \hat{\gamma}_{\tilde{n}} - \gamma_0$ is fixed. Note that each observation in the control group within $I_2$ can be written as $X = g_{\hat{\delta}_{\tilde{n}}}(\hat{\eta}) + u_{\hat{\delta}_{\tilde{n}}}$, where $\hat{\eta} = \eta - Z^\top \hat{\delta}_{\tilde{n}}$ and $\mathbb{E}(u_{\hat{\delta}_{\tilde{n}}} \mid \hat{\eta}, Q < \tau_0) = 0$. We can also write $X = g_0(\eta) + u_0$, where $\mathbb{E}(u_0 \mid \eta, Q < \tau_0) = 0$. Since $X_{(i)} = g_{\hat{\delta}_{\tilde{n}}}(\hat{\eta}_{(i)})+ +u_{\hat{\delta}_{\tilde{n}(i)}}$ for every $i$, we have
\begin{equation*}
    \begin{split}
        \frac{1}{{\tilde{n}}} (\Delta X)^\top \Delta X = \frac{1}{{\tilde{n}}} \sum_{i=1}^{|I_2^C|-1} (X_{(i+1)} - X_{(i)}) (X_{(i+1)} - X_{(i)})^\top
        =  L + (M + M^\top) + P \,,
    \end{split}
\end{equation*}
where
\begin{equation*}
    \begin{split}
        &L = \frac{1}{{\tilde{n}}} \sum_{i=1}^{|I_2^C|-1} \left(g_{\hat{\delta}_{\tilde{n}}}({\hat{\eta}}_{(i+1)}) - g_{\hat{\delta}_{\tilde{n}}}({\hat{\eta}}_{(i)}) \right) \left(g_{\hat{\delta}_{\tilde{n}}}({\hat{\eta}}_{(i+1)}) - g_{\hat{\delta}_{\tilde{n}}}({\hat{\eta}}_{(i)}) \right)^\top, \\
        &M = \frac{1}{{\tilde{n}}} \sum_{i=1}^{|I_2^C|-1} \left(g_{\hat{\delta}_{\tilde{n}}}({\hat{\eta}}_{(i+1)}) - g_{\hat{\delta}_{\tilde{n}}}({\hat{\eta}}_{(i)}) \right) \left(uu_{\hat{\delta}_{\tilde{n}(i+1)}} - u_{\hat{\delta}_{\tilde{n}(i)}}\right)^\top, \\
        &P = \frac{1}{{\tilde{n}}} \sum_{i=1}^{|I_2^C|-1} \left(u_{\hat{\delta}_{\tilde{n}(i+1)}} - u_{\hat{\delta}_{\tilde{n}(i)}} \right) \left(u_{\hat{\delta}_{\tilde{n}(i+1)}} - u_{\hat{\delta}_{\tilde{n}(i)}} \right)^\top.
    \end{split}
\end{equation*}

Utilizing the fact that the ordering is done on the $\hat{\eta}_i$'s and $g_{\hat{\delta}_{\tilde{n}}}$ is Lipschitz, we can show via the Cauchy-Schwarz inequality that $L$ and $M$ are both $o_p(1)$. Moreover, we can show that conditional on $\hat{\delta}_{\tilde{n}}$,
\begin{align*}
    o_p(1) & = P - \frac{2}{{\tilde{n}}} \sum_{i=1}^{|I_2^C|} u_{\hat{\delta}_{\tilde{n}i}} u_{\hat{\delta}_{\tilde{n}i}}^\top \\
    & =  P - 2\frac{|I_2^C|}{{\tilde{n}}}\left(\frac{1}{|I_2^C|} \sum_{i=1}^{|I_2^C|} (X_{i} - \bbE(X \mid \hat \eta_{i}, Q < \tau_0))(X_{i} - \bbE(X \mid \hat \eta_{i}, Q < \tau_0))^\top\right) \\
    \implies & P - 2\bbP(Q < \tau_0) \bbE\left(\var\left(X \mid \hat{\eta}, Q < \tau_0\right) \mid Q < \tau_0\right) = o_p(1).
\end{align*}
% To analyze the remaining term $P$, we can rewrite it as 
% \begin{equation*}
%     \begin{split}
%     P = \frac{1}{N} \sum_{i=1}^{|I_2^C|-1} \left(u_{\delta_{(i+1)}} - u_{\delta_{(i)}} \right) \left(u_{\delta_{(i+1)}} - u_{\delta_{(i)}} \right)^\top
%     = P_1 - P_2 - P_3 - P_3^\top,
%     \end{split}
% \end{equation*}
% where 
% \begin{equation*}
%     \begin{split}
%         P_1 = \frac{2}{N} \sum_{i=1}^{|I_2^C|} u_{\delta_{i}} u_{\delta_{i}}^\top, \hspace{2mm}
%         P_2 = \frac{1}{N} \left( u_{\delta_{(1)}} u_{\delta_{(1)}}^\top + u_{\delta_{(|I_2^C|)}}u_{\delta_{(|I_2^C|)}}^\top \right), \hspace{2mm}
%         P_3 = \frac{1}{N} \sum_{i=1}^{|I_2^C|-1} u_{\delta_{(i+1)}}u_{\delta_{(i)}}^\top.
%     \end{split}
% \end{equation*}
% Using the law of large numbers, we have $P_1 \overset{P} 2P(Q<\tau_0)\Sigma_{u,\delta}$ where $\Sigma_{u,\delta} = \mathbb{E}(\textrm{var}(X \mid \hat{\eta}, Q < \tau_0))$. Moreover, $P_2$ is clearly $o_p(1)$ and so is $P_3$ due to Assumption \ref{asm:momentofxz} by Chebyshev's inequality. 
% So far, we have shown that for any fixed $\delta$, we have
% $$\frac{1}{N} (\Delta X)^\top \Delta X = 2P(Q < \tau_0) \Sigma_{u,\delta} + o_p(1),$$
% where $\Sigma_{u,\delta} = \mathbb{E}(\textrm{var}(X \mid \hat{\eta}, Q < \tau_0))$. Under Assumptions \ref{asm:momentofxz} and \ref{asm:errormeanvar} and Lebesgue's DCT, 

We can now apply Lebesgue's DCT to prove that $P$ (and thus $ (\Delta X)^\top \Delta X/{\tilde{n}}$) converges to $2\bbP(Q < \tau_0) \bbE\left(\var\left(X \mid \eta, Q < \tau_0\right) \mid Q < \tau_0 \right)$ in probability. Finally, an application of the continuous mapping theorem yields
\begin{equation}
    \left( \frac{1}{{\tilde{n}}} (\Delta X)^\top \Delta X \right)^{-1} \overset{P}{\longrightarrow} \frac{1}{2\bbP (Q<\tau_0)}\Sigma_u^{-1},
    \label{eq:cmt}
\end{equation}
where $\Sigma_u := \bbE\left(\var\left(X \mid \eta, Q < \tau_0\right) \mid Q < \tau_0 \right)$.
% Details can be found in Appendix \ref{app:proof-prop}. 
% we have $$\frac{1}{N} (\Delta X)^\top \Delta X \overset{P} 2P(Q < \tau_0) \Sigma_{u}$$ unconditional on $\delta$, where $\Sigma_u = \mathbb{E}(\textrm{var}(X \mid {\eta}, Q < \tau_0))$. Finally, an application of the continuous mapping theorem yields 
% $$\left( \frac{1}{N} (\Delta X)^\top \Delta X \right)^{-1} \overset{P} \frac{1}{2P(Q<\tau_0)}\Sigma_u^{-1}.$$

\underline{\textbf{Step 3: Second term of RHS of Equation \eqref{eq:betahatnormeqp}}} 

We have
\begin{equation*}
    \begin{split}
        \frac{1}{\sqrt{{\tilde{n}}}} (\Delta X)^\top \Delta w &= \frac{1}{\sqrt{{\tilde{n}}}} \sum_{i=1}^{|I_2^C|-1} (X_{(i+1)} - X_{(i)}) \left(\ell(\eta_{(i+1)}) - \ell(\eta_{(i)}) + \epsilon_{(i+1)} - \epsilon_{(i)} \right) 
    \end{split}
\end{equation*}
As before, we expand the above equation by rewriting $X_i = g_{\hat{\delta}_{\tilde{n}}}(\hat \eta_i) + u_{\hat{\delta}_{\tilde{n}i}}$. Some customary algebra followed by ignoring lower order terms shows
\begin{align*}
      \frac{1}{\sqrt{{\tilde{n}}}} (\Delta X)^\top \Delta w &=\frac{1}{\sqrt{{\tilde{n}}}} \sum_{i=1}^{|I_2^C|-1} \left(\ell(\eta_{(i+1)}) - \ell(\eta_{(i)}) \right) \left( u_{\hat{\delta}_{\tilde{n}(i+1)}} - u_{\hat{\delta}_{\tilde{n}(i)}} \right) \\
      & \qquad \qquad + \frac{1}{\sqrt{{\tilde{n}}}} \sum_{i=1}^{|I_2^C|-1} \left(\epsilon_{(i+1)} - \epsilon_{(i)} \right) \left( u_{\hat{\delta}_{\tilde{n}(i+1)}} - u_{\hat{\delta}_{\tilde{n}(i)}} \right) + o_p(1) \\
      & \triangleq F + H + o_p(1) \,,
\end{align*}
% \begin{equation*}
%    \begin{split}
%         &E = \frac{1}{\sqrt{N}} \sum_{i=1}^{|I_2^C|-1} \left(f(\eta_{(i+1)}) - f(\eta_{(i)}) \right) \left( g_{\delta}(\hat{\eta}_{(i+1)}) - g_{\delta}(\hat{\eta}_{(i)}) \right),\\
%     &F = \frac{1}{\sqrt{N}} \sum_{i=1}^{|I_2^C|-1} \left(f(\eta_{(i+1)}) - f(\eta_{(i)}) \right) \left( u_{\delta_{(i+1)}} - u_{\delta_{(i)}} \right),\\
%     &G = \frac{1}{\sqrt{N}} \sum_{i=1}^{|I_2^C|-1} \left(\epsilon_{(i+1)} - \epsilon_{(i)} \right) \left( g_{\delta}(\hat{\eta}_{(i+1)}) - g_{\delta}(\hat{\eta}_{(i)}) \right),\\
%     &H = \frac{1}{\sqrt{N}} \sum_{i=1}^{|I_2^C|-1} \left(\epsilon_{(i+1)} - \epsilon_{(i)} \right) \left( u_{\delta_{(i+1)}} - u_{\delta_{(i)}} \right).
%     \end{split}
% \end{equation*}
% First, we can show $\sum_{i=1}^{|I_2^C|-1} \left( f(\eta_{(i+1)}) - f(\eta_{(i)})\right)^2 = O_p(1)$ under Assumptions \ref{asm:momentofxz}, \ref{asm:errormeanvar} and \ref{asm:fbounded}, and the fact that $||\delta||_2^2 = O_p(N^{-1})$. From here, both $E$ and $G$ can be shown to be $o_p(1)$
% by the Cauchy-Schwarz inequality. We now analyze $F$. 
where both $F$ and $H$ contribute to the asymptotic normality. 

First, observe that
% Adding and subtracting $f(\hat \eta_{(i)})$ and $f(\hat \eta_{(i+1)})$ yields
% Recall that we perform matching with respect to $\hat \eta_i$; in other words, we expect $\hat \eta_{(i+1)} - \hat \eta_{(i)}$ to be small. Also, if the function $f$ has a minimal degree of smoothness (e.g., $f$ is Lipschitz), then we expect $f(\hat \eta_{(i+1)}) - f(\hat \eta_{(i)})$ to be small. Note that $T_4$ consists of the terms $f(\eta_{(i+1)}) - f(\eta_{(i)})$. Therefore, one should add and subtract $f(\hat \eta_{(i)})$ and $f(\hat \eta_{(i+1)})$, resulting in the following expansion: 
\begin{equation*}
    \begin{split}
        F &= \frac{1}{\sqrt{{\tilde{n}}}} \sum_{i=1}^{|I_2^C|-1}\left( \ell(\eta_{(i+1)}) - \ell(\hat{\eta}_{(i+1)})\right) \left( u_{\hat{\delta}_{\tilde{n}(i+1)}} - u_{\hat{\delta}_{\tilde{n}(i)}} \right) \\
        & \qquad - \frac{1}{\sqrt{{\tilde{n}}}} \sum_{i=1}^{|I_2^C|-1}\left( \ell(\eta_{(i)}) - \ell(\hat{\eta}_{(i)})\right) \left( u_{\hat{\delta}_{\tilde{n}(i+1)}} - u_{\hat{\delta}_{\tilde{n}(i)}} \right) \\
        & \qquad \qquad + \frac{1}{\sqrt{{\tilde{n}}}} \sum_{i=1}^{|I_2^C|-1}\left( \ell(\hat{\eta}_{(i+1)}) - \ell(\hat{\eta}_{(i)})\right) \left( u_{\hat{\delta}_{\tilde{n}(i+1)}} - u_{\hat{\delta}_{\tilde{n}(i)}} \right).
    \end{split}
\end{equation*}

The last summand can be shown to be asymptotically negligible, again due to the fact that the ordering is done on the $\hat{\eta}_i$'s and $f$ is Lipschitz. Also, the first and second summands can be approximated via a two-step Taylor expansion: 
\begin{equation*}
    \ell(\eta_{(i)}) - \ell(\hat{\eta}_{(i)}) = (\eta_{(i)} - \hat{\eta}_{(i)}) f'(\hat{\eta}_{(i)}) + \frac{(\eta_{(i)} - \hat{\eta}_{(i)})^2}{2} \ell''(\tilde{\eta}_{(i)}) \,,
\end{equation*}
where $\tilde \eta_{(i)}$ is between $\eta_{(i)}$ and $\hat \eta_{(i)}$ (we can similarly expand $\ell(\eta_{(i+1)}) - \ell(\hat{\eta}_{(i+1)})$). The quadratic terms can be shown to be asymptotically negligible. 

For the linear terms, we first decompose $Z$ for control units into $q_{\hat{\delta}_{\tilde{n}}}(\hat{\eta}) + w_{\hat{\delta}_{\tilde{n}}}$, where $\hat{\eta} = \eta - Z^\top \hat{\delta}_{\tilde{n}}$ and $\mathbb{E}(w_{\hat{\delta}_{\tilde{n}}} \mid \hat{\eta}, Q < \tau_0) = 0$. Similarly, we can write $Z = q_0(\eta) + w_0$, where $\mathbb{E}(w_0 \mid \eta, Q < \tau_0) = 0$. Some algebra followed by ignoring lower order terms yields
$F = \sqrt{{\tilde{n}}}(\hat{\gamma}_{\tilde{n}} - \gamma_0) \bar{F} + o_p(1)$, where
$$\bar{F} = \frac{1}{{\tilde{n}}} \left( \sum_{i=1}^{|I_2^C|-1} \left( u_{{\hat{\delta}_{\tilde{n}(i+1)}}} - u_{{\hat{\delta}_{\tilde{n}(i)}}}\right) \left( \ell'(\hat{\eta}_{(i+1)}) (w_{{\hat{\delta}_{\tilde{n}(i+1)}}})^\top - \ell'(\hat{\eta}_{(i)}) (w_{{\hat{\delta}_{\tilde{n}(i)}}})^\top \right)\right).$$

% :
% \begin{equation*}
%     \left| \frac{1}{\sqrt{N}} \sum_{i=1}^{|I_2^C|-1}  (\eta_{(i)} - \hat{\eta}_{(i)})^2 \ell''(\tilde{\eta}_{(i)}) (u_{\delta_{(i+1)}} - u_{\delta_{(i)}})  \right| = O_p(N^{-1/2}) = o_p(1)
% \end{equation*}
% and its equivalent version for $\eta_{(i+1)} - \hat{\eta}_{(i+1)}$. 
% \begin{align*}
% &T_1 & = \sqrt{N}(\hat \gamma - \gamma_0)\left(\frac{1}{N} \left( \sum_{i=1}^{|I_2^C|-1} \left( u_{\delta_{(i+1)}} - u_{\delta_{(i)}}\right) \left( f'(\hat{\eta}_{(i+1)}) (u_{\delta_{(i+1)}})^\top - f'(\hat{\eta}_{(i)}) (u_{\delta_{(i)}})^\top \right)\right)\right) + o_p(1) \\
% & \triangleq \sqrt{N}(\hat \gamma - \gamma_0) \bar T + o_p(1) \,.
% \end{align*}
It is possible to show that conditional on $\hat{\delta}_{\tilde{n}}$,
\begin{align*}
    &\bar{F} - 2\frac{|I_2^C|}{{\tilde{n}}} \left(\frac{1}{|I_2^C|} \sum_{i=1}^{|I_2^C|} \ell'(\hat{\eta}_i) u_{{\hat{\delta}_{\tilde{n}i}}} w_{{\hat{\delta}_{\tilde{n}i}}}^\top \right) = o_p(1)\\
    &\implies \bar{F} - 2\bbP(Q < \tau_0) \bbE\left( \ell'(\hat{\eta}) u_{\hat{\delta}_{\tilde{n}}} w_{\hat{\delta}_{\tilde{n}}}^\top \mid Q < \tau_0\right) = o_p(1).
\end{align*}

Following a similar approach as for the term $P$ in Step 2, we have that $\bar{F}$ converges to $2 \bbP (Q < \tau_0) \mathbb{E}\left(\ell'({\eta}) u_{0} w_{0}^\top \mid Q < \tau_0 \right)$ in probability, where $u_0 = X - \mathbb{E}(X \mid \eta, Q < \tau_0)$ as before. Omitting $o_p(1)$ terms and using the expansion of $\sqrt{{\tilde{n}}}(\hat{\gamma}_{\tilde{n}} - \gamma_0)$ in Step 1, we have
\begin{equation}
    F = \sum_{i=1}^{\tilde{n}} \frac{2}{\sqrt{{\tilde{n}}}} \bbP(Q < \tau_0) \mathbb{E}\left(\ell'({\eta}) u_0 w_0^\top \mid Q < \tau_0 \right)  \left( \frac{\tilde{Z}^\top \tilde{Z}}{{\tilde{n}}}\right)^{-1} Z_i \eta_i.
    \label{eq:f}
\end{equation}

Finally, we can rewrite $H$ as
\begin{equation}
    \begin{split}
        H &= \frac{1}{\sqrt{{\tilde{n}}}} \sum_{i=1}^{|I_2^C|-1}\left(  \epsilon_{(i+1)} - \epsilon_{(i)} \right) \left( u_{{\hat{\delta}_{\tilde{n}(i+1)}}} -u_{{\hat{\delta}_{\tilde{n}(i)}}}  \right) \\
        &= \frac{1}{\sqrt{{\tilde{n}}}} \Big( \epsilon_{(1)} (u_{{\hat{\delta}_{\tilde{n}(1)}}} - u_{{\hat{\delta}_{\tilde{n}(2)}}}) + \epsilon_{(2)} (2u_{{\hat{\delta}_{\tilde{n}(2)}}} - u_{{\hat{\delta}_{\tilde{n}(1)}}} - u_{{\hat{\delta}_{\tilde{n}(3)}}}) + \cdots \\ &+ \epsilon_{(|I_2^C|-1)} (2u_{{\hat{\delta}_{\tilde{n}(|I_2^C|-1)}}} - u_{{\hat{\delta}_{\tilde{n}(|I_2^C|- 2)}}} - u_{{\hat{\delta}_{\tilde{n}(|I_2^C|)}}}) 
        + \epsilon_{(|I_2^C|)} (u_{{\hat{\delta}_{\tilde{n}(|I_2^C|)}}} - u_{{\hat{\delta}_{\tilde{n}(|I_2^C|- 1)}}}) \Big) \\
        &:= \frac{1}{\sqrt{{\tilde{n}}}} \sum_{i=1}^{|I_2^C|} \epsilon_{(i)} a_{{\hat{\delta}_{\tilde{n}(i)}}}  := \frac{1}{\sqrt{{\tilde{n}}}} \sum_{i={\tilde{n}}+1}^{2{\tilde{n}}} \epsilon_i a_{{\hat{\delta}_{\tilde{n}}},i},
    \end{split}
    \label{eq:h}
\end{equation}
where $a_{{\hat{\delta}_{\tilde{n}}}, i} = 0$ for any treatment observation $i$.

\underline{\textbf{Step 4: Putting everything together}}

Now, we can employ the martingale central limit theorem \citep{billingsley1995probability} to establish the asymptotic normality of $\frac{1}{\sqrt{{\tilde{n}}}} (\Delta X)^\top \Delta w = F+ H + o_p(1).$ Effectively, we need to show $F + H$ is asymptotically normal. Recall that $F$ is primarily a function of $\eta_i$'s (see Equation \eqref{eq:f}) and $H$ is a function of both $\eta_i$'s and $\eps_i$'s (see Equation \eqref{eq:h}). Consequently, they are \emph{not} independent and thus we need to establish the normality jointly.  The detailed argument is presented in Appendix \ref{app:proof-prop}. 

Lastly, we use Slutsky's theorem to establish the asymptotic normality of $\sqrt{{\tilde{n}}}({\hat{\beta}_{\tilde{n}}} - \beta_0)$.

\subsection{Proof sketch of Theorem \ref{thm:main}}
We divide the proof into 4 key steps. We first introduce some notations. Let $t_i = \mathds{1}_{Q_i \geq \tau_0}$ denote the treatment status of observation $i$, ${\tilde{n}}_1 = \sum_{i=2{\tilde{n}}+1}^{3{\tilde{n}}} t_i$ denote the number of treatment observations in $I_3$, and ${\tilde{n}}_0 = {\tilde{n}} - {\tilde{n}}_1$ denote the number of control observations in $I_3$. 
Also, let $I_3^T$ (resp. $I_3^C$) denotes the group of individuals in $I_3$ with $t_i = 1$ (resp. $t_i = 0$). 
For each $i \in I_3^T$, let $c(i) \in I_3^C$ be its \emph{nearest neighbor} in the control group with respect to $\hat \eta$  
% the index of the closest observation (w.r.t. $\hat{\eta}$) from $i$ that belongs to the opposite side of the treatment group. 
(see Step 5 of Algorithm \ref{alg:est}).

Following \citet{abadie2016matching}, for each $i \in I_3^C$, we define $K_{{\hat{\delta}_{\tilde{n}}},i}$ to be the number of times observation $i$ is used as a match. In other words, $K_{{\hat{\delta}_{\tilde{n}}}, i}$ denotes the number of treatment observations whose nearest neighbor (with respect to $\hat{\eta}$) is $i$. Here, ${\hat{\delta}_{\tilde{n}}} = \hat{\gamma}_{\tilde{n}} - \gamma_0$ that is obtained from $I_1$. We divide the proof into 4 key steps:

\textbf{\underline{Step 1: Decomposition of $\sqrt{{\tilde{n}}}(\hat{\theta}_{\tilde{n}} - \theta_0)$}} 

Write $\sqrt{{\tilde{n}}}(\hat{\theta}_{\tilde{n}} - \theta_0) = (1) + (2) + (3) + (4)$, where

\begin{equation*}
    \begin{split}
        &(1) = \frac{\sqrt{{\tilde{n}}}}{{\tilde{n}}_1} \sum_{i=2{\tilde{n}}+1}^{3{\tilde{n}}} t_i \left( \alpha_0(X_i, \eta_i) - \mathbb{E}(\alpha_0(X, \eta) \mid Q \geq \tau_0)\right),\\
        &(2) = \frac{\sqrt{{\tilde{n}}}}{{\tilde{n}}_1} (\beta_0 - {\hat{\beta}_{\tilde{n}}})^\top \sum_{i=2{\tilde{n}}+1}^{3{\tilde{n}}} \left( t_i - (1-t_i) K_{{\hat{\delta}_{\tilde{n}}},i} \right) X_i,\\
        &(3) = \frac{\sqrt{{\tilde{n}}}}{{\tilde{n}}_1} \sum_{i=2{\tilde{n}}+1}^{3{\tilde{n}}} \left( t_i - (1-t_i) K_{{\hat{\delta}_{\tilde{n}}},i} \right) \epsilon_i,\\
        &(4) = \frac{\sqrt{{\tilde{n}}}}{{\tilde{n}}_1} \sum_{i=2{\tilde{n}}+1}^{3{\tilde{n}}} t_i \left( \ell(\eta_i) - \ell(\eta_{c(i)})\right).
    \end{split}
\end{equation*}
We initially focus on examining terms (2) and (4) before returning to discuss terms (1) and (3) in Step 4.

\textbf{\underline{Step 2: Term $(2)$}} 

We can rewrite $(2)$ as
\begin{equation*}
    \begin{split}
        (2) = \left( \sqrt{{\tilde{n}}} (\beta_0 - {\hat{\beta}_{\tilde{n}}})\right)^\top \left( \frac{{\tilde{n}}}{{\tilde{n}}_1} \right) \left( \frac{1}{{\tilde{n}}} \sum_{i=2{\tilde{n}}+1}^{3{\tilde{n}}} t_i X_i - \frac{1}{{\tilde{n}}} \sum_{i=2{\tilde{n}}+1}^{3{\tilde{n}}} (1-t_i) K_{{\hat{\delta}_{\tilde{n}}},i} X_i \right).
    \end{split}
\end{equation*}

As established in Proposition \ref{prop:beta}, the first term in the product converges to $\mathcal{N}(0, \Sigma_\beta)$. Also, it is easy to see that the second term converges in to $1/\bbP (Q \geq \tau_0)$ in probability, and the first part of the third term converges to $\mathbb{E}(wX) = \bbP (Q \geq \tau_0) \mathbb{E}(X \mid Q \geq \tau_0)$ in probability. 

For the remaining term, it is possible to show that
\begin{equation*}
    \frac{1}{{\tilde{n}}} \sum_{i=2{\tilde{n}}+1}^{3{\tilde{n}}} (1-t_i) K_{{\hat{\delta}_{\tilde{n}}},i} X_i - \bbP (Q \geq \tau_0) \mathbb{E}\left( \frac{f_{1,{\hat{\delta}_{\tilde{n}}}}(\hat{\eta})}{f_{0,{\hat{\delta}_{\tilde{n}}}}(\hat{\eta})} X \mid Q < \tau_0 \right) = o_p(1)
\end{equation*}
conditional on ${\hat{\delta}_{\tilde{n}}}$ by slightly modifying the proofs of Lemmas S.6, S.7 and S.10 of \citet{abadie2016matching}. Here, $f_{i,{\hat{\delta}_{\tilde{n}}}}(\hat{\eta})$ is the density of $\hat{\eta}:= \eta - Z^\top {\hat{\delta}_{\tilde{n}}}$ conditional on $w = i$ for $i \in \{0,1\}$. 

Using Lebesgue's DCT, we have
\begin{equation*}
    \frac{1}{{\tilde{n}}} \sum_{i=2{\tilde{n}}+1}^{3{\tilde{n}}} (1-t_i) K_{{\hat{\delta}_{\tilde{n}}},i} X_i \overset{P}{\longrightarrow} \bbP (Q \geq \tau_0) \mathbb{E}\left( \frac{f_{1}({\eta})}{f_{0}({\eta})} X \mid Q < \tau_0 \right),
\end{equation*}
where $f_i(\eta)$ denotes the density of $\eta$ conditional on $w = i$ for $i \in \{0,1\}$. Intuitively, the density ratio term $f_1(\eta)/f_0(\eta)$ appears since for any control observation $i$, $K_{{\hat{\delta}_{\tilde{n}}},i}$ denotes the number of treatment observations having $i$ as their nearerst neighbor. 

Following the derivations in Proposition \ref{prop:beta} and after some algebra, we obtain 
%\KW{MB comment -- the density ratio should be explained as coming from the K's}
\begin{align*}
    (2) &= \left( \sqrt{{\tilde{n}}} ({\hat{\beta}_{\tilde{n}}} - \beta_0)\right)^\top \left( \mathbb{E}\left( \frac{f_1({\eta})}{f_0({\eta})} X \mid Q < \tau_0 \right) - \mathbb{E}\left( X \mid Q \geq \tau_0 \right) \right) + o_p(1) \\
    &= \sum_{i=1}^{\tilde{n}} \frac{1}{\sqrt{{\tilde{n}}}} A_3^\top A_1 \left( \frac{\tilde{Z}^\top \tilde{Z}}{{\tilde{n}}}\right)^{-1} Z_i \eta_i + \sum_{i={\tilde{n}}+1}^{2{\tilde{n}}} \frac{1}{\sqrt{{\tilde{n}}}} A_3^\top \epsilon_i a_{{\hat{\delta}_{\tilde{n}}},i} + o_p(1),
\end{align*}
where $A_1 = 2 \bbP (Q < \tau_0) \mathbb{E}\left(\ell'({\eta}) u_0 w_0^\top \mid Q < \tau_0 \right)$ and $$A_3 = \frac{1}{2\bbP (Q<\tau_0)} \Sigma_u^{-1} \left( \mathbb{E}\left( \frac{f_1({\eta})}{f_0({\eta})} X \mid Q < \tau_0 \right) - \mathbb{E}\left( X \mid Q \geq \tau_0 \right) \right).$$
The terms $A_1$ and $A_3$ emerge from the decomposition of $\sqrt{{\tilde{n}}}({\hat{\beta}_{\tilde{n}}} - \beta_0)$ in the roadmap of Proposition \ref{prop:beta}'s proof (see Equations \eqref{eq:betahatnormeqp} to \eqref{eq:h}). 

\textbf{\underline{Step 3: Term $(4)$}} 

Next, we can decompose $(4)$ into
\begin{equation}
    \begin{split}
        (4) &= \frac{\sqrt{{\tilde{n}}}}{{\tilde{n}}_1} \sum_{i=2{\tilde{n}}+1}^{3{\tilde{n}}} t_i \left( \ell(\hat{\eta}_i) - \ell(\hat{\eta}_{c(i)})\right) 
        + \frac{\sqrt{{\tilde{n}}}}{{\tilde{n}}_1} \sum_{i=2{\tilde{n}}+1}^{3{\tilde{n}}} t_i \left( \ell(\eta_i) - \ell(\hat{\eta}_{i})\right) \\
        & \qquad - \frac{\sqrt{{\tilde{n}}}}{{\tilde{n}}_1} \sum_{i=2{\tilde{n}}+1}^{3{\tilde{n}}} t_i \left( \ell(\eta_{c(i)}) - \ell(\hat{\eta}_{c(i)})\right).
    \end{split}
    \label{eq:4-decomp}
\end{equation}   

As before, we first develop our argument for a fixed ${\hat{\delta}_{\tilde{n}}}$. Following the proof of Proposition 1 in \citet{abadie2016matching}, we can show that the first summand is $o_p(1)$. To address the second and third summands, we again utilize two-step Taylor expansions akin to those employed for the term $F$. After some algebra, we obtain
\begin{equation*}
    \begin{split}
        (4) &= \frac{\sqrt{{\tilde{n}}}}{{\tilde{n}}_1} \sum_{i=2{\tilde{n}}+1}^{3{\tilde{n}}} (t_i - (1 - t_i) K_{{\hat{\delta}_{\tilde{n}}}, i}) (\eta_i - \hat{\eta}_i) \ell'(\hat{\eta}_i) \\
        &= (\sqrt{{\tilde{n}}}(\hat{\gamma}_{\tilde{n}} - \gamma_0))^\top \left( \frac{{\tilde{n}}}{{\tilde{n}}_1}\right) \left( \frac{1}{{\tilde{n}}} \sum_{i=2{\tilde{n}}+1}^{3{\tilde{n}}} (t_i - (1 - t_i) K_{{\hat{\delta}_{\tilde{n}}}, i}) Z_i \ell'(\hat{\eta}_i)\right).
    \end{split}
\end{equation*}
The expression $t_i - (1 - t_i) K_{{\hat{\delta}_{\tilde{n}}}, i}$ intuitively stems from the observation that in the last two summands of Equation \eqref{eq:4-decomp}, each treated observation appears once and each control observation appears $K_{{\hat{\delta}_{\tilde{n}}}, i}$ times. 

Finally, following a similar derivation to that for the term $(2)$, 
%we have
% \begin{align*}
%     \begin{split}
%    &\left( \frac{N}{N_1}\right) \left( \frac{1}{N} \sum_{i=2N+1}^{3N} (t_i - (1 - t_i) K_{\delta, i}) X_i f'(\hat{\eta}_i)\right) \overset{P}{\longrightarrow} \\ &\qquad \mathbb{E}\left(Xf'({{\eta}}) \mid Q \geq \tau_0 \right) - \mathbb{E}\left( \frac{f_1({{\eta}})}{f_0({{\eta}})} X f'({{\eta}}) \mid Q < \tau_0 \right).
%    \end{split}
% \end{align*}
% Thus, 
we have 
$$(4) = \sum_{i=1}^{\tilde{n}} \frac{1}{\sqrt{{\tilde{n}}}} A_4^\top  \left( \frac{\tilde{Z}^\top \tilde{Z}}{{\tilde{n}}}\right)^{-1} Z_i \eta_i + o_p(1),$$ where
$$A_4 =  \mathbb{E}\left(Z\ell'({\eta}) \mid Q \geq \tau_0 \right) - \mathbb{E}\left( \frac{f_1({\eta})}{f_0({\eta})} Z \ell'({\eta}) \mid Q < \tau_0 \right).$$
Again, the density ratio term appears due to the presence of $K_{{\hat{\delta}_{\tilde{n}}}, i}$.

\textbf{\underline{Step 4: Putting everything together}} 

To show that $\sqrt{{\tilde{n}}}(\hat{\theta}_{\tilde{n}} - \theta_0) = (1) + (2) + (3) + (4)$ is asymptotically normal, we substitute the terms $(2)$ and $(4)$ following our derivations in Steps 2 and 3. Meanwhile, we use the formulas for $(1)$ and $(3)$ as in Step 1.
We have
\begin{equation*}
    \begin{split}
        &\sqrt{{\tilde{n}}}(\hat{\theta}_{\tilde{n}} - \theta_0) \\
        &= \sum_{i=1}^{\tilde{n}} \frac{1}{\sqrt{{\tilde{n}}}} A_5^\top \left( \frac{\tilde{Z}^\top \tilde{Z}}{{\tilde{n}}}\right)^{-1} Z_i \eta_i + \sum_{i={\tilde{n}}+1}^{2{\tilde{n}}}  \frac{1}{\sqrt{{\tilde{n}}}} A_3^\top \epsilon_i a_{{\hat{\delta}_{\tilde{n}}},i}\\  
        &+  \sum_{i=2{\tilde{n}}+1}^{3{\tilde{n}}} \frac{\sqrt{{\tilde{n}}}}{{\tilde{n}}_1} t_i \left( \alpha_0(X_i, \eta_i) - \mathbb{E}(\alpha_0(X, \eta) \mid Q \geq \tau_0)\right)
        + \sum_{i=2{\tilde{n}}+1}^{3{\tilde{n}}} \frac{\sqrt{{\tilde{n}}}}{{\tilde{n}}_1} \left( t_i - (1-t_i) K_{{\hat{\delta}_{\tilde{n}}},i} \right) \epsilon_i,
    \end{split}
\end{equation*}
where $A_5 = A_1^\top A_3 + A_4$. Since the terms are not independent, we need to apply the martingale central limit theorem \citep{billingsley1995probability} to establish normality jointly. Details are provided in Appendix \ref{app:proof-thm}.

%\input{sec_cooccurrence.tex}
%\input{sec_position.tex}
%\input{sec_failure.tex}

% \input{sec_related.tex}
%\input{sec_discussion.tex}

%\vspace{20pt}

%\textbf{Acknowledgements. } This work was supported in part by the Office of Naval Research under
% grant number N00014-23-1-2590, the National Science Foundation under
% grant No. 2231174, No. 2310831, No. 2428059, and a Michigan Institute
% for Data Science Propelling Original Data Science (PODS) grant.

\clearpage
%\putbib[attention]

%\end{bibunit}

%\bibliographystyle{alp}
\bibliography{icl}

\clearpage
\appendix

\begin{center}
\Large{\textbf{SUPPLEMENTARY MATERIAL}}
\end{center}

\section{Proof of Proposition \ref{prop:beta}}
\label{app:proof-prop}
In this section we present the proof of Proposition \ref{prop:beta}. First, it is easy to see that
$$\sqrt{{\tilde{n}}}(\hat{\gamma}_{\tilde{n}} - \gamma_0) = \sum_{i=1}^{\tilde{n}} \frac{1}{\sqrt{{\tilde{n}}}} \left( \frac{\tilde{Z}^\top \tilde{Z}}{{\tilde{n}}}\right)^{-1} Z_i \eta_i,$$
where $\tilde{Z} = (Z_1; Z_2; \cdots; Z_{\tilde{n}})^\top \in \mathbb{R}^{{\tilde{n}} \times d_Z}$. Also, it is a standard result in regression analysis that $\sqrt{{\tilde{n}}}(\hat{\gamma}_{\tilde{n}} - \gamma_0) \overset{d}{\longrightarrow} \mathcal{N}(0, \Sigma_\gamma)$, where $\Sigma_\gamma = \sigma_\eta^2 \cdot \plim\left( \frac{\tilde{Z}^\top \tilde{Z}}{{\tilde{n}}}\right)^{-1}$ due to Assumption \ref{asm:errormeanvar}.

Next, we define $\hat \delta_{\tilde{n}}:= \hat{\gamma}_{\tilde{n}} - \gamma_0$ and $\mathcal{W}_{\tilde{n}}$ to be the event where $||\hat{\delta}_{\tilde{n}}||_2 < 1$. Note that
\begin{align*}
    \mathbb{P}(\sqrt{{\tilde{n}}}(\hat{\gamma}_{\tilde{n}} - \gamma_0) \leq \square) = \mathbb{P}(\sqrt{{\tilde{n}}}(\hat{\gamma}_{\tilde{n}} - \gamma_0) \leq \square \cap \mathcal{W}_{\tilde{n}}) + \mathbb{P}(\sqrt{{\tilde{n}}}(\hat{\gamma}_{\tilde{n}} - \gamma_0) \leq \square \cap \mathcal{W}_{\tilde{n}}^c),
\end{align*}
where $\mathbb{P}(\sqrt{{\tilde{n}}}(\hat{\gamma}_{\tilde{n}} - \gamma_0) \leq \square \cap \mathcal{W}_{\tilde{n}}^c) \leq \mathbb{P}(\mathcal{W}_{\tilde{n}}^c) \to 0$ as $\tilde{n} \to \infty$ since $||\hat{\delta}_{\tilde{n}}||_2 \xrightarrow{P} 0$. Therefore, WLOG, we can work on the event $\mathcal{W}_{\tilde{n}}$ in the subsequent parts of the proof (e.g., establishing the asymptotic normality of $\hat{\beta}_{\tilde{n}}$ and $\hat{\theta}_{\tilde{n}}$ in Theorem \ref{thm:main}).

Now, let $I_2^C \subseteq I_2$ be the indices of observations in the second partition that belong to the control group. Recall that $\{\hat{\eta}_{(i)}\}$ is the order statistics of $\{\hat{\eta}_i\}_{i \in I_2^C}$, which induces an ordering on the $Y_i$'s, $X_i$'s, $\epsilon_i$'s, and $\eta_i$'s. We then have
\begin{equation*}
    Y_{(i+1)} - Y_{(i)} = (X_{(i+1)} - X_{(i)})^\top \beta_0 + \ell(\eta_{(i+1)}) - \ell(\eta_{(i)}) + \epsilon_{(i+1)} - \epsilon_{(i)},
\end{equation*}
compactly written as $\Delta Y = (\Delta X) \beta_0 + \Delta w$. It is easy to see that
\begin{equation}
\label{eq:betahatnormeq}
    \sqrt{{\tilde{n}}}(\hat{\beta}_{\tilde{n}} - \beta_0) = \left( \frac{1}{{\tilde{n}}} (\Delta X)^\top \Delta X \right)^{-1} \left( \frac{1}{\sqrt{{\tilde{n}}}} (\Delta X)^\top \Delta w \right).
\end{equation}

We first focus on the first term of Equation \eqref{eq:betahatnormeq}, assuming $\hat \delta_{\tilde{n}}:= \hat{\gamma}_{\tilde{n}} - \gamma_0$ is fixed (i.e., we conduct our analysis conditional on $I_1$). Define
$$C = \bigcup_{||\delta|| \leq 1} \textrm{supp}(\eta - Z^\top \delta).$$
It is clear that under Assumptions 1 and 2, we have $C$ is compact and $\textrm{supp}({\eta}) \subseteq C$.

Note that for observations in the control group, we can decompose $X$ into $g_{\hat{\delta}_{\tilde{n}}}(\hat{\eta}) + u_{\hat{\delta}_{\tilde{n}}}$, where $\hat{\eta} = \eta - Z^\top \hat{\delta}_{\tilde{n}}$ and $\mathbb{E}(u_{\hat{\delta}_{\tilde{n}}} \mid \hat{\eta}, Q < \tau_0) = 0$. Similarly, we also have $X = g_0(\eta) + u_0$, where $\mathbb{E}(u_0 \mid \eta, Q < \tau_0) = 0$. Using this decomposition, we have
\begin{equation*}
    \begin{split}
        \frac{1}{{\tilde{n}}} (\Delta X)^\top \Delta X &= \frac{1}{{\tilde{n}}} \sum_{i=1}^{|I_2^C|-1} (X_{(i+1)} - X_{(i)}) (X_{(i+1)} - X_{(i)})^\top \\
        &=  L + M + M^\top + P,
    \end{split}
\end{equation*}
where
\begin{equation*}
    \begin{split}
        &L = \frac{1}{{\tilde{n}}} \sum_{i=1}^{|I_2^C|-1} \left(g_{\hat{\delta}_{\tilde{n}}}({\hat{\eta}}_{(i+1)}) - g_{\hat{\delta}_{\tilde{n}}}({\hat{\eta}}_{(i)}) \right) \left(g_{\hat{\delta}_{\tilde{n}}}({\hat{\eta}}_{(i+1)}) - g_{\hat{\delta}_{\tilde{n}}}({\hat{\eta}}_{(i)}) \right)^\top, \\
        &M = \frac{1}{{\tilde{n}}} \sum_{i=1}^{|I_2^C|-1} \left(g_{\hat{\delta}_{\tilde{n}}}({\hat{\eta}}_{(i+1)}) - g_{\hat{\delta}_{\tilde{n}}}({\hat{\eta}}_{(i)}) \right) \left(u_{\hat{\delta}_{\tilde{n}_{(i+1)}}} - u_{\hat{\delta}_{\tilde{n}_{(i)}}} \right)^\top, \\
        &P = \frac{1}{{\tilde{n}}} \sum_{i=1}^{|I_2^C|-1} \left(u_{\hat{\delta}_{\tilde{n}_{(i+1)}}} - u_{\hat{\delta}_{\tilde{n}_{(i)}}}\right) \left(u_{\hat{\delta}_{\tilde{n}_{(i+1)}}} - u_{\hat{\delta}_{\tilde{n}_{(i)}}}\right)^\top.
    \end{split}
\end{equation*}

For each $i \leq j,k \leq d_X$, we have
\begin{equation*}
    \begin{split}
        \left|L_{j,k}\right| &= \left|\frac{1}{{\tilde{n}}} \sum_{i=1}^{|I_2^C|-1} \left(g_{j,\hat{\delta}_{\tilde{n}}}(\hat{\eta}_{(i+1)}) - g_{j,\hat{\delta}_{\tilde{n}}}(\hat{\eta}_{(i)}) \right) \left(g_{k,\hat{\delta}_{\tilde{n}}}(\hat{\eta}_{(i+1)}) - g_{k,\hat{\delta}_{\tilde{n}}}(\hat{\eta}_{(i)}) \right) \right| \\
        &\leq \frac{\nu_1}{{\tilde{n}}} \sum_{i=1}^{|I_2^C|-1} \left( \hat{\eta}_{(i+1)} - \hat{\eta}_{(i)} \right)^2 \\
        &= O_p(n^{-2})
    \end{split}
\end{equation*}
for some Lipschitz constant $\nu_1$ using the Cauchy-Schwarz inequality, Assumption \ref{asm:smoothness_conditional}, and the fact that the ordering is done on the $\hat{\eta}_i$'s. 

Similarly, we have
\begin{equation*}
    \begin{split}
        \left|M_{j,k}\right| &= \left|\frac{1}{{\tilde{n}}} \sum_{i=1}^{|I_2^C|-1} \left(g_{j,\hat{\delta}_{\tilde{n}}}(\hat{\eta}_{(i+1)}) - g_{j,\hat{\delta}_{\tilde{n}}}(\hat{\eta}_{(i)}) \right) \left(u_{k,\hat{\delta}_{\tilde{n}_{(i+1)}}} - u_{k,\hat{\delta}_{\tilde{n}_{(i)}}}\right) \right| \\
        &\leq \frac{\nu_2}{{\tilde{n}}} \sqrt{\sum_{i=1}^{|I_2^C|-1} \left( \hat{\eta}_{(i+1)} - \hat{\eta}_{(i)} \right)^2}  \sqrt{2 \sum_{i=1}^{|I_2^C|} (u_{k,\hat{\delta}_{\tilde{n}_{i}}})^2}\\
        &= O_p(n^{-1})
    \end{split}
\end{equation*}
for some Lipschitz constant $\nu_2$ using the Cauchy-Schwarz inequality and the elementary inequality $(a-b)^2 \leq 2(a^2+b^2)$, Assumptions \ref{asm:momentofxz} and \ref{asm:smoothness_conditional}, and the fact that the ordering is done on the $\hat{\eta}_i$'s. Specifically, Assumption \ref{asm:momentofxz} implies that $\mathbb{E}(u_{k,{\hat{\delta}_{\tilde{n}}}}^2 \mid Q < \tau_0) = \mathbb{E}(\textrm{var}(X_k \mid \hat{\eta}, Q < \tau_0) \mid Q < \tau_0)$ is finite. Thus, we have $L = o_p(1)$ and $M = o_p(1)$.

We now analyze $P$. Observe that
\begin{equation*}
    \begin{split}
    P &= \frac{1}{{\tilde{n}}} \sum_{i=1}^{|I_2^C|-1} \left(u_{\hat{\delta}_{\tilde{n}_{(i+1)}}} - u_{\hat{\delta}_{\tilde{n}_{(i)}}}\right) \left(u_{\hat{\delta}_{\tilde{n}_{(i+1)}}} - u_{\hat{\delta}_{\tilde{n}_{(i)}}}\right)^\top \\
    &= P_1 - P_2 - P_3 - P_3^\top,
    \end{split}
\end{equation*}
where 
\begin{equation*}
    \begin{split}
        &P_1 = \frac{2}{{\tilde{n}}} \sum_{i=1}^{|I_2^C|} u_{\hat{\delta}_{\tilde{n}_{i}}} u_{\hat{\delta}_{\tilde{n}_{i}}}^\top, \\
        &P_2 = \frac{1}{{\tilde{n}}} \left( u_{\hat{\delta}_{\tilde{n}_{(1)}}} u_{\hat{\delta}_{\tilde{n}_{(1)}}}^\top + u_{\hat{\delta}_{\tilde{n}_{(|I_2^C|)}}}u_{\hat{\delta}_{\tilde{n}_{(|I_2^C|)}}}^\top \right),\\
        &P_3 = \frac{1}{{\tilde{n}}} \sum_{i=1}^{|I_2^C|-1} u_{\hat{\delta}_{\tilde{n}_{(i+1)}}}u_{\hat{\delta}_{\tilde{n}_{(i)}}}^\top.
    \end{split}
\end{equation*}
It is easy to see that $P_1 - 2\bbP(Q < \tau_0)\Sigma_{u,{\hat{\delta}_{\tilde{n}}}} = o_p(1)$ where $\Sigma_{u,{\hat{\delta}_{\tilde{n}}}} = \mathbb{E}(\textrm{var}(X \mid \hat{\eta}, Q < \tau_0) \mid Q < \tau_0)$, and $P_2 = o_p(1)$. Moreover, we can also show that $P_3 = o_p(1)$. To see this, let $$W = \frac{1}{{\tilde{n}}} \sum_{i=1}^{|I_2^C|-1} u_{j,\hat{\delta}_{\tilde{n}_{(i+1)}}}u_{k,\hat{\delta}_{\tilde{n}_{(i)}}}$$
be the $(j,k)$-th entry of $P_3$. Then, $\mathbb{E}(W \mid Q_1 < \tau_0, \cdots, Q_{|I_2^C|} < \tau_0) = 0$ and 
\begin{align*}
    \textrm{var}(W \mid Q_1 < \tau_0, \cdots, Q_{|I_2^C|} < \tau_0) &= \frac{1}{{\tilde{n}}^2} \sum_{i=1}^{|I_2^C|-1} \mathbb{E}(u_{j,{\hat{\delta}_{\tilde{n}}}}^2 \mid Q < \tau_0) \mathbb{E}(u_{k,{\hat{\delta}_{\tilde{n}}}}^2 \mid Q < \tau_0) \\
    &= O(n^{-1})
\end{align*}
due to Assumption \ref{asm:momentofxz}. This implies $P_3 = o_p(1)$ since $W = o_p(1)$ by Chebyshev's inequality.

So far, we have shown that conditional on $\hat{\delta}_{\tilde{n}}$, we have
$$\frac{1}{{\tilde{n}}} (\Delta X)^\top \Delta X - 2\bbP(Q < \tau_0) \Sigma_{u,{\hat{\delta}_{\tilde{n}}}} = o_p(1),$$
where $\Sigma_{u,{\hat{\delta}_{\tilde{n}}}} = \mathbb{E}(\textrm{var}(X \mid \hat{\eta}, Q < \tau_0) \mid Q < \tau_0)$ and $\hat{\eta} = \eta - Z^\top \hat{\delta}_{\tilde{n}}$. 

We first show the following lemma:
\vspace{2mm}
\begin{lemma}
\label{lemma:cont}
For any sequence $\{\hat{\delta}_{\tilde{n}}\}_{\tilde{n} \geq 1}$, where $||\hat \delta_{\tilde{n}}||_2 \leq 1$ for every $\tilde{n}$, that converges to $0$, we have $\Sigma_{u,{\hat{\delta}_{\tilde{n}}}} \rightarrow \Sigma_u := \mathbb{E}(\textrm{var}(X \mid {\eta}, Q < \tau_0) \mid Q < \tau_0)$ as $\tilde{n} \rightarrow \infty$.
\end{lemma}
\begin{proof}
Note that for any $\tilde{n}$, we have
$$\Sigma_{u,{\hat{\delta}_{\tilde{n}}}} = \int_C \textrm{var} \left(X \mid \eta - Z^\top \hat{\delta}_{\tilde{n}} = t , Q < \tau_0 \right) f_{\eta - Z^\top \hat{\delta}_{\tilde{n}} \mid Q < \tau_0}(t) \mathds{1}_{t \in \textrm{supp}(\eta - Z^\top \hat{\delta}_{\tilde{n}})} dt.$$
Moreover,
$$\Sigma_{u} = \int_C \textrm{var} \left(X \mid \eta = t , Q < \tau_0 \right) f_{\eta \mid Q < \tau_0}(t) \mathds{1}_{t \in \textrm{supp}(\eta)} dt.$$
We first prove that for every $t \in C$, $\textrm{var} \left(X \mid \eta - Z^\top \hat{\delta}_{\tilde{n}} = t , Q < \tau_0 \right) \rightarrow \textrm{var} \left(X \mid \eta = t , Q < \tau_0 \right)$. We consider two cases: (1) $t \in \textrm{supp}(\eta)$; and (2) $t \in C \setminus \textrm{supp}(\eta)$. For the first case, the statement clearly follows from Assumptions 1, 2 and 3. For the second case, note that Assumption 1, 2 and the fact that $\hat{\delta}_{\tilde{n}} \rightarrow 0$ implies that we can find some $n^*$ such that for every $\tilde{n} \geq n^*$, we have $t \notin \textrm{supp}({\eta - Z^\top \hat{\delta}_{\tilde{n}}})$. The conclusion is thus immediate since $\textrm{var} \left(X \mid \eta - Z^\top \hat{\delta}_{\tilde{n}} = t , Q < \tau_0 \right) = \textrm{var}(X)$ for every $\tilde{n} \geq n^*$, and $\textrm{var} \left(X \mid \eta = t , Q < \tau_0 \right) = \textrm{var}(X)$.

In a similar manner, for every $t \in C$, we can show that $f_{\eta - Z^\top \hat{\delta}_{\tilde{n}} \mid Q < \tau_0}(t) \rightarrow f_{\eta \mid Q < \tau_0}(t)$ under Assumptions 1, 2 and 6, and $\mathds{1}_{t \in \textrm{supp}(\eta - Z^\top \hat{\delta}_{\tilde{n}})} \rightarrow \mathds{1}_{t \in \textrm{supp}(\eta)}$ under Assumptions 1 and 2. The lemma thus follows upon applying Lebesgue's dominated convergence theorem (DCT) under Assumptions 1 and 2 using the fact that $C$ is compact. This completes the proof.
\end{proof}

We now use Lemma \ref{lemma:cont} to prove the following lemma:
\begin{lemma}
\label{lemma:uncond}
    As ${\tilde{n}} \rightarrow \infty$, we have
    $$\frac{1}{{\tilde{n}}} (\Delta X)^\top \Delta X \overset{P}{\longrightarrow} 2\bbP(Q < \tau_0) \Sigma_{u},$$
    where  $\Sigma_u = \mathbb{E}(\textrm{var}(X \mid {\eta}, Q < \tau_0) \mid Q < \tau_0)$.
\end{lemma}
\begin{proof}
Let $\Lambda_{\tilde{n}} = \frac{1}{{\tilde{n}}} (\Delta X)^\top \Delta X$. For every coordinate $t$ and $\epsilon > 0$, we have 
    \begin{equation*}
        \mathbb{P}\left(\left| \Lambda_{\tilde{n},t} - 2\bbP(Q < \tau_0) \Sigma_{u, \hat{\delta}_{\tilde{n}}, t} \right| \geq \epsilon \mid \hat{\delta}_{\tilde{n}} \right) \rightarrow 0.
    \end{equation*}
Applying Lebesgue's DCT, we have
\begin{equation*}
        \mathbb{P}\left(\left| \Lambda_{\tilde{n},t} - 2\bbP(Q < \tau_0) \Sigma_{u, \hat{\delta}_{\tilde{n}}, t} \right| \geq \epsilon \right) \rightarrow 0.
    \end{equation*}
Fix any $\epsilon > 0$. From Lemma \ref{lemma:cont}, we know that for every coordinate $t$, there exists some $\xi > 0$ such that $|\Sigma_{u, \delta, t} - \Sigma_{u, t}| < \frac{\epsilon}{4 \bbP(Q < \tau_0)}$ whenever $||\delta||_2 < \xi$. Now, observe that
\begin{equation*}
    \begin{split}
        &\mathbb{P}\left(\left| \Lambda_{\tilde{n},t} - 2\bbP(Q < \tau_0) \Sigma_{u, t} \right| \geq \epsilon\right) \\
    &= \mathbb{P}\left(\left| \Lambda_{\tilde{n},t} - 2\bbP(Q < \tau_0) \Sigma_{u, t} \right| \geq \epsilon \cap ||\hat{\delta}_{\tilde{n}}||_2 < \xi\right) + \mathbb{P}\left(\left| \Lambda_{\tilde{n},t} - 2\bbP(Q < \tau_0) \Sigma_{u, t} \right| \geq \epsilon \cap ||\hat{\delta}_{\tilde{n}}||_2 \geq \xi\right)  \\
    &\leq \mathbb{P}\left(\left| \Lambda_{\tilde{n},t} - 2\bbP(Q < \tau_0) \Sigma_{u, \hat{\delta}_{\tilde{n}}, t} \right| \geq \frac{\epsilon}{2} \cap ||\hat{\delta}_{\tilde{n}}||_2 < \xi \right) + \mathbb{P}(||\hat{\delta}_{\tilde{n}}||_2 \geq \xi) \\
    &\leq \mathbb{P}\left(\left| \Lambda_{\tilde{n},t} - 2\bbP(Q < \tau_0) \Sigma_{u, \hat{\delta}_{\tilde{n}}, t} \right| \geq \frac{\epsilon}{2}\right) + \mathbb{P}(||\hat{\delta}_{\tilde{n}}||_2 \geq \xi).
    \end{split}
\end{equation*}
Note that the first term goes to $0$ as established above, and so does the second term since $\hat{\delta}_{\tilde{n}} \overset{P}{\longrightarrow} 0$. This implies $\Lambda_{\tilde{n},t} \overset{P}{\longrightarrow} 2 \bbP(Q < \tau_0) \Sigma_{u,t}$ for every coordinate $t$, which means $\Lambda_{\tilde{n}} \overset{P}{\longrightarrow} 2 \bbP(Q < \tau_0) \Sigma_u$. This completes the proof.
\end{proof}
From Lemma \ref{lemma:uncond}, an application of the continuous mapping theorem yields 
$$\left( \frac{1}{{\tilde{n}}} (\Delta X)^\top \Delta X \right)^{-1} \overset{P}{\longrightarrow} \frac{1}{2\bbP(Q < \tau_0)}\Sigma_u^{-1}.$$

We now work on the second term of Equation \eqref{eq:betahatnormeq}. Note that we can decompose it as
\begin{equation*}
    \begin{split}
        \frac{1}{\sqrt{{\tilde{n}}}} (\Delta X)^\top \Delta w &= \frac{1}{\sqrt{{\tilde{n}}}} \sum_{i=1}^{|I_2^C|-1} (X_{(i+1)} - X_{(i)}) \left(\ell(\eta_{(i+1)}) - \ell(\eta_{(i)}) + \epsilon_{(i+1)} - \epsilon_{(i)} \right) \\
        &= E + F + G + H,
    \end{split}
\end{equation*}
where
\begin{equation*}
   \begin{split}
        &E = \frac{1}{\sqrt{{\tilde{n}}}} \sum_{i=1}^{|I_2^C|-1} \left(\ell(\eta_{(i+1)}) - \ell(\eta_{(i)}) \right) \left( g_{{\hat{\delta}_{\tilde{n}}}}(\hat{\eta}_{(i+1)}) - g_{{\hat{\delta}_{\tilde{n}}}}(\hat{\eta}_{(i)}) \right),\\
    &F = \frac{1}{\sqrt{{\tilde{n}}}} \sum_{i=1}^{|I_2^C|-1} \left(\ell(\eta_{(i+1)}) - \ell(\eta_{(i)}) \right) \left( u_{\hat{\delta}_{\tilde{n}_{(i+1)}}} - u_{\hat{\delta}_{\tilde{n}_{(i)}}}\right),\\
    &G = \frac{1}{\sqrt{{\tilde{n}}}} \sum_{i=1}^{|I_2^C|-1} \left(\epsilon_{(i+1)} - \epsilon_{(i)} \right) \left( g_{{\hat{\delta}_{\tilde{n}}}}(\hat{\eta}_{(i+1)}) - g_{{\hat{\delta}_{\tilde{n}}}}(\hat{\eta}_{(i)}) \right),\\
    &H = \frac{1}{\sqrt{{\tilde{n}}}} \sum_{i=1}^{|I_2^C|-1} \left(\epsilon_{(i+1)} - \epsilon_{(i)} \right) \left( u_{\hat{\delta}_{\tilde{n}_{(i+1)}}} - u_{\hat{\delta}_{\tilde{n}_{(i)}}} \right).
    \end{split}
\end{equation*}

First, observe that 
\begin{equation*}
    \begin{split}
        \sum_{i=1}^{|I_2^C|-1} \left( \ell(\eta_{(i+1)}) - \ell(\eta_{(i)})\right)^2 &\leq \nu_3^2 \sum_{i=1}^{I_2^C|-1} \left( |\eta_{(i+1)} - \hat{\eta}_{(i+1)} | + |\eta_{(i)} - \hat{\eta}_{(i)} | + |\hat{\eta}_{(i+1)} - \hat{\eta}_{(i)} |\right)^2 \\
        &\leq 3\nu_3^2 \sum_{i=1}^{I_2^C|-1} \left( (\eta_{(i+1)} - \hat{\eta}_{(i+1)})^2 + (\eta_{(i)} - \hat{\eta}_{(i)})^2 + (\hat{\eta}_{(i+1)} - \hat{\eta}_{(i)})^2 \right) \\
        &\leq 6\nu_3^2 \sum_{i=1}^{|I_2^C|} ((Z_i^C)^\top (\hat{\gamma}_{\tilde{n}} - \gamma_0))^2 + 3\nu_3^2 \sum_{i=1}^{|I_2^C|-1} (\hat{\eta}_{(i+1)} - \hat{\eta}_{(i)})^2 \\
        &= O_p(n^{-1}) O_p(n) + O_p(n^{-1}) \\
        &= O_p(1).
    \end{split}
\end{equation*}
for some Lipschitz constant $\nu_3$. Here, we used Assumptions \ref{asm:momentofxz}, \ref{asm:errormeanvar} and \ref{asm:fbounded}, the triangle inequality and the elementary inequality $(a+b+c)^2 \leq 3(a^2+b^2+c^2)$, as well as the fact that $||{\hat{\delta}_{\tilde{n}}}||_2^2 = O_p(n^{-1})$ and the ordering is done on the $\hat{\eta}_i$'s. 

Now, we can easily utilize the Cauchy-Schwarz inequality to show that $E = O_p(n^{-1}) = o_p(1)$ using Assumption \ref{asm:smoothness_conditional}, the above result, and the fact that the ordering is done on the $\hat{\eta}_i$'s. Similarly, we can show $G = O_p(n^{-1/2}) = o_p(1)$ using the elementary $(a-b)^2 \leq 2(a^2+b^2)$ under Assumptions \ref{asm:errormeanvar} and \ref{asm:smoothness_conditional}. Let us look at $F$. We can rewrite it as
\begin{equation*}
    \begin{split}
        F &= \frac{1}{\sqrt{{\tilde{n}}}} \sum_{i=1}^{|I_2^C|-1}\left( \ell(\eta_{(i+1)}) - \ell(\hat{\eta}_{(i+1)})\right) \left( u_{\hat{\delta}_{\tilde{n}_{(i+1)}}} - u_{\hat{\delta}_{\tilde{n}_{(i)}}}\right) \\
        &- \frac{1}{\sqrt{{\tilde{n}}}} \sum_{i=1}^{|I_2^C|-1}\left( \ell(\eta_{(i)}) - \ell(\hat{\eta}_{(i)})\right) \left( u_{\hat{\delta}_{\tilde{n}_{(i+1)}}} - u_{\hat{\delta}_{\tilde{n}_{(i)}}}\right) \\
        &+ \frac{1}{\sqrt{{\tilde{n}}}} \sum_{i=1}^{|I_2^C|-1}\left( \ell(\hat{\eta}_{(i+1)}) - \ell(\hat{\eta}_{(i)})\right) \left( u_{\hat{\delta}_{\tilde{n}_{(i+1)}}} - u_{\hat{\delta}_{\tilde{n}_{(i)}}}\right).
    \end{split}
\end{equation*}

A straightforward application of the Cauchy-Schwarz inequality allows us to show that the third term in $F$ is $o_p(1)$ under Assumptions \ref{asm:momentofxz}, \ref{asm:errormeanvar} and \ref{asm:fbounded} and the elementary inequality $(a-b)^2 \leq 2(a^2+b^2)$. 
Moreover, observe that 
\begin{equation*}
    \ell(\eta_{(i)}) - \ell(\hat{\eta}_{(i)}) = (\eta_{(i)} - \hat{\eta}_{(i)}) \ell'(\hat{\eta}_{(i)}) + (\eta_{(i)} - \hat{\eta}_{(i)})^2 \ell''(\tilde{\eta}_{(i)})
\end{equation*}
for some $\tilde{\eta}_{(i)}$ between $\eta_{(i)}$ and $\hat{\eta}_{(i)}$, and that
\begin{equation*}
    \left| \frac{1}{\sqrt{{\tilde{n}}}} \sum_{i=1}^{|I_2^C|-1}  (\eta_{(i+1)} - \hat{\eta}_{(i+1)})^2 \ell''(\tilde{\eta}_{(i+1)}) (u_{\hat{\delta}_{\tilde{n}_{(i+1)}}} - u_{\hat{\delta}_{\tilde{n}_{(i)}}})  \right| = O_p(n^{-1/2}) = o_p(1)
\end{equation*}
using the Cauchy-Schwarz inequality under Assumptions \ref{asm:momentofxz}, \ref{asm:errormeanvar} and \ref{asm:fbounded}. 

Similar to $X$, we can decompose $Z$ for observations in the control group into $q_{\hat{\delta}_{\tilde{n}}}(\hat{\eta}) + w_{\hat{\delta}_{\tilde{n}}}$, where $\mathbb{E}(w_{\hat{\delta}_{\tilde{n}}} \mid \hat{\eta}, Q < \tau_0) = 0$ and $\hat{\eta} = \eta - Z^\top \hat{\delta}_{\tilde{n}}$. Similarly, we also have $Z = q_0(\eta) + w_0$, where $\mathbb{E}(w_0 \mid \eta, Q < \tau_0) = 0$. Omitting $o_p(1)$ terms, we can further rewrite $F$ as
\begin{equation*}
    \begin{split}
         F &= \frac{1}{\sqrt{{\tilde{n}}}} \sum_{i=1}^{|I_2^C|-1}  (\eta_{(i+1)} - \hat{\eta}_{(i+1)}) \ell'(\hat{\eta}_{(i+1)}) ( u_{\hat{\delta}_{\tilde{n}_{(i+1)}}} - u_{\hat{\delta}_{\tilde{n}_{(i)}}}) - \frac{1}{\sqrt{{\tilde{n}}}} \sum_{i=1}^{|I_2^C|-1}  (\eta_{(i)} - \hat{\eta}_{(i)}) \ell'(\hat{\eta}_{(i)}) ( u_{\hat{\delta}_{\tilde{n}_{(i+1)}}} - u_{\hat{\delta}_{\tilde{n}_{(i)}}}) \\
        &= \frac{1}{{\tilde{n}}} \left( \sum_{i=1}^{|I_2^C|-1} \left( u_{\hat{\delta}_{\tilde{n}_{(i+1)}}} - u_{\hat{\delta}_{\tilde{n}_{(i)}}}\right) \left( \ell'(\hat{\eta}_{(i+1)}) (Z_{(i+1)})^\top - \ell'(\hat{\eta}_{(i)}) (Z_{(i)})^\top \right)\right) \left(\sqrt{{\tilde{n}}}(\hat{\gamma}_{\tilde{n}} - \gamma_0)\right) \\
        &= \frac{1}{{\tilde{n}}} \left( \sum_{i=1}^{|I_2^C|-1} \left( u_{\hat{\delta}_{\tilde{n}_{(i+1)}}} - u_{\hat{\delta}_{\tilde{n}_{(i)}}}\right) \left( \ell'(\hat{\eta}_{(i+1)}) (w_{\hat{\delta}_{\tilde{n}_{(i+1)}}})^\top - \ell'(\hat{\eta}_{(i)}) (w_{\hat{\delta}_{\tilde{n}_{(i)}}})^\top \right)\right) \left(\sqrt{{\tilde{n}}}(\hat{\gamma}_{\tilde{n}} - \gamma_0)\right) \\
        &+ \frac{1}{{\tilde{n}}} \left( \sum_{i=1}^{|I_2^C|-1} \left( u_{\hat{\delta}_{\tilde{n}_{(i+1)}}} - u_{\hat{\delta}_{\tilde{n}_{(i)}}}\right) \ell'(\hat{\eta}_{(i+1)}) \left( q_{{\hat{\delta}_{\tilde{n}}}}(\hat{\eta}_{(i+1)}) - q_{{\hat{\delta}_{\tilde{n}}}}(\hat{\eta}_{(i)})\right)^\top \right)
        \left(\sqrt{{\tilde{n}}}(\hat{\gamma}_{\tilde{n}} - \gamma_0)\right) \\
        &+ \frac{1}{{\tilde{n}}} \left( \sum_{i=1}^{|I_2^C|-1} \left( u_{\hat{\delta}_{\tilde{n}_{(i+1)}}} - u_{\hat{\delta}_{\tilde{n}_{(i)}}}\right) \left(\ell'(\hat{\eta}_{(i+1)}) - \ell'(\hat{\eta}_{(i)}) \right) \left( q_{{\hat{\delta}_{\tilde{n}}}}(\hat{\eta}_{(i)})\right)^\top \right)
    \left(\sqrt{{\tilde{n}}}(\hat{\gamma}_{\tilde{n}} - \gamma_0)\right).
    \end{split}
\end{equation*}
The last two terms of the third equality can be shown to be $o_p(1)$ using the Cauchy-Schwarz inequality under Assumptions \ref{asm:momentofxz}, \ref{asm:errormeanvar} and \ref{asm:fbounded}. Again, omitting $o_p(1)$ terms, we have $F = \overline{F} \left(\sqrt{{\tilde{n}}}({\hat{\gamma}_{\tilde{n}}} - \gamma_0)\right)$, where
$$\overline{F} = \frac{1}{{\tilde{n}}} \left( \sum_{i=1}^{|I_2^C|-1} \left( u_{\hat{\delta}_{\tilde{n}_{(i+1)}}} - u_{\hat{\delta}_{\tilde{n}_{(i)}}}\right) \left( \ell'(\hat{\eta}_{(i+1)}) (w_{\hat{\delta}_{\tilde{n}_{(i+1)}}})^\top - \ell'(\hat{\eta}_{(i)}) (w_{\hat{\delta}_{\tilde{n}_{(i)}}})^\top \right)\right).$$

Conditional on $\hat{\delta}_{\tilde{n}}$ (i.e., $I_1$), we can show that 
$$\overline{F} - 2 \bbP(Q < \tau_0) \mathbb{E}\left(\ell'(\hat{\eta}) u_{\hat{\delta}_{\tilde{n}}} (w_{\hat{\delta}_{\tilde{n}}})^\top \mid Q < \tau_0 \right) = o_p(1),$$
where $\hat{\eta} = \eta - Z^\top \hat{\delta}_{\tilde{n}}$, under Assumptions \ref{asm:momentofxz} and \ref{asm:fbounded} using the same method as for the term $P$. We now prove the following lemma:
\vspace{2mm}
\begin{lemma}
\label{lem:cont2}
For any sequence $\{\hat{\delta}_{\tilde{n}}\}_{\tilde{n} \geq 1}$, where $||\hat \delta_{\tilde{n}}||_2 \leq 1$ for every $\tilde{n}$, that converges to $0$, we have 
$$\mathbb{E}\left(\ell'(\hat{\eta}) u_{\hat{\delta}_{\tilde{n}}} (w_{\hat{\delta}_{\tilde{n}}})^\top \mid Q < \tau_0 \right) \rightarrow \mathbb{E}\left(\ell'({\eta}) u_0 (w_0)^\top \mid Q < \tau_0 \right),$$
where $u_{\hat{\delta}_{\tilde{n}}} = X - \mathbb{E}(X \mid \hat{\eta}, Q < \tau_0)$, $w_{\hat{\delta}_{\tilde{n}}} = Z - \mathbb{E}(Z \mid \hat{\eta}, Q < \tau_0)$, $u_0 = X - \mathbb{E}(X \mid \eta, Q < \tau_0)$, and $w_0 = Z - \mathbb{E}(Z \mid \eta, Q < \tau_0)$.
\end{lemma}
\begin{proof}
It is easy to see that the statement we want to show is equivalent to
$$\mathbb{E}\left( \ell'(\hat{\eta}) \textrm{cov}(X, Z \mid \hat{\eta}, Q < \tau_0) \mid Q < \tau_0\right) \rightarrow \mathbb{E}\left( \ell'({\eta}) \textrm{cov}(X, Z \mid {\eta}, Q < \tau_0) \mid Q < \tau_0\right).$$
Note that for any $\tilde{n}$, we have
\begin{align*}
&\mathbb{E}\left( \ell'(\hat{\eta}) \textrm{cov}(X, Z \mid \hat{\eta}, Q < \tau_0) \mid Q < \tau_0\right) \\
&= \int_C \ell'(t) \textrm{cov}(X,Z \mid \eta - Z^\top \hat{\delta}_{\tilde{n}} = t , Q < \tau_0) f_{\eta - Z^\top \hat{\delta}_{\tilde{n}} \mid Q < \tau_0}(t) \mathds{1}_{t \in \textrm{supp}(\eta - Z^\top \hat{\delta}_{\tilde{n}})} dt.
\end{align*}
Moreover, 
\begin{align*}
&\mathbb{E}\left( \ell'({\eta}) \textrm{cov}(X, Z \mid {\eta}, Q < \tau_0) \mid Q < \tau_0\right) \\
&= \int_C \ell'(t) \textrm{cov}(X,Z \mid \eta = t, Q < \tau_0) f_{\eta  \mid Q < \tau_0}(t) \mathds{1}_{t \in \textrm{supp}(\eta)} dt.
\end{align*}
Using the same method as in Lemma \ref{lemma:cont}, the conclusion follows via Lebesgue's DCT under Assumptions 1, 2, 3, 5 and 6 using the fact that $C$ is compact. This completes the proof.
\end{proof}
From here, a simple adaptation of the proof of Lemma \ref{lemma:uncond} yields
$$\overline{F} \overset{P}{\longrightarrow} 2 \bbP(Q < \tau_0) \mathbb{E}\left(\ell'({\eta}) u_{0} (w_{0})^\top \mid Q < \tau_0 \right).$$
Therefore, omitting $o_p(1)$ terms, we can write $F$ as
\begin{equation*}
    F = 2 \bbP(Q < \tau_0) \mathbb{E}\left(\ell'({\eta}) u_0 (w_0)^\top \mid Q < \tau_0 \right) \left(\sqrt{{\tilde{n}}}(\hat{\gamma}_{\tilde{n}} - \gamma_0)\right)
    = \sum_{i=1}^{\tilde{n}} \frac{1}{\sqrt{{\tilde{n}}}} A_1  \left( \frac{\tilde{Z}^\top \tilde{Z}}{{\tilde{n}}}\right)^{-1} Z_i \eta_i,
\end{equation*}
where $A_1 = 2 \bbP(Q < \tau_0) \mathbb{E}\left(\ell'({\eta}) u_0 w_0^\top \mid Q < \tau_0 \right)$. 

Lastly, we consider $H$. Note that $H$ can be written as
\begin{equation*}
    \begin{split}
        H &= \frac{1}{\sqrt{{\tilde{n}}}} \sum_{i=1}^{|I_2^C|-1}\left(  \epsilon_{(i+1)} - \epsilon_{(i)} \right) \left( u_{\hat{\delta}_{\tilde{n}_{(i+1)}}} - u_{\hat{\delta}_{\tilde{n}_{(i)}}}  \right) \\
        &= \frac{1}{\sqrt{{\tilde{n}}}} \Big( \epsilon_{(1)} (u_{\hat{\delta}_{\tilde{n}_{(1)}}} - u_{\hat{\delta}_{\tilde{n}_{(2)}}}) + \epsilon_{(2)} (2u_{\hat{\delta}_{\tilde{n}_{(2)}}} - u_{\hat{\delta}_{\tilde{n}_{(1)}}} - u_{\hat{\delta}_{\tilde{n}_{(3)}}}) + \cdots \\ &+ \epsilon_{(|I_2^C|-1)} (2u_{\hat{\delta}_{\tilde{n}_{(|I_2^C| - 1)}}} - u_{\hat{\delta}_{\tilde{n}_{(|I_2^C| - 2)}}} - u_{\hat{\delta}_{\tilde{n}_{(|I_2^C|)}}})
        + \epsilon_{(|I_2^C|)} (u_{\hat{\delta}_{\tilde{n}_{(|I_2^C|)}}} - u_{\hat{\delta}_{\tilde{n}_{(|I_2^C|-1)}}}) \Big) \\
        &:= \frac{1}{\sqrt{{\tilde{n}}}} \sum_{i=1}^{|I_2^C|} \epsilon_{(i)} a_{\hat{\delta}_{\tilde{n}_{(i)}}}  \\
        &:= \frac{1}{\sqrt{{\tilde{n}}}} \sum_{i={\tilde{n}}+1}^{2{\tilde{n}}} \epsilon_i a_{{\hat{\delta}_{\tilde{n}}},i},
    \end{split}
\end{equation*}
where $a_{{\hat{\delta}_{\tilde{n}}},i} = 0$ for each observation $i$ in the treatment group. Now, let $c \in \mathbb{R}^{d_X}$ be an arbitrary vector such that $||c||_2 = 1$. Up to $o_p(1)$ terms, we have
\begin{equation*}
    \begin{split}
         c^\top \frac{1}{\sqrt{{\tilde{n}}}} (\Delta X)^\top \Delta w &= \sum_{i=1}^{\tilde{n}} \frac{1}{\sqrt{{\tilde{n}}}} c^\top A_1 \left( \frac{\tilde{Z}^\top \tilde{Z}}{{\tilde{n}}}\right)^{-1} Z_i \eta_i + \sum_{i={\tilde{n}}+1}^{2{\tilde{n}}} \frac{1}{\sqrt{{\tilde{n}}}}  c^\top \epsilon_i a_{{\hat{\delta}_{\tilde{n}}},i} \\
        &= \xi_{{\tilde{n}},1} + \xi_{{\tilde{n}},2} + \cdots + \xi_{{\tilde{n}},2{\tilde{n}}}.
    \end{split}
\end{equation*}
Now, we consider the following $\sigma$-fields: $\mathcal{F}_{{\tilde{n}},1} = \sigma(Z_{1:{\tilde{n}}}, \eta_1)$,$\cdots, \mathcal{F}_{{\tilde{n}},{\tilde{n}}} = \sigma(Z_{1:{\tilde{n}}}, \eta_{1:{\tilde{n}}})$, $\mathcal{F}_{{\tilde{n}},{\tilde{n}}+1} = \sigma(Z_{1:2{\tilde{n}}}, \eta_{1:2{\tilde{n}}}, X_{1:2{\tilde{n}}},\epsilon_{{\tilde{n}}+1}), \cdots$, and $\mathcal{F}_{{\tilde{n}},2{\tilde{n}}} = \sigma(Z_{1:2{\tilde{n}}}, \eta_{1:2{\tilde{n}}}, X_{1:2{\tilde{n}}},\epsilon_{{\tilde{n}}+1:2{\tilde{n}}})$. For each ${\tilde{n}}$, it is easy to see that
\begin{equation*}
    \left\{ \sum_{j=1}^i \xi_{{\tilde{n}},j}, \mathcal{F}_{{\tilde{n}},i}, 1 \leq i \leq 2{\tilde{n}} \right\}
\end{equation*}
is a martingale. We now use \citeauthor{billingsley1995probability}'s (\citeyear{billingsley1995probability}) martingale central limit theorem. Note that using Assumption \ref{asm:errormeanvar}, we have
\begin{equation*}
    \begin{split}
        \sum_{i=1}^{\tilde{n}} \mathbb{E}(\xi_{{\tilde{n}},i}^2 \mid \mathcal{F}_{{\tilde{n}},i-1}) &=  \sum_{i=1}^{\tilde{n}} \mathbb{E}\left(\left( \frac{1}{\sqrt{{\tilde{n}}}} c^\top A_1 \left( \frac{\tilde{Z}^\top \tilde{Z}}{{\tilde{n}}}\right)^{-1} Z_i \eta_i \right)^2 \mid Z_{1:{\tilde{n}}}, \eta_{1:i-1}\right) \\
        &= \sigma_\eta^2 c^\top A_1 \left(\frac{\tilde{Z}^\top \tilde{Z}}{{\tilde{n}}}\right)^{-1} A_1^\top c \\
        &\overset{P}{\longrightarrow} c^\top A_1 \Sigma_\gamma A_1^\top c
    \end{split}
\end{equation*}
and
\begin{equation*}
    \begin{split}
         \sum_{i={\tilde{n}}+1}^{2{\tilde{n}}} \mathbb{E}(\xi_{{\tilde{n}},i}^2 \mid \mathcal{F}_{{\tilde{n}},i-1}) &=   \sum_{i={\tilde{n}}+1}^{2{\tilde{n}}} \mathbb{E}\left( \left( \frac{1}{\sqrt{{\tilde{n}}}}  c^\top \epsilon_i a_{{\hat{\delta}_{\tilde{n}}},i}\right)^2 \mid Z_{1:2{\tilde{n}}}, \eta_{1:2{\tilde{n}}}, X_{1:2{\tilde{n}}}, \epsilon_{{\tilde{n}}+1:i-1} \right) \\
        &= \sigma_\epsilon^2 \sum_{i={\tilde{n}}+1}^{2{\tilde{n}}} \frac{1}{{\tilde{n}}} (c^\top a_{{{\hat{\delta}_{\tilde{n}}}},i})^2.
    \end{split}
\end{equation*}
Conditional on $\hat{\delta}_{\tilde{n}}$ (i.e., $I_1$), it is easy to show that
\begin{equation*}
    \frac{1}{{\tilde{n}}} \sum_{i={\tilde{n}}+1}^{2{\tilde{n}}} (c^\top a_{{{\hat{\delta}_{\tilde{n}}}},i})^2  - 6 \bbP(Q < \tau_0) c^\top \Sigma_{u, {{\hat{\delta}_{\tilde{n}}}}} c = o_p(1)
\end{equation*}
under Assumptions \ref{asm:momentofxz} using the same method as for the term $P$. Following the proof of Lemmas \ref{lemma:cont} and \ref{lemma:uncond}, we have 
\begin{equation*}
    \frac{1}{{\tilde{n}}} \sum_{i={\tilde{n}}+1}^{2{\tilde{n}}} (c^\top a_{{\hat{\delta}_{\tilde{n}}},i})^2  \overset{P}{\longrightarrow} 6 \bbP(Q < \tau_0) c^\top \Sigma_{u} c,
\end{equation*}
whence
\begin{equation*}
    \begin{split}
         \sum_{i={\tilde{n}}+1}^{2{\tilde{n}}} \mathbb{E}(\xi_{{\tilde{n}},i}^2 \mid \mathcal{F}_{{\tilde{n}},i-1}) \overset{P}{\longrightarrow} 6 \bbP(Q < \tau_0) \sigma_\epsilon^2 c^\top \Sigma_{u} c.
    \end{split}
\end{equation*}
Therefore, the martingale central limit theorem and Cramer-Wold device gives us
\begin{equation*}
    \frac{1}{\sqrt{{\tilde{n}}}} (\Delta X)^\top \Delta w \overset{d}{\longrightarrow} \mathcal{N}(0, A_1 \Sigma_\gamma A_1^\top + 6 \bbP(Q < \tau_0) \sigma_\epsilon^2  \Sigma_u).
\end{equation*}
The asymptotic normality of $\sqrt{{\tilde{n}}}(\hat{\beta}_{\tilde{n}} - \beta_0)$ is now just a consequence of Slutsky's theorem. Concretely, we have
\begin{equation}
\label{eq:betahat_var}
    \sqrt{{\tilde{n}}}(\hat{\beta}_{\tilde{n}} - \beta_0) \overset{d}{\longrightarrow} \mathcal{N}\left(0, 
    \frac{1}{4\bbP(Q < \tau_0)^2} \Sigma_u^{-1} \zeta \Sigma_u^{-1} \right),
\end{equation}
where 
\begin{itemize}
    \item $\zeta = 4 (\bbP(Q < \tau_0))^2 \mathbb{E}\left(\ell'(\eta) u_0 w_0^\top \mid Q < \tau_0 \right) \Sigma_\gamma \mathbb{E}\left(\ell'(\eta) w_0 u_0^\top \mid Q < \tau_0 \right) + 6 \bbP(Q < \tau_0) \sigma_\epsilon^2 \Sigma_u$,
    \item $\Sigma_u = \mathbb{E}(\textrm{var}(X \mid {\eta}, Q < \tau_0) \mid Q < \tau_0)$,
    \item $u_0 = X - \mathbb{E}(X \mid \eta, Q < \tau_0)$,
    \item $w_0 = Z - \mathbb{E}(Z \mid \eta, Q < \tau_0)$.
\end{itemize}

To finish the proof, we need to establish the Lindeberg's condition for the martingale central limit theorem. A pair of sufficient conditions is $\sum_{i=1}^{{\tilde{n}}} \mathbb{E}(|\xi_{{\tilde{n}},i}|^3) \rightarrow 0$ and $\sum_{i={\tilde{n}}+1}^{2{\tilde{n}}} \mathbb{E}(|\xi_{{\tilde{n}},i}|^3) \rightarrow 0$ as ${\tilde{n}} \rightarrow \infty$, whence the Lyapunov's (and consequently Lindeberg's) condition is satisfied.

Observe that
\begin{align*}
    \sum_{i=1}^{{\tilde{n}}} \mathbb{E}(|\xi_{{\tilde{n}},i}|^3) &= \sum_{i=1}^{{\tilde{n}}} \mathbb{E}\left( \mathbb{E}\left( \left| \frac{1}{\sqrt{{\tilde{n}}}} c^\top A_1 \left( \frac{\tilde{Z}^\top \tilde{Z}}{{\tilde{n}}}\right)^{-1} Z_i \eta_i \right|^3 \mid Z_{1:{\tilde{n}}} \right) \right) \\
    &= \frac{1}{{\tilde{n}}^{3/2}} \sum_{i=1}^{\tilde{n}} \mathbb{E}\left( \left| c^\top A_1 \left( \frac{\tilde{Z}^\top \tilde{Z}}{{\tilde{n}}}\right)^{-1} Z_i \right|^3 \mathbb{E}(|\eta_i|^3 \mid Z_i)\right) \\
    &\leq \frac{\nu_4}{{\tilde{n}}^{1/2}} \mathbb{E}\left( \left|\left| \left( \frac{\tilde{Z}^\top \tilde{Z}}{{\tilde{n}}}\right)^{-1} Z_1 \right|\right|^3 \right) \\
    &\leq \frac{\nu_5}{{\tilde{n}}^{1/2}} \mathbb{E}\left( \left|\left| \left( \frac{\tilde{Z}^\top \tilde{Z}}{{\tilde{n}}}\right)^{-1} \right|\right|_{\textrm{op}}^3 ||Z_1||^3 \right) \\
    &= \frac{\nu_5}{{\tilde{n}}^{1/2}} \mathbb{E}\left( \left( \frac{1}{\lambda_{\min}\left( \frac{\tilde{Z}^\top \tilde{Z}}{{\tilde{n}}}\right)} \right)^3 ||Z_1||^3 \right) \\
    &\leq \frac{\nu_6}{{\tilde{n}}^{1/2}} \sqrt{\mathbb{E} \left( \left( \frac{1}{\lambda_{\min}\left( \frac{\tilde{Z}^\top \tilde{Z}}{{\tilde{n}}}\right)} \right)^6 \right)} \sqrt{\mathbb{E}\left( ||Z_1^3||^2 \right)} \\
    &\rightarrow 0
\end{align*}
under Assumptions \ref{asm:momentofxz}, \ref{asm:errormeanvar},  \ref{asm:eigenvalue} and \ref{asm:fbounded}. Here, $\nu_4, \nu_5, \nu_6$ are positive constant. Moreover,
\begin{align*}
    \sum_{i={\tilde{n}}+1}^{2{\tilde{n}}} \mathbb{E}(|\xi_{{\tilde{n}},i}|^3) &= \sum_{i={\tilde{n}}+1}^{2{\tilde{n}}} \mathbb{E}\left( \mathbb{E}\left( \left| \frac{1}{\sqrt{{\tilde{n}}}}  c^\top \epsilon_i a_{{\hat{\delta}_{\tilde{n}}},i}\right|^3 \mid Z_{1:2{\tilde{n}}}, \eta_{1:2{\tilde{n}}}, X_{1:2{\tilde{n}}} \right) \right) \\
    &= \sum_{i={\tilde{n}}+1}^{2{\tilde{n}}} \frac{1}{{\tilde{n}}^{3/2}} \mathbb{E}\left( \left| c^\top a_{{\hat{\delta}_{\tilde{n}}},i} \right|^3 \mathbb{E}(|\epsilon_i|^3 \mid X_i, \eta_i)\right) \\
    &\leq \frac{\nu_7}{{\tilde{n}}^{1/2}} \mathbb{E}\left( || u_{{\hat{\delta}_{\tilde{n}}}}^{3/2}||^2 \mid Q < \tau_0 \right) \\
    &\rightarrow 0
\end{align*}
using Assumptions \ref{asm:momentofxz} and \ref{asm:errormeanvar}, as well as the inequalities $|a-b|^3 \leq (|a| + |b|)^3 \leq 4 (|a|^3 + |b|^3)$. Here, $\nu_7$ is a positive constant. The proof is complete since we have verified the Lyapunov's condition.

\section{Proof of Theorem \ref{thm:main}}
\label{app:proof-thm}
Let $t_i = \mathds{1}_{Q_i \geq \tau_0}$ denote the treatment status of observation $i$, ${\tilde{n}}_1 = \sum_{i=2{\tilde{n}}+1}^{3{\tilde{n}}} t_i$ denote the number of observations in $I_3$ belonging to the treatment group, and ${\tilde{n}}_0 = {\tilde{n}} - {\tilde{n}}_1$ denote the number of observations in $I_3$ belonging to the control group. Moreover, for each $i \in I_3$, let $c(i) \in I_3$ be the index of the closest $\hat{\eta}_j$ from $\hat{\eta}_i$ such that $t_j = 1 - t_i$. Intuitively speaking, $c(i)$ represents the index of the closest observation (w.r.t. $\hat{\eta}$) from $i$ that belongs to the opposite side of the treatment group. 

Recall that Algorithm \ref{alg:est} involves matching each observation in the treatment group of $I_3$ to an observation in the control group of $I_3$ based on their $\hat{\eta}$ values. Following the notation of \citet{abadie2016matching}, we denote by $K_{{\hat{\delta}_{\tilde{n}}},i}$ the number of times observation $i$ is used as a match. Here, $\hat{\delta}_{\tilde{n}} = \hat{\gamma}_{\tilde{n}} - \gamma_0$ that is obtained from $I_1$. It is easy to see that
\begin{equation*}
    \sqrt{{\tilde{n}}}(\hat{\theta}_{\tilde{n}} - \theta_0) = (1) + (2) + (3) + (4),
\end{equation*}
where 
\begin{equation*}
    \begin{split}
        &(1) = \frac{\sqrt{{\tilde{n}}}}{{\tilde{n}}_1} \sum_{i=2{\tilde{n}}+1}^{3{\tilde{n}}} t_i \left( \alpha_0(X_i, \eta_i) - \mathbb{E}(\alpha_0(X, \eta) \mid Q \geq \tau_0)\right),\\
        &(2) = \frac{\sqrt{{\tilde{n}}}}{{\tilde{n}}_1} (\beta_0 - \hat{\beta}_{\tilde{n}})^\top \sum_{i=2{\tilde{n}}+1}^{3{\tilde{n}}} \left( t_i - (1-t_i) K_{{\hat{\delta}_{\tilde{n}}},i} \right) X_i,\\
        &(3) = \frac{\sqrt{{\tilde{n}}}}{{\tilde{n}}_1} \sum_{i=2{\tilde{n}}+1}^{3{\tilde{n}}} \left( t_i - (1-t_i) K_{{\hat{\delta}_{\tilde{n}}},i} \right) \epsilon_i,\\
        &(4) = \frac{\sqrt{{\tilde{n}}}}{{\tilde{n}}_1} \sum_{i=2{\tilde{n}}+1}^{3{\tilde{n}}} t_i \left( \ell(\eta_i) - \ell(\eta_{c(i)})\right).
    \end{split}
\end{equation*}
We begin by looking at (2). Note that we have
\begin{equation*}
    \begin{split}
        (2) &= \frac{\sqrt{{\tilde{n}}}}{{\tilde{n}}_1} (\beta_0 - \hat{\beta}_{\tilde{n}})^\top \sum_{i=2{\tilde{n}}+1}^{3{\tilde{n}}} \left( t_i - (1-t_i)K_{{\hat{\delta}_{\tilde{n}}},i} \right) X_i \\
        &= \left( \sqrt{{\tilde{n}}} (\beta_0 - \hat{\beta}_{\tilde{n}})\right)^\top \left( \frac{{\tilde{n}}}{{\tilde{n}}_1} \right) \left( \frac{1}{{\tilde{n}}} \sum_{i=2{\tilde{n}}+1}^{3{\tilde{n}}} t_i X_i - \frac{1}{{\tilde{n}}} \sum_{i=2{\tilde{n}}+1}^{3{\tilde{n}}} (1-t_i) K_{{\hat{\delta}_{\tilde{n}}},i} X_i \right).
    \end{split}
\end{equation*}
The first term in the product converges in distribution to $\mathcal{N}(0, \Sigma_\beta)$ established in Proposition \ref{prop:beta}, while the second term converges in probability to $1/\bbP(Q \geq \tau_0)$. Moreover, the first part of the third term converges in probability to $\mathbb{E}(tX) = \bbP(Q \geq \tau_0) \mathbb{E}(X \mid Q \geq \tau_0)$. The remainder term can be written as
\begin{equation*}
    \frac{1}{{\tilde{n}}} \sum_{i=2{\tilde{n}}+1}^{3{\tilde{n}}} (1-t_i) K_{{\hat{\delta}_{\tilde{n}}},i} X_i = \left(\frac{{\tilde{n}}_0}{{\tilde{n}}}\right) \left( \frac{1}{{\tilde{n}}_0} \sum_{\substack{i: t_i = 0 \\ 2{\tilde{n}}+1}}^{3{\tilde{n}}} K_{{\hat{\delta}_{\tilde{n}}},i} X_i \right).
\end{equation*}

As in the proof of Proposition \ref{prop:beta}, we first fix ${\hat{\delta}_{\tilde{n}}}$. Under Assumptions \ref{asm:momentofxz} and \ref{asm:abadie}, a slight modification to the proofs of Lemmas S.6, S.7 and S.10 of \citet{abadie2016matching} yields the following result, whose proof is omitted.
\begin{lemma}
\label{lemma:6710}
    Reorder the data in $I_3$ such that observations in the control group are first. Let $({\tilde{n}}_1, P_{0,2{\tilde{n}}+1}, \\\cdots, P_{0,2{\tilde{n}}+{\tilde{n}}_0})$ be the parameter of the distribution of $(K_{{\hat{\delta}_{\tilde{n}}},2{\tilde{n}}+1}, \cdots, K_{{\hat{\delta}_{\tilde{n}}},2{\tilde{n}}+{\tilde{n}}_0})$ given $t_{2{\tilde{n}}+1:3{\tilde{n}}}$ and $\hat{\eta}_{2{\tilde{n}}+1:2{\tilde{n}}+{\tilde{n}}_0}$. For $i \in \{2{\tilde{n}}+1, \cdots, 2{\tilde{n}}+{\tilde{n}}_0\}$, let $\hat{\eta}_{(i)}$'s be the order statistics for the $\hat{\eta}_i$'s. 
    % Define $\hat{\eta}_{(2N)} = a$ and $\hat{\eta}_{(2N + N_0 + 1)} = b$ where $[a,b]$ is the support of $\hat{\eta}$ when $Q < \tau$. 
    Also, let $f_{i,{\hat{\delta}_{\tilde{n}}}}(\hat{\eta})$ ($F_{i,{\hat{\delta}_{\tilde{n}}}}(\hat{\eta})$) be the density (distribution function) of $\hat{\eta}:= \eta - Z^\top {\hat{\delta}_{\tilde{n}}}$ conditional on $t = i$, for $i \in \{0,1\}$. We have
    \begin{equation*}
        \begin{split}
            &\frac{1}{{\tilde{n}}_0} \sum_{i=2{\tilde{n}}+1}^{2{\tilde{n}}+{\tilde{n}}_0} X_i (K_{{\hat{\delta}_{\tilde{n}}},i} - {\tilde{n}}_1 P_{0,i}) = o_p(1),\\
            &\sum_{i=2{\tilde{n}}+1}^{2{\tilde{n}}+{\tilde{n}}_0} X_i \left( P_{0,i} - \frac{f_{1,{\hat{\delta}_{\tilde{n}}}}(\hat{\eta}_i)}{f_{0,{\hat{\delta}_{\tilde{n}}}}(\hat{\eta}_i)} \frac{F_{0,{\hat{\delta}_{\tilde{n}}}}(\hat{\eta}_{(i+1)}) - F_{0,{\hat{\delta}_{\tilde{n}}}}(\hat{\eta}_{(i-1)}) }{2}\right)  = o_p(1),\\
            &\sum_{i=2{\tilde{n}}+1}^{2{\tilde{n}}+{\tilde{n}}_0} X_i \frac{f_{1,{\hat{\delta}_{\tilde{n}}}}(\hat{\eta}_i)}{f_{0,{\hat{\delta}_{\tilde{n}}}}(\hat{\eta}_i)} \frac{F_{0,{\hat{\delta}_{\tilde{n}}}}(\hat{\eta}_{(i+1)}) - F_{0,{\hat{\delta}_{\tilde{n}}}}(\hat{\eta}_{(i-1)})}{2} -\mathbb{E}\left( \frac{f_{1,{\hat{\delta}_{\tilde{n}}}}(\hat{\eta})}{f_{0,{\hat{\delta}_{\tilde{n}}}}(\hat{\eta})} X \mid Q < \tau_0 \right) = o_p(1).
        \end{split}
    \end{equation*}
\end{lemma}
From Lemma \ref{lemma:6710}, conditional on $\hat{\delta}_{\tilde{n}}$, we have
\begin{equation*}
    \frac{1}{{\tilde{n}}} \sum_{i=2{\tilde{n}}+1}^{3{\tilde{n}}} (1-t_i) K_{{{\hat{\delta}_{\tilde{n}}}},i} X_i - \bbP(Q \geq \tau_0) \mathbb{E}\left( \frac{f_{1,{{\hat{\delta}_{\tilde{n}}}}}(\hat{\eta})}{f_{0,{{\hat{\delta}_{\tilde{n}}}}}(\hat{\eta})} X \mid Q < \tau_0 \right) = o_p(1),
\end{equation*}
where $\hat{\eta} = \eta - Z^\top {{\hat{\delta}_{\tilde{n}}}}$. We now prove the following lemma:
\vspace{2mm}
\begin{lemma}
\label{lem:cont3}
For any sequence $\{\hat{\delta}_{\tilde{n}}\}_{\tilde{n} \geq 1}$, where $||\hat \delta_{\tilde{n}}||_2 \leq 1$ for every $\tilde{n}$, that converges to $0$, we have 
$$\mathbb{E}\left( \frac{f_{1,{{\hat{\delta}_{\tilde{n}}}}}(\hat{\eta})}{f_{0,{{\hat{\delta}_{\tilde{n}}}}}(\hat{\eta})} X \mid Q < \tau_0 \right) \rightarrow \mathbb{E}\left( \frac{f_{1,{0}}({\eta})}{f_{0,{0}}({\eta})} X \mid Q < \tau_0 \right).$$
\end{lemma}
\begin{proof}
Note that for any $\tilde{n}$, we have
\begin{align*}
&\mathbb{E}\left( \frac{f_{1,{{\hat{\delta}_{\tilde{n}}}}}(\hat{\eta})}{f_{0,{{\hat{\delta}_{\tilde{n}}}}}(\hat{\eta})} X \mid Q < \tau_0 \right) \\
&= \mathbb{E}\left( \frac{f_{1,{{\hat{\delta}_{\tilde{n}}}}}(\hat{\eta})}{f_{0,{{\hat{\delta}_{\tilde{n}}}}}(\hat{\eta})} \mathbb{E}(X \mid \hat{\eta}, Q < \tau_0) \mid Q < \tau_0 \right) \\
&= \int_C \frac{f_{1,{{\hat{\delta}_{\tilde{n}}}}}(t)}{f_{0,{{\hat{\delta}_{\tilde{n}}}}}(t)} \mathbb{E}(X \mid \eta - Z^\top \hat{\delta}_{\tilde{n}}= t ) f_{\eta - Z^\top \hat{\delta}_{\tilde{n}} \mid Q < \tau_0}(t) \mathds{1}_{t \in \textrm{supp}(\eta - Z^\top {\hat{\delta}_{\tilde{n}}})} dt
\end{align*}
and
\begin{align*}
&\mathbb{E}\left( \frac{f_{1,{0}}({\eta})}{f_{0,{0}}({\eta})} X \mid Q < \tau_0 \right) \\
&= \mathbb{E}\left( \frac{f_{1,{0}}({\eta})}{f_{0,{0}}({\eta})} \mathbb{E}(X \mid \eta, Q < \tau_0) \mid Q < \tau_0 \right) \\
&= \int_C \frac{f_{1,{{0}}}(t)}{f_{0,{{0}}}(t)} \mathbb{E}(X \mid \eta = t) f_{\eta \mid Q < \tau_0}(t) \mathds{1}_{t \in \textrm{supp}(\eta)} dt.
\end{align*}
Using the same method as in Lemma \ref{lemma:cont}, the conclusion follows via Lebesgue's DCT under Assumptions 1, 2, 3 and 6 using the fact that $C$ is compact. Here, we define $0/0 = 0$.
\end{proof}
From here, an argument similar to Lemma \ref{lemma:uncond} yields
\begin{equation*}
    \frac{1}{{\tilde{n}}} \sum_{i=2{\tilde{n}}+1}^{3{\tilde{n}}} (1-t_i) K_{{\hat{\delta}_{\tilde{n}}},i} X_i \overset{P}{\longrightarrow} \bbP(Q \geq \tau_0) \mathbb{E}\left( \frac{f_{1}({\eta})}{f_{0}({\eta})} X \mid Q < \tau_0 \right),
\end{equation*}
where $f_0(\eta) := f_{0,0}(\eta)$ and $f_1(\eta) := f_{1,0}(\eta)$. Thus, up to $o_p(1)$ terms, we can write $(2) = \left( \sqrt{{\tilde{n}}} (\hat{\beta}_{\tilde{n}} - \beta_0)\right)^\top A_2,$ where 
\begin{equation*}
A_2 = \mathbb{E}\left( \frac{f_1({\eta})}{f_0({\eta})} X \mid Q < \tau_0 \right) - \mathbb{E}\left( X \mid Q \geq \tau_0 \right).
\end{equation*}
Following the derivations in Proposition \ref{prop:beta}, we can again rewrite (2) as
\begin{equation*}
    (2) = \sum_{i=1}^{\tilde{n}} \frac{1}{\sqrt{{\tilde{n}}}} A_3^\top A_1 \left( \frac{\tilde{Z}^\top \tilde{Z}}{{\tilde{n}}}\right)^{-1} Z_i \eta_i + \sum_{i={\tilde{n}}+1}^{2{\tilde{n}}} \frac{1}{\sqrt{{\tilde{n}}}} A_3^\top \epsilon_i a_{{\hat{\delta}_{\tilde{n}}},i}
\end{equation*}
up to $o_p(1)$ terms, where
\begin{equation*}
    A_3 = \frac{1}{2\bbP(Q < \tau_0)} \Sigma_u^{-1} A_2.
\end{equation*}

We now consider (4), which can be decomposed into
\begin{equation*}
    \begin{split}
        (4) = \frac{\sqrt{{\tilde{n}}}}{{\tilde{n}}_1} \sum_{i=2{\tilde{n}}+1}^{3{\tilde{n}}} t_i \left( \ell(\eta_i) - \ell(\eta_{c(i)})\right)
        = (5) + (6) - (7),
    \end{split}
\end{equation*}        
where
\begin{align*}
    (5) &= \frac{\sqrt{{\tilde{n}}}}{{\tilde{n}}_1} \sum_{i=2{\tilde{n}}+1}^{3{\tilde{n}}} t_i \left( \ell(\hat{\eta}_i) - \ell(\hat{\eta}_{c(i)})\right), \\
    (6) &= \frac{\sqrt{{\tilde{n}}}}{{\tilde{n}}_1} \sum_{i=2{\tilde{n}}+1}^{3{\tilde{n}}} t_i \left( \ell(\eta_i) - \ell(\hat{\eta}_{i})\right), \\
    (7) &= \frac{\sqrt{{\tilde{n}}}}{{\tilde{n}}_1} \sum_{i=2{\tilde{n}}+1}^{3{\tilde{n}}} t_i \left( \ell(\eta_{c(i)}) - \ell(\hat{\eta}_{c(i)})\right).
\end{align*}     

As usual, we condition on ${\hat{\delta}_{\tilde{n}}}$ (i.e., $I_1$). Following the proof of the first part of Proposition 1 in \citet{abadie2016matching}, we can show that (5) is $o_p(1)$ under Assumption \ref{asm:abadie}. Now, observe that (6) and (7) can be respectively written as
\begin{equation*}
    (6) = \frac{\sqrt{{\tilde{n}}}}{{\tilde{n}}_1} \sum_{i=2{\tilde{n}}+1}^{3{\tilde{n}}} t_i \left( \eta_i - \hat{\eta}_{i}\right)\ell'(\hat{\eta}_i) + \frac{\sqrt{{\tilde{n}}}}{{\tilde{n}}_1} \sum_{i=2{\tilde{n}}+1}^{3{\tilde{n}}} t_i \left( \eta_i - \hat{\eta}_{i}\right)^2 \ell''(\tilde{\eta}_i).
\end{equation*}
and
\begin{equation*}
    (7) = \frac{\sqrt{{\tilde{n}}}}{{\tilde{n}}_1} \sum_{i=2{\tilde{n}}+1}^{3{\tilde{n}}} t_i \left( \eta_{c(i)} - \hat{\eta}_{c(i)}\right)\ell'(\hat{\eta}_{c(i)}) + \frac{\sqrt{{\tilde{n}}}}{{\tilde{n}}_1} \sum_{i=2{\tilde{n}}+1}^{3{\tilde{n}}} t_i \left( \eta_{c(i)} - \hat{\eta}_{c(i)}\right)^2 \ell''(\tilde{\eta}_{c(i)}),
\end{equation*}
where $\tilde{\eta}_i$ lies between $\eta_i$ and $\hat{\eta}_i$ and $\tilde{\eta}_{c(i)}$ lies between $\eta_{c(i)}$ and $\hat{\eta}_{c(i)}$. It is easy to see that the second summands of (6) and (7) are $o_p(1)$ under Assumptions \ref{asm:momentofxz} and \ref{asm:fbounded}. Therefore, omitting $o_p(1)$ terms, we can rewrite (4) as
\begin{equation*}
    \begin{split}
        (4) &= \frac{\sqrt{{\tilde{n}}}}{{\tilde{n}}_1} \sum_{i=2{\tilde{n}}+1}^{3{\tilde{n}}} (t_i - (1 - t_i) K_{{\hat{\delta}_{\tilde{n}}}, i}) (\eta_i - \hat{\eta}_i) \ell'(\hat{\eta}_i) \\
        &= (\sqrt{{\tilde{n}}}(\hat{\gamma}_{\tilde{n}} - \gamma_0))^\top \left( \frac{{\tilde{n}}}{{\tilde{n}}_1}\right) \left( \frac{1}{{\tilde{n}}} \sum_{i=2{\tilde{n}}+1}^{3{\tilde{n}}} (t_i - (1 - t_i) K_{{\hat{\delta}_{\tilde{n}}}, i}) Z_i \ell'(\hat{\eta}_i)\right).
    \end{split}
\end{equation*}
Under Assumptions \ref{asm:momentofxz}, \ref{asm:fbounded} and  \ref{asm:abadie}, conditional on $\hat{\delta}_{\tilde{n}}$, we have
\begin{align*}
   &\left( \frac{{\tilde{n}}}{{\tilde{n}}_1}\right) \left( \frac{1}{{\tilde{n}}} \sum_{i=2{\tilde{n}}+1}^{3{\tilde{n}}} (t_i - (1 - t_i) K_{{\hat{\delta}_{\tilde{n}}}, i}) Z_i \ell'(\hat{\eta}_i)\right) \\ &- \mathbb{E}\left(Z\ell'({\hat{\eta}}) \mid Q \geq \tau_0 \right) - \mathbb{E}\left( \frac{f_{1,{\hat{\delta}_{\tilde{n}}}}({\hat{\eta}})}{f_{0,{\hat{\delta}_{\tilde{n}}}}({\hat{\eta}})} Z \ell'({\hat{\eta}}) \mid Q < \tau_0 \right) = o_p(1),
\end{align*}
where $\hat{\eta} = {\eta} - Z^\top {\hat{\delta}_{\tilde{n}}}$, using a similar result to Lemma \ref{lemma:6710}. An argument similar to Lemmas \ref{lemma:cont} and \ref{lemma:uncond} yields
\begin{align*}
   \left( \frac{{\tilde{n}}}{{\tilde{n}}_1}\right) \left( \frac{1}{{\tilde{n}}} \sum_{i=2{\tilde{n}}+1}^{3{\tilde{n}}} (t_i - (1 - t_i) K_{{\hat{\delta}_{\tilde{n}}}, i}) Z_i \ell'(\hat{\eta}_i)\right) \overset{P}{\longrightarrow} \\ \mathbb{E}\left(Z\ell'({{\eta}}) \mid Q \geq \tau_0 \right) - \mathbb{E}\left( \frac{f_1({{\eta}})}{f_0({{\eta}})} Z \ell'({{\eta}}) \mid Q < \tau_0 \right)
\end{align*}
under Assumptions \ref{asm:momentofxz}, \ref{asm:errormeanvar}, \ref{asm:smoothness_conditional}, \ref{asm:fbounded} and \ref{asm:abadie}. 
Up to $o_p(1)$ terms, we can thus rewrite $(4)$ as
\begin{equation*}
    (4) = \sum_{i=1}^{\tilde{n}} \frac{1}{\sqrt{{\tilde{n}}}} A_4^\top  \left( \frac{\tilde{Z}^\top \tilde{Z}}{{\tilde{n}}}\right)^{-1} Z_i \eta_i,
\end{equation*}
where
\begin{equation*}
    A_4 = \mathbb{E}\left(Z\ell'({\eta}) \mid Q \geq \tau_0 \right) - \mathbb{E}\left( \frac{f_1({\eta})}{f_0({\eta})} Z \ell'({\eta}) \mid Q < \tau_0 \right).
\end{equation*}

Now, we are ready to establish the asymptotic normality of $\hat{\theta}_{\tilde{n}}$. Ignoring $o_p(1)$ terms, we have
\begin{equation*}
    \begin{split}
        &\sqrt{{\tilde{n}}}(\hat{\theta}_{\tilde{n}} - \theta_0) \\
        &= \sum_{i=1}^{\tilde{n}} \frac{1}{\sqrt{{\tilde{n}}}} A_5^\top \left( \frac{\tilde{Z}^\top \tilde{Z}}{{\tilde{n}}}\right)^{-1} Z_i \eta_i + \sum_{i={\tilde{n}}+1}^{2{\tilde{n}}}  \frac{1}{\sqrt{{\tilde{n}}}} A_3^\top \epsilon_i a_{{\hat{\delta}_{\tilde{n}}},i}\\  
        &+  \sum_{i=2{\tilde{n}}+1}^{3{\tilde{n}}} \frac{\sqrt{{\tilde{n}}}}{{\tilde{n}}_1} t_i \left( \alpha_0(X_i, \eta_i) - \mathbb{E}(\alpha_0(X, \eta) \mid Q \geq \tau_0)\right)
        + \sum_{i=2{\tilde{n}}+1}^{3{\tilde{n}}} \frac{\sqrt{{\tilde{n}}}}{{\tilde{n}}_1} \left( t_i - (1-t_i) K_{{\hat{\delta}_{\tilde{n}}},i} \right) \epsilon_i \\
    &= \xi_{{\tilde{n}},1} + \cdots + \xi_{{\tilde{n}},{\tilde{n}}} + \xi_{{\tilde{n}},{\tilde{n}}+1} + \cdots + \xi_{{\tilde{n}},2{\tilde{n}}} + \xi_{{\tilde{n}},2{\tilde{n}}+1} + \cdots + \xi_{{\tilde{n}},3{\tilde{n}}} +
    \xi_{{\tilde{n}},3{\tilde{n}}+1} + \cdots + \xi_{{\tilde{n}},4{\tilde{n}}},
    \end{split}
\end{equation*}
where $A_5 = A_1^\top A_3 + A_4$. 

Consider the following $\sigma$-fields: $\mathcal{F}_{{\tilde{n}},1} = \sigma(Z_{1:{\tilde{n}}}, \eta_1)$,$\cdots, \mathcal{F}_{{\tilde{n}},{\tilde{n}}} = \sigma(Z_{1:{\tilde{n}}}, \eta_{1:{\tilde{n}}})$,\\ $\mathcal{F}_{{\tilde{n}},{\tilde{n}}+1} = \sigma(Z_{1:2{\tilde{n}}}, \eta_{1:2{\tilde{n}}}, X_{1:2{\tilde{n}}},\epsilon_{{\tilde{n}}+1}), \cdots$, $\mathcal{F}_{{\tilde{n}},2{\tilde{n}}} = \sigma(Z_{1:2{\tilde{n}}}, \eta_{1:2{\tilde{n}}}, X_{1:2{\tilde{n}}},\epsilon_{{\tilde{n}}+1:2{\tilde{n}}})$,\\$\mathcal{F}_{{\tilde{n}},2{\tilde{n}}+1} = \sigma(Z_{1:2{\tilde{n}}}, \eta_{1:2{\tilde{n}}+1}, X_{1:2{\tilde{n}}+1},\epsilon_{{\tilde{n}}+1:2{\tilde{n}}}),\cdots,\mathcal{F}_{{\tilde{n}},3{\tilde{n}}} = \sigma(Z_{1:2{\tilde{n}}}, \eta_{1:3{\tilde{n}}}, X_{1:3{\tilde{n}}},\epsilon_{{\tilde{n}}+1:2{\tilde{n}}})$,\\
$\mathcal{F}_{{\tilde{n}},3{\tilde{n}}+1} = \sigma(Z_{1:3{\tilde{n}}}, \eta_{1:3{\tilde{n}}}, X_{1:3{\tilde{n}}},\epsilon_{{\tilde{n}}+1:2{\tilde{n}}+1}),\cdots,\mathcal{F}_{{\tilde{n}},4{\tilde{n}}} = \sigma(Z_{1:3{\tilde{n}}}, \eta_{1:3{\tilde{n}}}, X_{1:3{\tilde{n}}},\epsilon_{{\tilde{n}}+1:3{\tilde{n}}})$. For each ${\tilde{n}}$, it is easy to see that
\begin{equation*}
    \left\{ \sum_{j=1}^i \xi_{{\tilde{n}},j}, \mathcal{F}_{{\tilde{n}},i}, 1 \leq i \leq 4{\tilde{n}} \right\}
\end{equation*}
is a martingale. We now use \citeauthor{billingsley1995probability}'s (\citeyear{billingsley1995probability}) martingale central limit theorem. Note that using Assumption \ref{asm:errormeanvar}, we have
\begin{equation*}
    \begin{split}
        \sum_{i=1}^{\tilde{n}} \mathbb{E}(\xi_{{\tilde{n}},i}^2 \mid \mathcal{F}_{{\tilde{n}},i-1}) &=  \sum_{i=1}^{\tilde{n}} \mathbb{E}\left(\left( \frac{1}{\sqrt{{\tilde{n}}}} A_5^\top \left( \frac{\tilde{Z}^\top \tilde{Z}}{{\tilde{n}}}\right)^{-1} Z_i \eta_i \right)^2 \mid Z_{1:{\tilde{n}}}, \eta_{1:i-1}\right) \\
        &= \sigma_\eta^2 A_5^\top \left(\frac{\tilde{Z}^\top \tilde{Z}}{{\tilde{n}}}\right)^{-1} A_5 \\
        &\overset{P}{\longrightarrow} A_5^\top \Sigma_\gamma A_5
    \end{split}
\end{equation*}
and
\begin{equation*}
    \begin{split}
         \sum_{i={\tilde{n}}+1}^{2{\tilde{n}}} \mathbb{E}(\xi_{{\tilde{n}},i}^2 \mid \mathcal{F}_{{\tilde{n}},i-1}) &=   \sum_{i={\tilde{n}}+1}^{2{\tilde{n}}} \mathbb{E}\left( \left( \frac{1}{\sqrt{{\tilde{n}}}}  A_3^\top \epsilon_i a_{{\hat{\delta}_{\tilde{n}}},i}\right)^2 \mid Z_{1:2{\tilde{n}}}, \eta_{1:2{\tilde{n}}}, X_{1:2{\tilde{n}}}, \epsilon_{{\tilde{n}}+1:i-1} \right) \\
        &= \sigma_\epsilon^2 \sum_{i={\tilde{n}}+1}^{2{\tilde{n}}} \frac{1}{{\tilde{n}}} (A_3^\top a_{{{\hat{\delta}_{\tilde{n}}}},i})^2.
    \end{split}
\end{equation*}
Conditional on ${{\hat{\delta}_{\tilde{n}}}}$ (i.e., $I_1$), it is easy to show that
\begin{equation*}
    \frac{1}{{\tilde{n}}} \sum_{i={\tilde{n}}+1}^{2{\tilde{n}}} (A_3^\top a_{{{\hat{\delta}_{\tilde{n}}}},i})^2  - 6 \bbP(Q < \tau_0) A_3^\top \Sigma_{u, {{\hat{\delta}_{\tilde{n}}}}} A_3 = o_p(1)
\end{equation*}
under Assumptions \ref{asm:momentofxz} using the same method as for the term $P$ in the proof of Proposition \ref{prop:beta}. Following the proof of Lemma \ref{lemma:uncond}, we have 
\begin{equation*}
    \frac{1}{{\tilde{n}}} \sum_{i={\tilde{n}}+1}^{2{\tilde{n}}} (A_3^\top a_{{\hat{\delta}_{\tilde{n}}},i})^2  \overset{P}{\longrightarrow} 6 \bbP(Q < \tau_0) A_3^\top \Sigma_{u} A_3,
\end{equation*}
whence
\begin{equation*}
    \begin{split}
         \sum_{i={\tilde{n}}+1}^{2{\tilde{n}}} \mathbb{E}(\xi_{{\tilde{n}},i}^2 \mid \mathcal{F}_{{\tilde{n}},i-1}) \overset{P}{\longrightarrow} 6 \bbP(Q < \tau_0) \sigma_\epsilon^2 A_3^\top \Sigma_{u} A_3.
    \end{split}
\end{equation*}
Moreover, we have
\begin{equation*}
    \begin{split}
        &\sum_{i=2{\tilde{n}}+1}^{3{\tilde{n}}} \mathbb{E}(\xi_{{\tilde{n}},i}^2 \mid \mathcal{F}_{{\tilde{n}},i-1}) \\ &= \sum_{i=2{\tilde{n}}+1}^{3{\tilde{n}}} \mathbb{E}\left(\left(\frac{\sqrt{{\tilde{n}}}}{{\tilde{n}}_1} t_i \left( \alpha_0(X_i, \eta_i) - \mathbb{E}(\alpha_0(X, \eta) \mid Q \geq \tau_0)\right)\right)^2 \mid Z_{1:2{\tilde{n}}}, \eta_{1:i-1}, X_{1:i-1}, \epsilon_{{\tilde{n}}+1:2{\tilde{n}}}\right) \\
        &= \left(\frac{{\tilde{n}}}{{\tilde{n}}_1} \right)^2 \textrm{var}(\alpha_0(X, \eta) \mid Q \geq \tau_0) \bbP(Q \geq \tau_0) \ \overset{P}{\longrightarrow} \ \frac{\textrm{var}(\alpha(X, \eta) \mid Q \geq \tau_0)}{\bbP(Q \geq \tau_0)} \,.
    \end{split}
\end{equation*}
and
\begin{equation*}
    \begin{split}
        &\sum_{i=3{\tilde{n}}+1}^{4{\tilde{n}}} \mathbb{E}(\xi_{{\tilde{n}},i}^2 \mid \mathcal{F}_{{\tilde{n}},i-1}) \\
        &= \sum_{i=2{\tilde{n}}+1}^{3{\tilde{n}}} \mathbb{E}\left( \left( \frac{\sqrt{{\tilde{n}}}}{{\tilde{n}}_1} \left( t_i - (1-t_i) K_{{\hat{\delta}_{\tilde{n}}},i} \right) \epsilon_i\right)^2 \mid Z_{1:3{\tilde{n}}}, \eta_{1:3{\tilde{n}}}, X_{1:3{\tilde{n}}}, \epsilon_{{\tilde{n}}+1:i-1}  \right) \\
    &= \left( \frac{{\tilde{n}}}{{\tilde{n}}_1} \right)^2 \sigma_\epsilon^2 \left( \frac{1}{{\tilde{n}}} \sum_{i=2{\tilde{n}}+1}^{3{\tilde{n}}} \left( t_i - (1-t_i) K_{{\hat{\delta}_{\tilde{n}}},i} \right)^2 \right) \\
    &\overset{P}{\longrightarrow} \sigma_\epsilon^2 \left( \frac{2}{\bbP(Q \geq \tau_0)} + \frac{3}{2 \bbP(Q < \tau_0)}  \mathbb{E}\left( \left(\frac{f_1(\eta)}{f_0(\eta)} \right)^2 \mid Q < \tau_0 \right) \right).
    \end{split}
\end{equation*}
In order to derive the last line, note that
\begin{equation*}
    \begin{split}
       \frac{1}{{\tilde{n}}} \sum_{i=2{\tilde{n}}+1}^{3{\tilde{n}}} \left( t_i - (1-t_i)K_{{\hat{\delta}_{\tilde{n}}}, i} \right)^2 &= \frac{1}{{\tilde{n}}} \sum_{i=2{\tilde{n}}+1}^{\tilde{n}} (t_i^2 + (1-t_i)^2 K_{{\hat{\delta}_{\tilde{n}}}, i}^2) \\
    &= \frac{1}{{\tilde{n}}} \sum_{i=2{\tilde{n}}+1}^{3{\tilde{n}}} (t_i + (1-t_i) K_{{\hat{\delta}_{\tilde{n}}}, i}^2) \\
    &= \frac{1}{{\tilde{n}}} \sum_{i=2{\tilde{n}}+1}^{3{\tilde{n}}}  t_i + \frac{1}{{\tilde{n}}} \sum_{i=2{\tilde{n}}+1}^{3{\tilde{n}}} (1-t_i) K_{{\hat{\delta}_{\tilde{n}}}, i}^2. 
    \end{split}
\end{equation*}
The first term clearly converges in probability to $P(Q \geq \tau)$. Also, according to Lemma S.11 of \citet{abadie2016matching}, conditional on ${\hat{\delta}_{\tilde{n}}}$ (i.e., $I_1$),
\begin{equation*}
    \begin{split}
         \frac{1}{{\tilde{n}}} \sum_{\substack{i: t_i = 0 \\ 2{\tilde{n}}+1}}^{3{\tilde{n}}} K_{{\hat{\delta}_{\tilde{n}}}, i}^2  &= \left(\frac{{\tilde{n}}_0}{{\tilde{n}}}\right) \left( \frac{1}{{\tilde{n}}_0} \sum_{\substack{i: t_i = 0 \\ 2{\tilde{n}}+1}}^{3{\tilde{n}}} K_{{\hat{\delta}_{\tilde{n}}},i}^2 \right) \\
        &= \bbP(Q \geq \tau_0) + \frac{3}{2} \frac{(\bbP(Q \geq \tau_0))^2}{\bbP(Q < \tau_0)} \mathbb{E}\left( \left(\frac{f_{1,{\hat{\delta}_{\tilde{n}}}}(\hat{\eta})}{f_{0,{\hat{\delta}_{\tilde{n}}}}(\hat{\eta})} \right)^2  \mid Q < \tau_0 \right) + o_p(1),
    \end{split}
\end{equation*}
where $\hat{\eta} = \eta - Z^\top {\hat{\delta}_{\tilde{n}}}$, whence the conclusion immediately follows under Assumption \ref{asm:abadie} by an argument similar to Lemmas \ref{lemma:cont} and \ref{lemma:uncond} under Assumptions 1, 2, 3 and 6. 

Therefore, an application of the martingale central limit theorem \citep{billingsley1995probability} gives us
\begin{equation*}
    \sqrt{{\tilde{n}}}(\hat{\theta}_{\tilde{n}} - \theta_0) \overset{d}{\longrightarrow} \mathcal{N}\left(0, \sigma_\theta^2 \right),
\end{equation*}
where
\begin{equation}
    \label{eq:thetavar}
    \sigma^2_\theta = A_5^\top \Sigma_\gamma A_5 + 6\mathbb{P}(Q < \tau_0) \sigma_\epsilon^2 A_3^\top \Sigma_u A_3 + B + C  \,,
\end{equation}
with 
\begin{align*}
    A_3 & = \frac{1}{2\mathbb{P}(Q < \tau_0)} \Sigma_u^{-1} \left( \mathbb{E}\left( \frac{f_1({\eta})}{f_0({\eta})} X \mid Q < \tau_0 \right) - \mathbb{E}\left(X \mid Q \geq \tau_0\right) \right), \\
    A_5 & =  2\mathbb{P}(Q < \tau_0) \mathbb{E}(\ell'(\eta)w_0 u_0^\top \mid Q < \tau_0) A_3 + \mathbb{E}\left(Z\ell'({\eta}) \mid Q \geq \tau_0\right) - \mathbb{E}\left( \frac{f_1({\eta})}{f_0({\eta})} Z \ell'({\eta}) \mid Q < \tau_0 \right), \\
    B & = \frac{\textrm{var}(\alpha_0(X,\eta) \mid Q \geq \tau_0)}{\bbP(Q \geq \tau_0)}, \\
    C & =   \sigma_\epsilon^2 \left( \frac{2}{\bbP(Q \geq \tau_0)} + \frac{3}{2 \bbP(Q < \tau_0)}  \mathbb{E}\left( \left(\frac{f_1(\eta)}{f_0(\eta)} \right)^2 \mid Q < \tau_0 \right) \right)\,.
\end{align*}
% \begin{equation}
% \label{eq:thetavar}
%     \begin{split}
%         &\sigma_\theta^2 = \\
%         & \left( \mathbb{E}\left(Xf'({\eta}) \mid Q \geq \tau_0\right) - \mathbb{E}\left( \frac{f_1({\eta})}{f_0({\eta})} X f'({\eta}) \mid Q < \tau_0 \right) \right)^\top \Sigma_\gamma \left( \mathbb{E}\left(Xf'({\eta}) \mid Q \geq \tau_0\right) - \mathbb{E}\left( \frac{f_1({\eta})}{f_0({\eta})} X f'({\eta}) \mid Q < \tau_0 \right) \right) \\
%         &+\left( \mathbb{E}\left(X \mid Q \geq \tau_0\right) - \mathbb{E}\left( \frac{f_1({\eta})}{f_0({\eta})} X \mid Q < \tau_0 \right) \right)^\top \Sigma_\beta \left( \mathbb{E}\left(X \mid Q \geq \tau_0\right) - \mathbb{E}\left( \frac{f_1({\eta})}{f_0({\eta})} X \mid Q < \tau_0 \right) \right) \\
%         &+ \frac{\textrm{var}(\alpha_0(X,\eta) \mid Q \geq \tau_0)}{\bbP(Q \geq \tau_0)} + \sigma_\epsilon^2 \left( \frac{2}{\bbP(Q \geq \tau_0)} + \frac{3}{2 \bbP(Q < \tau_0)}  \mathbb{E}\left( \left(\frac{f_1(\eta)}{f_0(\eta)} \right)^2 \mid Q < \tau_0 \right) \right).
%     \end{split}
% \end{equation}
To complete the proof, we need to show that the Lindeberg's condition for the martingale central limit theorem is satisfied. Following the proof of Proposition \ref{prop:beta}, we have
$\sum_{i=1}^{{\tilde{n}}} \mathbb{E}(|\xi_{{\tilde{n}},i}|^3) \rightarrow 0$ and $\sum_{i={\tilde{n}}+1}^{2{\tilde{n}}} \mathbb{E}(|\xi_{{\tilde{n}},i}|^3) \rightarrow 0$ as ${\tilde{n}} \rightarrow \infty$ under Assumptions \ref{asm:momentofxz}, \ref{asm:errormeanvar}, \ref{asm:eigenvalue} and \ref{asm:fbounded}, whence the Lyapunov's (and consequently Lindeberg's) condition is satisfied. 

Moreover, we have $\sum_{i=2{\tilde{n}}+1}^{3{\tilde{n}}} \mathbb{E}(|\xi_{{\tilde{n}},i}|^3) \rightarrow 0$ provided $\mathbb{E}\left( \left| \alpha_0(X, \eta) - \mathbb{E}(\alpha_0(X, \eta) \mid Q \geq \tau_0)\right|^3 \right)$ is finite, which follows from Assumption \ref{asm:fbounded}. Lastly, we have $\sum_{i=3{\tilde{n}}+1}^{4{\tilde{n}}} \mathbb{E}(|\xi_{{\tilde{n}},i}|^3) \rightarrow 0$ due to Assumption \ref{asm:errormeanvar} and Lemma S.8 of \citet{abadie2016matching} on the uniform boundedness of the moments of $K_{{\hat{\delta}_{\tilde{n}}}, i}$. This finishes the proof.

\newpage

\end{document}